\documentclass[american,letterpaper]{article}
\usepackage{fullpage}
\usepackage[utf8]{inputenc}
\usepackage{amsmath}
\usepackage{amssymb,amsthm}
\usepackage{natbib}
\usepackage{caption}
\usepackage{setspace}
\usepackage{hyperref}
\usepackage{subcaption}
\usepackage{algorithm}
\usepackage{algorithmicx}
\usepackage{algpseudocode}
\usepackage{footmisc}
\usepackage{multirow}
\usepackage{fancyhdr,graphicx,color}
\usepackage{graphicx}
\usepackage{graphics}
\usepackage{paralist}
\usepackage{booktabs}
\usepackage[skip=2pt]{parskip}
\usepackage[affil-it]{authblk}

\theoremstyle{plain}

\newtheorem{lemma}{Lemma}
\newtheorem{theorem}{Theorem}
\newtheorem{proposition}{Proposition}
\newtheorem{corollary}{Corollary}
\theoremstyle{definition}
\newtheorem{example}{Example}
\newtheorem{asu}{Assumption}
\newtheorem*{asu*}{Assumption}
\newtheorem{remark}{Remark}
\newtheorem{definition}{Definition}

\newcommand{\argmin}{\mathop{\mathrm{argmin}}}
\newcommand{\argmax}{\mathop{\mathrm{argmax}}}

\newcommand{\cA}{\mathcal{A}}
\newcommand{\cG}{\mathcal{G}}
\newcommand{\cN}{\mathcal{N}}
\newcommand{\cI}{\mathcal{I}}
\newcommand{\cJ}{\mathcal{J}}

\newcommand{\cF}{\mathcal{F}}
\newcommand{\bG}{\mathbb{G}}
\newcommand{\R}{\mathbb{R}}
\newcommand{\cS}{\mathcal{S}}
\newcommand{\cQ}{\mathcal{Q}}

\newcommand{\CR}{\mathcal{R}}

\newcommand{\bepsilon}{\boldsymbol \epsilon}
\newcommand{\bgamma}{{\boldsymbol \gamma}}
\newcommand{\bbeta}{\boldsymbol \beta}
\newcommand{\balpha}{\boldsymbol \alpha}

\newcommand{\bY}{{\B{y}}}
\newcommand{\bP}{{\bf P}}
\newcommand{\bX}{{\bf X}}
\newcommand{\bI}{{\bf I}}

\newcommand{\bff}{{\bf f}}

\newcommand{\ourmethodlz}{\texttt{ClusterLearn-L0}}
\newcommand{\ourmethod}{\texttt{ClusterLearn}}

\newcommand{\obj}{\mathsf{Obj}}

\newcommand{\p}{{\rm I}\kern-0.18em{\rm P}}
\newcommand{\E}{{\rm I}\kern-0.18em{\rm E}}

\newcommand{\1}{{\rm 1}\kern-0.24em{\rm I}}
\newcommand{\B}{\boldsymbol}

\author[1]{Kayhan Behdin\thanks{behdink@mit.edu}}
\author[1]{Riade Benbaki}
\author[2]{Peter Radchenko}
\author[1,3]{Rahul Mazumder}

\affil[1]{Operations Research Center, Massachusetts Institute of Technology}
\affil[2]{The University of Sydney}
\affil[3]{Sloan School of Management, Massachusetts Institute of Technology}

\date{}
\begin{document}

\title{ Modeling with Categorical Features via Exact Fusion and Sparsity Regularisation}

\maketitle

\begin{abstract}
    We study the high-dimensional linear regression problem with categorical predictors that have many levels. We propose a new estimation approach, which performs model compression via two mechanisms by simultaneously encouraging (a) clustering of the regression coefficients to collapse some of the categorical levels together; and (b) sparsity of the regression coefficients.  We present novel mixed integer programming formulations for our estimator, and develop a custom row generation procedure to speed up the exact off-the-shelf solvers. We also propose a fast approximate algorithm for our method that obtains high-quality feasible solutions via block coordinate descent.  As the main building block of our algorithm, we develop an exact algorithm for the univariate case based on dynamic programming, which can be of independent interest.  We establish new theoretical guarantees for both the prediction and the cluster recovery performance of our estimator. Our numerical experiments on synthetic and real datasets demonstrate that our proposed estimator tends to outperform the state-of-the-art. 
\end{abstract}

\section{Introduction}

Nominal categorical predictors with many levels often arise in real-world applications. For example, this is the case for geographical data with the ZIP code feature~\citep{kuhn2019feature}, motor insurance claims data with the vehicle brand feature~\citep{hu2018motor}, and electronic health data with the diagnosis code feature~\citep{jensen2012mining}.
In this paper, we focus on the linear regression problem where some or all of the predictors are categorical.
Specifically, we assume there are $q\geq 1$ categorical features: $\B C=(C_1,\cdots,C_q)$ and $N$ continuous ones: $\B W=( W_1,\cdots,W_N)$.  Let~$p_j$ denote the number of levels of the categorical predictor $C_j$, which takes values in the set $[p_j]=\{1,\cdots,p_j\}$. We study linear regression estimators of the form 
\begin{equation}\label{linear-categorical-intro}
  \alpha+  \sum_{j=1}^{q}\sum_{k=1}^{p_j} \theta_{j,k}\1(C_j=k)+\sum_{j=1}^N\theta_j W_j ,
\end{equation}
where $\1(x=y)=0$ if $x\neq y$ and $\1(x=y)=1$ otherwise,  and $\alpha,\{\theta_{j,k}\},\{\theta_j\}$ are the coefficients that we seek to estimate.

From the interpretability and the prediction performance perspectives, it is often desirable to regularise the coefficients of estimator~\eqref{linear-categorical-intro} to obtain more compact models. This is especially important when we have a large number of categorical predictors or when some categorical predictors have a large number of levels.
A natural regularisation in this case is the sparsity regularisation, which ensures that only a few coefficients among $\{\theta_{j,k}\},\{\theta_j\}$ are nonzero, allowing for the reduction of the model dimension. When the data includes categorical predictors with many levels, it is also desirable to cluster the regression coefficients corresponding to the same categorical predictor, so that they only take on a few distinct values. In other words, we are interested in models with a small number of distinct coefficient values\footnote{For example, $|\{\theta_{j,k}:k\in[p_j]\}|$ is the number of distinct coefficient values for $C_j$.} for each categorical predictor. The reduction in the number of coefficient values can be viewed as fusing some of the categorical levels together, effectively reducing the model dimension. 

We propose a new estimator for the high-dimensional linear regression problem with nominal categorical predictors that have a large number of levels. Our approach simultaneously clusters the regression coefficients corresponding to the same categorical predictor and encourages the coefficients to be sparse, allowing for model compression via two mechanisms. In what follows, we review the literature on sparse high-dimensional linear regression and the estimators specialised to the setting with categorical predictors. Then, we summarise our approach and contributions.

\subsection{Related Work}

\textbf{Sparse linear regression.} Given a response vector $\B{y}\in\R^n$ and a model matrix $\B{X}\in\R^{n\times p}$,
we seek to find an estimator~$\B{\beta}\in\R^p$ such that~$\B{X\beta}$ is a good predictor of~$\B{y}$. Numerous estimators have been proposed for this problem~\citep[for example, see][]{hastie2015statistical}. Many of the proposed methods minimise a penalised least squares objective $\|\B{y}-\B{X}\B{\beta}\|_2^2$, using a penalty to improve the quality of the estimator in the high-dimensional settings where~$p$ can be much greater than~$n$, and to improve the interpretability of the solution. 

The $\ell_0$ penalty~\citep{l0-1}, which directly counts the total number of nonzero regression coefficients, is a natural option for inducing sparse solutions. The statistical properties of the corresponding best subset selection-type estimator have been well-studied in the literature~\citep[e.g.,][]{l0-2,l0-3}, but solving the $\ell_0$-regularised optimisation problem in high-dimensional settings had long been considered infeasible. However, due to the recent advances in Mixed Integer Programming (MIP) and discrete optimisation~\citep{nemhauser}, starting with the work of~\citet{l0-modern}, there has been a growing interest in developing efficient algorithms for sparse linear regression problem with $\ell_0$ regularisation. More recent work~\citep[e.g.,][]{outerapprox-l0,l0bnb} has shown that solving~$\ell_0$ linear regression problems to global optimality can be feasible even with millions of variables.

Other examples of penalties include the $\ell_1$ or Lasso penalty~\citep{lasso}, the $\ell_2$ or ridge penalty~\citep{ridge}, and non-convex penalties~\citep[for example,][]{nonconcave1,nonconcave2}. Moreover,
new combined penalties have been developed~\citep[for example,][]{elasticnet,mazumder2017subset,hazimeh2020fast} to adapt to a wide range of noise regimes and sparsity settings.

\noindent\textbf{Linear regression for categorical predictors.} General linear estimators discussed above do not utilise the specific structures of problems with categorical predictors to improve the quality of the estimator. As a result, tailored estimators for linear 
models with categorical predictors (or, more generally, models with clustered regression coefficients) have been developed. From the model accuracy, parsimony and interpretability perspectives, it might be desirable to collapse some categorical levels together to reduce the model dimension. This is equivalent to clustering or fusing together the regression coefficients corresponding to the same categorical predictor.  One common approach for obtaining such clustered solutions is penalising pairwise differences of regression coefficients, thus encouraging clustering. In particular,~\citet{categorical1} propose the CAS-ANOVA approach, which uses an all-pair penalty of the form $\sum_{k_1,k_2}|\theta_{j,k_1}-\theta_{j,k_2}|$. When a natural order can be established among the regression coefficients (e.g., for ordinal variables), \citet{categorical-ordinal} propose penalties of the form $\sum_k |\theta_{j,k}-\theta_{j,k-1}|$, which are closely related to the Fused Lasso penalty~\citep{fusedlasso}. \citet{categorical-l0} utilise pairwise penalties that provide a smooth approximation to~$\ell_0$ fusion penalties of the form $\1(\theta_{j,k_1}\neq \theta_{j,k_2})$.  
The problem of clustering regression coefficients with continuous predictors has been studied extensively in the literature~\citep{wang2024supervised,she2010sparse,bondell2008simultaneous,ke2015homogeneity,shen2010grouping,zhu2013simultaneous}. Another interesting line of work is by~\citet{categorical-grouplasso}, who introduce a two-stage algorithm that uses the Group Lasso estimator~\citep{grouplasso}. We refer to~\citet{categorical-review} for a comprehensive comparison of different penalty-based methods for linear regression with categorical variables, and to~\citet{scope,she2022supervised,radchenko2017convex} for detailed discussions of the statistical and computational drawbacks associated with methods that use all-pairs type penalties.

Our work is motivated by the SCOPE estimator of~\citet{scope}, which is the
state-of-the-art method for the high-dimensional linear regression problem with categorical predictors that have many levels. It uses the Minimax Concave Penalty (MCP) on differences between the order statistics of the coefficients for a categorical variable, thereby clustering the coefficients. SCOPE is an approximate algorithm based on block coordinate descent, which has appealing theoretical guarantees and outperforms prior methods empirically. In a departure from the scope of the work of~\citet{scope}, we develop an exact MIP approach for our problem.

\subsection{Summary of the Approach and Contributions}\label{sec:intro-approach}

We propose \ourmethodlz, a new discrete optimisation based estimator for the sparse high-dimensional linear regression problem with categorical predictors. 
Our approach penalises: (a) the total number of distinct values of the regression coefficients for each categorical predictor, summed over the predictors, $\sum_{j\in[q]}|\{\theta_{j,k}:k\in[p_j]\}|$, to encourage fusion of coefficients corresponding to the same categorical predictor; and (b) the total number of nonzero regression coefficients, $\sum_{j\in[q], k\in[p_j]}\1(\theta_{j,k}\ne 0)$, to encourage sparsity. 
This regularisation scheme allows direct control over both the amount of clustering in the coefficients and the sparsity of the solution, performing model compression through two different  mechanisms. 
The details of our methodology are provided in Section~\ref{sec:proposed.est}.

Compared to the state-of-the-art approach of~\citet{scope}, which uses the MCP penalty to reduce the number of coefficient clusters,  we directly control the number of clusters via our penalty (a). In addition, our penalty (b) allows the option of encouraging sparsity in the coefficients, which can be desirable from a model parsimony standpoint. Furthermore, this penalty also helps with identifiability and can lead to performance improvements in terms of prediction error and model selection in high-dimensional settings.

To illustrate the effects of our regularisation scheme, we consider the bike sharing dataset from the UCI data repository~\citep{misc_bike_sharing_275}. Further details of the data and the outcomes of our experiments are provided and discussed in Section~\ref{bikedata}. Here, we focus on an illustrative example with 100 observations, \texttt{hr} and \texttt{weekday} as the categorical predictors, 
and the number of bike rentals as the response variable. Figure~\ref{fig:intro} illustrates the regression coefficients for various values of $\lambda$ and~$\lambda_0$, which are the penalty weight tuning parameters for the fusion and sparsity penalties, respectively. We observe that by increasing $\lambda$, more coefficients are clustered together as the fusion penalty becomes stronger. Moreover, we see that our estimator with $\lambda_0>0$ leads to a substantially more compact model than the one with~$\lambda_0=0$, while achieving similar out-of-sample performance.

\begin{figure}
\setlength{\tabcolsep}{1pt}
     \centering
    \begin{tabular}{ccc}
      ~~$\lambda=0.5$ & $\lambda=1$ & $\lambda=75$\\
       \includegraphics[width=0.345\linewidth]{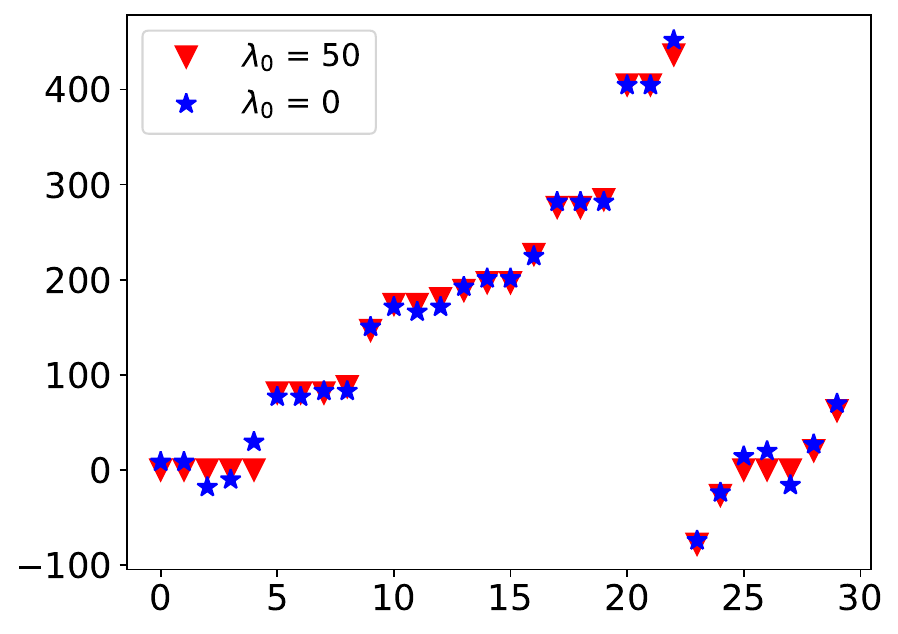}  &  \includegraphics[trim={2cm 0 0 0},clip,width=0.3\linewidth]{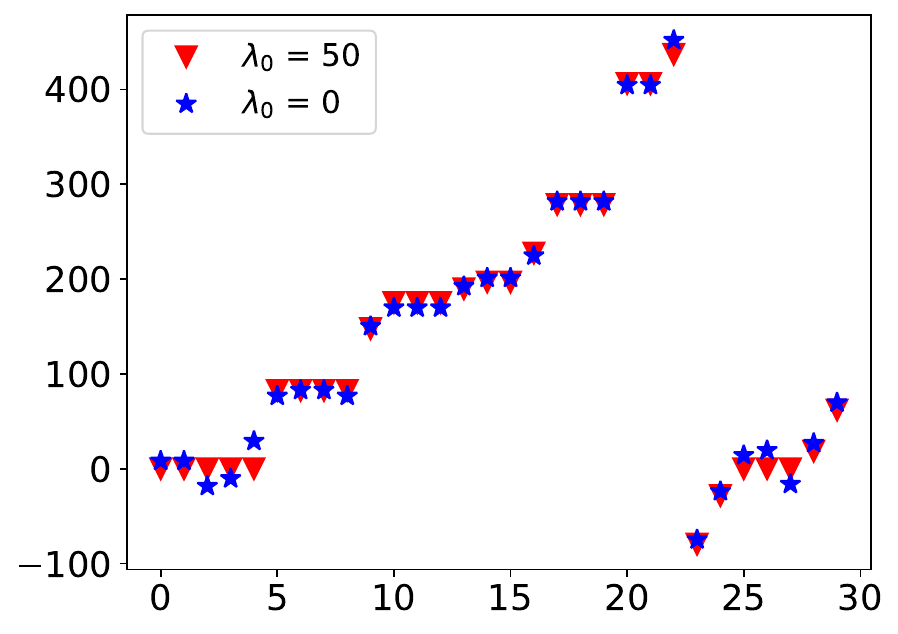} &
       \includegraphics[trim={2cm 0 0 0},clip,width=0.30\linewidth]{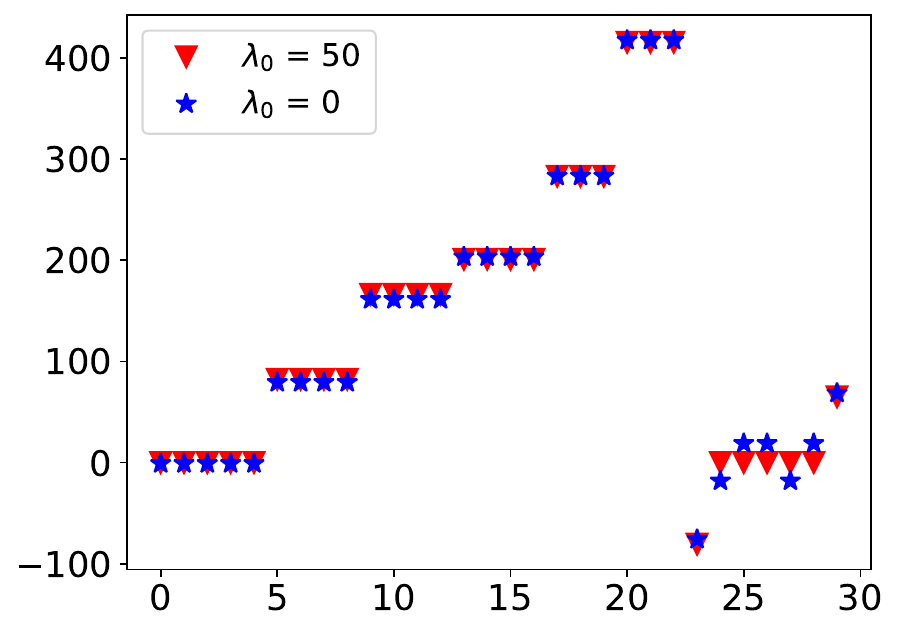} \\
    \end{tabular}
    \captionof{figure}{The regression coefficients for the illustrative example in Section~\ref{sec:intro-approach}. We consider three values of parameter~$\lambda$, which controls the strength of the fusion penalty, and two values of~$\lambda_0$, which controls the strength of the sparsity penalty. 
    The first 23 coefficients correspond to the hour of the day (sorted based on their coefficient values when $\lambda_0=50$), while the rest represent the weekday (sorted similarly).
     The estimators in the two left-most panels (nearly) attain the best out-of-sample performance across~$\lambda$ values with a test set $R^2$ of 0.26. See Appendix~\ref{supp:illustrative} for the interpretation of the regression coefficients.
    }
    \label{fig:intro}
\end{figure}

We now summarise our main contributions.

\noindent\textbf{Algorithms.}
Using binary variables to encode the sparsity and the clustering pattern of the solution, we reformulate our estimator as a solution to a Mixed Integer Program (MIP), where we minimise a convex quadratic objective over a mixed-integer linear constraint set. As a consequence, globally optimal solutions to our optimisation problem can be found using off-the-shelf MIP solvers such as Gurobi (in minutes when $p\sim 100$ based on our experiments). This feature marks an important difference between our approach and the state-of-the-art method of~\citet{scope}, which 
does not solve their optimisation problem to global optimality and instead only
solves it
approximately using Block Coordinate Descent (BCD).

Modeling the fusion of regression coefficient levels with binary variables is more nuanced than MIP formulations arising in the usual $\ell_0$ regression, as in the earlier work~\citep{l0-modern}.
In particular, our formulation allows us to develop a custom row generation procedure for our estimator, which we use to accelerate the exact off-the-shelf solvers. Furthermore, we develop a fast approximate algorithm for our method that obtains high-quality feasible solutions via BCD. Such fast algorithms are desirable for practitioners, for example, when warm-starting exact MIP solvers or selecting tuning parameters. We also extend our approximate algorithm to binary logistic regression with categorical predictors.

As the main building block of our BCD algorithm, we develop a new \emph{exact} algorithm for the univariate case (i.e., with just one categorical predictor) based on dynamic programming, which can be of independent interest. Our algorithm extends the earlier work of \cite{johnson2013dynamic} on change-point detection/signal segmentation. In contrast, the univariate problems in the BCD approach of \citet{scope} use the continuous MCP penalty to encourage coefficient fusion. Our BCD algorithm appears to be faster (up to 500 times in some settings) than the corresponding algorithm for the SCOPE estimator.

\noindent\textbf{Statistical theory.} We show that our approach controls the number of the resulting coefficient clusters while achieving the prediction error rate at least as good as $s^*\sigma^2\log(p)/n$. Here, $p=N+\sum_{j=1}^q p_j$ is the overall model dimension, $n$ is the number of observations, and~$s^*$ is the number of nonzero regression coefficients in the underlying model. 
Moreover, we show that in some settings our estimator can achieve an improved prediction error rate as low as $s^*\sigma^2\log(q)/n$, where~$q$ is the total number of categorical predictors.
We are unaware of existing results
on the prediction performance of specialised linear regression estimators that cluster the levels of categorical predictors.

Furthermore, we show that when the true regression coefficients are clustered, our estimator can recover the correct clustering pattern with high probability under a suitable minimum separation condition on the distinct true coefficient values corresponding to the same categorical predictor. In the univariate setting with just one categorical predictor, our minimum separation level of $\sigma\sqrt{p\log(p)/n}$ is consistent with the existing results and is minimax optimal. In the multivariate setting, we are not aware of any other results that provide similar cluster-recovery guarantees (see Section~\ref{sec:clusterrecovery} for more details).

\noindent \textbf{Numerical results.} Our experiments on synthetic and real datasets  demonstrate that the proposed estimator appears to have better prediction and model selection performance than other existing methods, such as the state-of-the-art approach of~\citet{scope}. In terms of computational efficiency, we show that our approximate algorithm is faster than the corresponding approximate algorithm of~\citet{scope} for large problems, scaling in seconds to problems with~$p$ on the order of thousands. Moreover, our exact algorithm can provide optimality certificates for problems with $p$ in thousands, and shows notable improvements over off-the-shelf MIP solvers for $p$ in the hundreds.

\noindent\textbf{Organisation.} We formulate the problem and introduce our estimator in Section~\ref{prob-form-sec}. We investigate the theoretical properties of our estimator in Section~\ref{sec:theory}. In Section~\ref{sec:algs}, we develop the algorithms for our estimator and provide an extension of our approach to binary classification. We conduct numerical experiments with synthetic and real data in Section~\ref{sec:numerical}. Proofs and technical details are provided in the Appendix.

\section{Problem Formulation and Proposed Methodology}\label{prob-form-sec}

The default version of our method uses data to determine the ``baseline'' categories, that is, categories whose indicator variables get assigned zero coefficients in the regression model. However, at the end of Section~\ref{sec:proposed.est}, we discuss the setting where the user wishes to pre-specify the baseline categories.

\subsection{Notation}
Recall that we have~$q$ categorical predictors $C_1,\cdots,C_q$, and~$N$ continuous predictors $W_1,\cdots,W_N$, with each~$C_j$ taking values in~$[p_j]$ for some $p_j\geq 1$. The data consists of observations $(y^{(i)},\B C^{(i)},\B W^{(i)})_{i=1}^n$ made on~$n$ individuals, where~$\B C^{(i)}$ and~$\B{W}^{(i)}$ are the observed categorical and continuous predictors for the $i$-th individual. 
We introduce dummy variables $C^{(i)}_{j,1},\cdots, C^{(i)}_{j,p_j}$ such that $C^{(i)}_{j,k}=\1(C^{(i)}_j=k)$, and define the expanded $p$-dimensional sets of variables and regression coefficients as
\begin{equation}\label{x-def}
    \B{x}^{(i)}= \Big(\underbrace{C^{(i)}_{1,1},\cdots,C^{(i)}_{1,p_1}}_{\text{categorical predictor 1}},\cdots,\underbrace{C^{(i)}_{q,1},\cdots,C^{(i)}_{q,p_{q}}}_{\text{categorical predictor }q},\underbrace{W^{(i)}_{1},\cdots,W^{(i)}_N}_{\text{continuous features}}\Big)\in\R^p
\end{equation}
\begin{equation}\label{beta-alpha}
    \text{and}\qquad\B{\beta}=(\theta_{1,1},\cdots,\theta_{1,p_1},\cdots,\theta_{q,1},\cdots,\theta_{q,p_{q}},\theta_{1},\cdots,\theta_N)\in\R^p.~~~~~~
\end{equation}
We write $\cI_j\subseteq[p]$ for the index set of the variables in the expanded set that correspond to predictor $C_j$. Specifically, $\cI_1=\{1,\cdots,p_1\}$, $\cI_2=\{p_1+1,\cdots,p_1+p_2\}$ and so on.  
Observations $(y^{(i)},\B{x}^{(i)})_{i=1}^n$ are stored in the response vector~$\B{y}\in\R^n$ and the rows of the model matrix $\B{X}\in\R^{n\times p}$. For $j\in[q]$ and $k\in[p_j]$, we also define
\begin{equation}\label{J_k_j}
    \cJ_{k}^j = \{i\in[n]: C_j^{(i)}=k\}
\end{equation}
as the set of observations for which categorical feature~$j$ takes value~$k$. For each~$\B\beta\in\R^p$, we let $S(\B\beta)=\{j:\beta_j\neq 0\}$ denote the support of $\B\beta$. For $S\subseteq[p]$, we let $\B\beta_{S}$ denote the sub-vector of $\B\beta$ with coordinates indexed by $S$. Similarly, we write~$\B X_{S}$ for the sub-matrix of~$\B X$ whose columns are indexed by~$S$. In particular, we let~$\B X_j$ denote the $j$-th column of~$\B X$.
We treat expressions of the form $0\times\infty$ and $0/0$ as zero. We use the notation $E_1\lesssim E_2$ to mean that inequality $E_1\le E_2$ holds up to a positive universal multiplicative constant, and define $\gtrsim$ analogously. Finally, we write $E_1\asymp E_2$ when $E_1\lesssim E_2\lesssim E_1$.
\subsection{Estimator}
\label{sec:proposed.est}

We propose and analyze the following estimator:
\begin{align}\label{clusteringproblem-reg}
(\hat\alpha,\hat{\B{\beta}})\in\argmin_{\alpha\in\R, \B{\beta}\in\R^p} \quad &   \tfrac{1}{n}\|\B{y} -\alpha\B{1} -\B{X\beta}\|_2^2 + \lambda_0 \|\B{\beta}\|_0 + \lambda \sum_{j=1}^{q} \left\vert\left\{\beta_k:k\in\cI_j\right\}\right\vert,
\end{align}
where $\B{1}\in\R^n$ is the all-one vector and $\lambda,\lambda_0$ are pre-specified non-negative tuning parameters. The $\|\B{\beta}\|_0$ penalty  encourages sparse solutions, while the $\sum_{j=1}^{q} \left\vert\left\{\beta_k: k\in\cI_j\right\}\right\vert$  penalty reduces the total number of levels in~$\B\beta$, thus leading to a clustered solution. By varying $\lambda,\lambda_0$ in~\eqref{clusteringproblem-reg} we can obtain a wide range of estimators that can be sparse, clustered, or both, which illustrates the modeling flexibility of our discrete optimisation based approach. We refer to the above estimator as  \ourmethodlz~when $\lambda_0>0$, and refer to it as \ourmethod~when $\lambda_0=0$.

When $\lambda_0=0$, the solution to Problem~\eqref{clusteringproblem-reg} is generally not unique -- one can add a constant to all the regression coefficients corresponding to a categorical predictor and then subtract the same constant from the intercept without changing either the clustering pattern or the prediction vector. In contrast, when $\lambda_0>0$, Problem~\eqref{clusteringproblem-reg} ensures that the largest cluster\footnote{When there is more than one such cluster, we get multiple equivalent solutions.} of regression coefficients for each categorical predictor is assigned a zero coefficient. This property of the solution is especially interesting when most levels of a categorical predictor have a similar effect on the response, implying that the corresponding regression coefficients could be allocated to the same large cluster. The coefficients for such large clusters are set to zero when $\lambda_0>0$, which leads to a more compact model. Our numerical experiments in Section~\ref{sec:numerical} also demonstrate that the addition of the sparsity penalty helps improve the prediction performance of our estimator in high-dimensional settings.

    In practice, it might be desirable to set to zero the regression coefficients corresponding to some prespecified baseline level of each categorical predictor.  A user can easily achieve this by shifting the coefficients of the solution to Problem~\eqref{clusteringproblem-reg} as described above, without changing either the clustering pattern or the estimated regression function. Specifically, for each categorical predictor, one would subtract the estimated coefficient for the baseline level from all the regression coefficients corresponding to this predictor and then add the same value to the intercept.  
    
    We also note a simple modification of our approach in which the regression coefficients for the baseline levels are constrained to equal zero in Problem~\eqref{clusteringproblem-reg}; all our algorithms and theory carry over with the corresponding minor adjustments. 
   Here, the $\|\B{\beta}\|_0$ penalty would encourage sparsity for the \emph{specific} representation of the solution in which the baseline coefficients are zeroed out. However, in general, we recommend the original version of our approach as it could perform better when the aforementioned representation of the true regression function is not sparse.

\section{Statistical Properties}
\label{sec:theory}
\noindent We study the theoretical properties of our estimator, which is defined as the global solution to Problem~\eqref{clusteringproblem-reg}. First, we present the prediction error bounds and discuss the settings in which they can improve upon the existing bounds from the sparse linear regression literature. Next, we study the clustering performance of our estimator and show that 
it can recover the correct clustering pattern with high probability under a suitable minimum separation condition on the true regression coefficients.
In this work, we do not pursue the question of how well the true regression coefficients are estimated; as we note in Section~\ref{sec:proposed.est}, these coefficients are generally not identifiable.

\noindent\textbf{Model Setup.}
 Throughout this section, we assume that the  observations are generated as $\B{y} = \mathbf{f}^* + \B{\epsilon}$. Here, $\mathbf{f}^*\in\R^n$ is the (unobserved) mean response vector such that $\mathbf{f}^*=(f^*(\B x^{(1)}),...,f^*(\B x^{(n)}))^\top$ for some function~$f^*: \mathbb{R}^p\mapsto \mathbb{R}$, which is not required to have any particular form, and the noise vector $\B{\epsilon}\in\R^n$ is drawn from the $\cN(\B{0},\sigma^2 \B{I}_n)$ distribution. We also assume that $\B{X}$ and $\mathbf{f}^*$ are deterministic.

\subsection{Prediction Error Bounds}

Suppose that $(\hat{\alpha},\hat{\B\beta})$ is a global solution to Problem~\eqref{clusteringproblem-reg}. We define 
\begin{equation}
    \hat{K} = \sum_{j=1}^q |\{\hat{\beta}_i:i\in \cI_j,\hat\beta_i\neq 0\}| \quad\text{and}\quad K(\bbeta) = \sum_{j=1}^q |\{\beta_i:i\in \cI_j,\beta_i\neq 0\}|
\end{equation}
as the total number of \emph{nonzero} regression coefficient clusters determined by~$\hat{\B\beta}$ and~$\bbeta$, respectively. Our first result derives an oracle prediction error bound for~$\hat{\B\beta}$  without imposing any further assumptions (linearity, sparsity or clustering) on the true model, thus allowing for model misspecification. 
\begin{theorem}\label{estimation_thm}
Let~$(\hat\alpha,\hat{\B{\beta}})$ be a global optimal solution to Problem~\eqref{clusteringproblem-reg} with 
$\lambda_0\geq c_{\lambda_0}\sigma^2 \log (ep)/n$
for some sufficiently large $c_{\lambda_0}>0$. 
Then, with high probability,\footnote{An explicit expression for the probability can be found in~\eqref{thm1-prob} in the appendix.}
\begin{equation*}
        \tfrac{1}{n}\|\mathbf{f}^*-\B{X}\hat{\B{\beta}}-\hat\alpha \B{1}\|_2^2 + \lambda \hat{K}  \lesssim \inf_{\substack{\B\beta\in\R^p \\ \alpha\in\R}} \big\{\tfrac{1}{n}\|\mathbf{f}^*-\B X\B\beta-\alpha\B{1}\|_2^2 + \lambda_0 \|\bbeta\|_0+ \lambda K(\bbeta)\big\}+\tfrac{\sigma^2\log(ep)}{n}. 
\end{equation*}
\end{theorem}
We are not aware of other existing results on the prediction performance of a regularised linear regression estimator that clusters levels of categorical predictors.

Consider the linear model setting, where $\mathbf{f}^*=\alpha^*\B{1}+\B X\B\beta^*$. Note that one can modify $\bbeta^*$ without changing the linear approximation by adding a constant to all the $\bbeta^*$-coefficients corresponding to a particular categorical predictor and then updating~$\alpha^*$ accordingly. In our analysis, we focus on the \emph{sparsest} representation of~$\bbeta^*$. Consequently, the $\bbeta^*$-coefficients of the largest true cluster for every categorical predictor are zero. Let  $\|\B\beta^*\|_0=s^*$ and $K^* = \sum_{j=1}^q |\{\beta^*_i:i\in \cI_j,\beta^*_i\neq 0\}|$. Theorem~\ref{estimation_thm} shows that estimator~\eqref{clusteringproblem-reg} achieves a prediction error rate of $(\sigma^2 s^*/n)\log (ep)+ \lambda K^*$. If 
   $\lambda\lesssim \sigma^2 s^* \log (ep)/(n K^*)$,
Theorem~\ref{estimation_thm} implies an error rate of $\sigma^2s^*\log(p)/n$. Note that when no clustering penalty is applied, i.e., $\lambda=0$, standard results from the sparse linear regression literature \citep[for example,][]{minimax} demonstrate that the $\sigma^2s^*\log (p/s^*)/n$ error rate is minimax optimal. Therefore, 
our estimator achieves the optimal prediction error rate up to a logarithmic factor, while also applying a clustering penalty. Hence, estimator~\eqref{clusteringproblem-reg} can achieve good prediction performance while clustering the coefficients even when the true regression coefficients may not be clustered. Moreover, Theorem~\ref{estimation_thm} implies the following upper bound on the number of nonzero clusters determined by~$\hat{\B\beta}$:
\begin{equation*}
    \hat{K} \lesssim  K^* + \frac{\sigma^2s^*\log (ep)}{n\lambda}.
\end{equation*}

Next, we show that in some settings we can improve the prediction error bound in Theorem~\ref{estimation_thm}. 
We note that when $\lambda_0>0$, our estimator~\eqref{clusteringproblem-reg} can perform feature selection among the categorical predictors. More specifically, when the clustering penalty parameter~$\lambda$ is sufficiently large, some of the predictors will have a single estimated cluster whose coefficients will be set to zero by adjusting the intercept accordingly, thus removing the entire predictor (and all of its dummy variables) from the estimated model.
Such group-wise variable selection results in more compact models and may improve the prediction performance.  We formalise this intuition in Theorem~\ref{estimation_thm2}.  To simplify the presentation, we focus on the case where all predictors are categorical, i.e., $N=0$.
\begin{theorem}\label{estimation_thm2}
    Suppose that $N=0$ and define $p_{\max}=\max_{j\in[q]}p_j$.
    Let~$(\hat\alpha,\hat{\B{\beta}})$ be a global solution to Problem~\eqref{clusteringproblem-reg} with $\lambda,\lambda_0$ satisfying at least one of the inequalities $\lambda\geq c_{\lambda}\sigma^2 p_{\max} \log(eq)/n$ and~$  \lambda_0\geq c_{\lambda_0}\sigma^2 \log (ep)/n$
for sufficiently large positive universal constants $c_{\lambda},c_{\lambda_0}$. Then, with high probability,\footnote{\label{thm1.5prob}An explicit expression for the probability can be found in~\eqref{thm2new-prob} in the appendix.}
$$    \tfrac{1}{n}\|\mathbf{f}^*-\B{X}\hat{\B{\beta}}-\hat\alpha\B{1}\|_2^2  \lesssim  \inf_{\substack{\B\beta\in\R^p \\ \alpha\in\R}} \big\{\tfrac{1}{n}\|\mathbf{f}^*-\B X\B\beta-\alpha\B{1}\|_2^2 + \lambda_0 (\|\bbeta\|_0\lor 1)+ \lambda (K(\bbeta)\lor 1)\big\}. $$
\end{theorem}

The following result is a direct consequence of Theorem~\ref{estimation_thm2} in the linear setting.


\begin{corollary}\label{cor-thm-2}
Suppose  that $\bff^*=\alpha^*\B{1}+\B X\B\beta^*$ and the assumptions of Theorem~\ref{estimation_thm2} hold. With high probability,\footref{thm1.5prob} we have
    $$ \tfrac{1}{n}\|\mathbf{f}^*-\B{X}\hat{\B{\beta}}-\hat\alpha\B{1}\|_2^2\lesssim \frac{\sigma^2 p_{\max}K^*\log(eq)}{n}\land \frac{\sigma^2 s^*\log(ep)}{n}.$$
\end{corollary}

For illustration, consider the setting where $p_1=\cdots=p_q=p_{\max}$ and $s^*\asymp K^*p_{\max}$. The last relationship holds, for example, when the number of true coefficient clusters for each predictor is bounded, and the number of true nonzero coefficients for each true predictor is of order~$p_{\max}$. 
In this setting, Corollary~\ref{cor-thm-2} implies an error rate of
$
\sigma^2s^*\log(eq)/n$,
while Theorem~\ref{estimation_thm} implies an error rate of 
$\sigma^2s^*\log(eqp_{\max})/n$.
Hence, the rate in Corollary~\ref{cor-thm-2} is tighter by a factor of $\log(q)/\log(qp_{\max})$. This improvement is favorable when the predictors have many levels, i.e., when~$p_{\max}$ is large.

We also note that when $q=1$ and $s^*\asymp p$, 
the prediction error bound in Corollary~\ref{cor-thm-2} matches the minimax lower bound etablished in~\citet{she2022supervised} for the linear regression problem with multivariate response and continuous predictors.

\subsection{Cluster Recovery Guarantees}\label{sec:clusterrecovery}
In this section, we show that under suitable conditions our estimator correctly recovers the true clusters determined by~$\bbeta^*$. We specify the model setup and introduce some notation in Section~\ref{sec.clust.th.prelim}. We present a high-level main result in Section~\ref{sec.clust.th.main}, and study concrete examples and special cases in Section~\ref{sec.clust.th.suff}. In Section~\ref{sec.clust.th.comp}, we discuss the relevant results in the literature.

\subsubsection{Preliminaries}
\label{sec.clust.th.prelim}
Throughout Section~\ref{sec:clusterrecovery}, we focus on the linear model setting:  $\bff^*=\alpha^*\B{1}+\bX\bbeta^*$. As before, we choose the value of~$\alpha^*$ so that~$\bbeta^*$ is the sparsest possible. In particular, the largest true cluster for every categorical predictor has zero $\bbeta^*$-coefficients.  To simplify the exposition, we focus on the setting where all the predictors are categorical, i.e., $N=0$. 

We use the term ``clustering pattern'' to refer to 
a partition of the~$p$ dummy variables, where each cluster corresponds to one particular categorical predictor. That is, dummy variables corresponding to different categorical predictors cannot belong to the same cluster. A formal definition is provided in Appendix~\ref{sec.prf.prelim}. Note that each candidate regression coefficient vector~$\bbeta\in\mathbb{R}^p$ determines a clustering pattern, where if two dummy variables for the same predictor have the same $\bbeta$-coefficient then they are clustered together. We write $\cG=\cG(\B\beta)$ for the clustering pattern determined by~$\bbeta$ and let $\cG^*=\cG(\B\beta^*)$. We also let $s^*=\|\bbeta^*\|_0$ be the sparsity level of the true coefficients, and we let $c^*=(s^*/K^*)\vee 1$ and $n^*=n/(\min_{j\in[q]} p_j)$. 

Next, we define a measure of ``clustering impurity'' for a candidate clustering pattern~$\mathcal{G}=\mathcal{G}(\bbeta)$ with respect to the true pattern~$\mathcal{G}^*$:
$$
\mathcal{E}(\mathcal{G})=\sum_{k=1}^{q} \; \sum_{a\in\{\beta_j:\,j\in\cI_k\}} \Big(\big|\{l\in\cI_k: {\beta}_l=a\}\big|-\max_{b\in\{\beta_j^*:\,j\in\cI_k\}}\big|\{l\in\cI_k: {\beta}_l=a,\,\hat\beta^*_l=b\}\big|\Big).
$$
\begin{remark}
    If $\mathcal{G}=\mathcal{G}(\bbeta)=\mathcal{G}(\tilde{\bbeta})$, then using $\tilde{\bbeta}$ instead of $\bbeta$ in the above display does not change the value of $\mathcal{E}(\mathcal{G})$. Thus, $\mathcal{E}(\mathcal{G})$ is well-defined and depends only on the clustering pattern~$\mathcal{G}$.
\end{remark}

Our impurity measure~$\mathcal{E}(\mathcal{G})$ counts the number of variables that need to be removed from the $\mathcal{G}$-clusters to make them ``pure'' with respect to the $\mathcal{G}^*$-cluster membership.\footnote{That is, the members of each resulting cluster belong to the same $\mathcal{G}^*$-cluster.} In fact, quantity $1-\mathcal{E}(\mathcal{G})/p$, known as purity~\citep[Ch. 16]{puritydef}, is often used to measure clustering quality. We now illustrate $\mathcal{E}(\mathcal{G})$ with an example.

\begin{example}
    Consider a cluster $G\in\cG$ that has four members with the following values of the true $\bbeta^*$-coefficients: $c,c,c,d$, where $c\ne d$. Then, the contribution of this cluster towards~$\mathcal{E}(\bbeta)$ is~$1$. Alternatively, if the corresponding $\bbeta^*$-coefficients are $c,c,d,d$, then the contribution towards~$\mathcal{E}(\bbeta)$ is~$2$.
\end{example}

Given a prediction vector~$\bff\in\mathbb{R}^n$ of the form $\bff=\alpha\B{1}+\bX\bbeta$ for some $\alpha\in\mathbb{R}$ and $\bbeta\in\mathbb{R}^p$, we define $\mathcal{G}_{\bff}=\mathcal{G}(\bbeta)$. Note that if $\tilde{\alpha}\B{1}+\bX\tilde{\bbeta}=\alpha\B{1}+\bX{\bbeta}$ then $\mathcal{G}(\tilde{\bbeta})=\mathcal{G}(\bbeta)$, and hence~$\mathcal{G}_{\bff}$ is well-defined. In our cluster recovery results, we need to quantify the amount of the signal $\bff^*$ that remains unexplained when a candidate clustering pattern~$\cG$ is used instead of~$\cG^*$. A formalisation of this notion is given below.

\begin{definition}[Approximation Error]\label{def:ebeta}
Given a clustering pattern~$\mathcal{G}$, we define
    \begin{equation}
       w(\mathcal{G}) =  \inf_{\bff: \cG_{\bff}= \cG} \Big(\frac{1}{n}\|\bff^*-\bff\|_2^2\Big).
    \end{equation}
\end{definition}

We can think of~$w(\cG)$ as the error made when $\bff^*$ is approximated by vectors~$\bff$ of the form $\alpha\B{1}+\bX\bbeta$, such that the clustering pattern of~$\bbeta$ (i.e., $\mathcal{G}_{\bff}$) is the same as~$\cG$.

\subsubsection{Main result}
\label{sec.clust.th.main}

The following assumption imposes some regularity conditions on the sizes of the true clusters in the clustering pattern~$\cG^*$, which is determined by the regression coefficient vector~$\bbeta^*$. If a cluster contains zero $\bbeta^*$-coefficients, we call it a \textit{zero} cluster, otherwise we call it a \textit{nonzero} cluster.
\begin{asu}\label{assumptiona}
    The sizes of the true nonzero clusters are of the same order:
        $|G|\asymp c^*$ for all the nonzero clusters $G\in\cG^*$.
\end{asu}

We say that a clustering pattern~$\cG$ is $s$-sparse if there exists $\bbeta\in\mathbb{R}^p$ with $\|\bbeta\|_0\le s$ such that $\cG=\cG(\bbeta)$. We are now ready to present the main result of this section.

\begin{theorem}\label{clust.recov.thm}
Let $\hat{\B\beta}$ be an optimal solution to Problem~\eqref{clusteringproblem-reg} with
\begin{equation}\label{lambda-thm2}
    a_0\frac{\sigma^2\log p}{n}\leq \lambda_0 \leq b_0\frac{\sigma^2\log p}{n},~~~~~~a\frac{c^*\sigma^2\log p}{n}\leq \lambda \leq b\frac{c^*\sigma^2\log p}{n}
\end{equation}
for positive constants  $a_0, b_0,a$, and~$b$, such that~$a_0$ and~$a$ are sufficiently large.
In addition, suppose that there exists a constant $\kappa>1$ such that
\begin{equation}
\label{assumptionA.bnd}
w(\cG)\gtrsim \frac{\mathcal{E}(\cG)\sigma^2\log(p)}{n}~~\text{for all~} \kappa s^*\text{-sparse\ clustering patterns~} \cG
\end{equation}
with a sufficiently large multiplicative constant. If Assumption~\ref{assumptiona} holds, then, $\hat{\B\beta}$ and $\B\beta^*$ have the same clustering pattern with high probability. Specifically, $\p(\widehat{\cG}=\cG^*)\geq 1-p^{-c}$ for some positive constant~$c$.
\end{theorem}

\begin{remark}
Theorem~\ref{clust.recov.thm} shows that our approach recovers the true clusters when $w(\cG)/\mathcal{E}(\cG)$ is sufficiently large for all $\cG\ne \cG^*$. Interestingly, our independently derived ratio $w(\cG)/\mathcal{E}(\cG)$ has notable similarities (but also key differences) with a corresponding quantity~$c_{\min}$ in Theorem~2.4 of \cite{wang2024supervised}, which analyses an estimator that clusters the regression coefficients of continuous predictors.
\end{remark}

\begin{remark}
We note that the impurity measure $\mathcal{E}(\cG)$ satisfies $\mathcal{E}(\cG)\le 2s^*$ for all~$\cG$, and hence a simple sufficient condition for~(\ref{assumptionA.bnd}) to hold is $w(\cG)\gtrsim s^*\log p /n$. However, our results in Section~\ref{sec.clust.th.suff} demonstrate that we can obtain milder conditions under which~(\ref{assumptionA.bnd}) holds true.
\end{remark}

\begin{remark}
We show in the proof of Theorem~\ref{clust.recov.thm} 
that we can extend the range of~$\lambda$ to $c^*\sigma^2\log(K^*\vee 2)/n\lesssim\lambda\lesssim c^*\sigma^2\log p/n$,
in which case we have $\p(\widehat{\cG}=\cG^*)\geq 1-\big[(K^*\vee 2)^{-cc^*}\vee p^{-c}\big]$, for some~$c>0$.
\end{remark}

To better understand the main assumption of Theorem~\ref{clust.recov.thm}, we derive a sufficient condition for~\eqref{assumptionA.bnd} to hold.

\subsubsection{Sufficient Condition for~\eqref{assumptionA.bnd}}
\label{sec.clust.th.suff}

Let $n_{j,k}=|\mathcal{J}_k^j|$ denote the count of observations corresponding to the $k$-th level of the $j$-th predictor, where $\mathcal{J}_k^j$ is defined in~\eqref{J_k_j}. To derive our result, we impose the following assumption.

\begin{asu}\label{assumptionb}
    The counts of observations per predictor level are of the same order across all
 predictors and  
levels, i.e., $n_{j,k}\asymp n^*$ for all~$j\in[q]$ and $k\in[p_j]$.
\end{asu}

Informally speaking, in this section we show that if there is sufficient separation between the distinct values of the true regression coefficients corresponding to the same categorical predictor, then~\eqref{assumptionA.bnd} holds provided there is not too much overlap among the predictors. In the definition below we formalise the notion of the minimum separation between the true regression coefficients.

\begin{definition}[Minimum Separation]\label{def:min.sep}
Let us define 
$$\Delta_{\text{min}} = \min \{ |\beta^*_{k}-\beta^*_{l}| : \beta^*_k\neq \beta^*_l, j\in[q], k,l\in\cI_j  \}.$$
\end{definition}

We also define $\mathcal{B}(\bff)=\{\bbeta\in\mathbb{R}^p, \text{ s.t. }\; \bff=\alpha\B{1}+\bX\bbeta \text{ for some } \alpha\in\mathbb{R} \}$,
which is the set of all possible regression coefficient vectors that correspond to the prediction vector~$\bff$; and $\mathcal{F}(s)=\{\bff\in\mathbb{R}^n, \text{ s.t. } \bff=\alpha\B{1}+\bX\bbeta \text{ for some } \alpha\in\mathbb{R},\;\text{and}\;\bbeta\in\mathbb{R}^p \; \text{with}\;\|\bbeta\|_0\le s \}$,
which is the set of all prediction vectors that have an $s$-sparse linear representation.
The following result provides a sufficient condition for the main assumption of Theorem~\ref{clust.recov.thm}.
\begin{theorem}
\label{th.BA}
Suppose that Assumptions~\ref{assumptiona} and~\ref{assumptionb} hold. Moreover, assume that there exists a constant $\kappa>1$ such that 
\begin{equation}
\label{cond.BA}
\inf_{ \substack{\bff\in\mathcal{F}(\kappa s^*) \\ \bff\ne\bff^* }} 
\sup_{ \substack{{\bbeta}\in\mathcal{B}(\bff) \\ \tilde{\bbeta}^*\in\mathcal{B}(\bff^*) }}
\,\frac{\|\bff^*-\bff\|_2}{\|\tilde{\bbeta}^*-{\bbeta}\|_2}\gtrsim \sqrt{n^*}.
\end{equation}
If $\Delta_{\text{min}}\gtrsim \sigma\sqrt{\frac{\log p}{n^*}}$
for a sufficiently large multiplicative constant, then condition~(\ref{assumptionA.bnd}) is satisfied.
\end{theorem}

For an illustration, we consider the univariate case ($q=1$) and verify condition~\eqref{cond.BA}. In the univariate case, each row of $\B X\in\R^{n\times p}$ has exactly one nonzero coordinate, which is equal to~$1$.
Hence, there exists a unique vector~$\tilde{\bbeta}^*$ such that 
$\bff^*=\bX\tilde{\bbeta}^*$. Similarly, for every $\bff\in\mathcal{F}(\kappa s^*)$, we can 
find a unique~${\bbeta}$ to satisfy  $\bff=\bX{\bbeta}$. Note that the
columns of~$\B X$ are orthogonal. Thus, we have
$$
\|\bff^*-\bff\|_2^2=\sum_{k=1}^p\|\bX_{k}\|^2(\tilde{\beta}^*_k-{\beta}_k)^2
\gtrsim n^*\|\tilde{\bbeta}^*-{\bbeta}\|_2^2,
$$
where the last inequality holds by Assumption~\ref{assumptionb}, and hence condition~(\ref{cond.BA}) is satisfied. Note that $n^*\asymp n/p$ by Assumption~\ref{assumptionb}. Consequently, we have the following result as a direct corollary of Theorem~\ref{th.BA}.
\begin{corollary}
\label{univariate.corollary}
Suppose that $q=1$ and Assumptions~\ref{assumptiona} and~\ref{assumptionb} hold. If
$\lambda_0,\lambda$ are taken as in~\eqref{lambda-thm2} and $\Delta_{\text{min}} \gtrsim \sigma\sqrt{p\log(p)/n}$, with sufficiently large  multiplicative constants in the lower bounds, then
$\p(\widehat{\cG}=\cG^*)\geq 1-p^{-c}$ for some positive constant~$c$.
\end{corollary}

The following result shows that if we strengthen Assumption~\ref{assumptiona} by imposing $|G|\lesssim c^*$ on all the true \textit{zero} clusters~$G$, then Corollary~\ref{univariate.corollary} continues to hold even with $\lambda_0=0$, i.e., \textit{without} penalising the number of nonzero coefficients.
\begin{corollary}
\label{cor.univ.dense}
Under an additional assumption that $|G|\lesssim c^*$ for all the true zero clusters~$G\in\cG^*$, Corollary~\ref{univariate.corollary} holds for a wider range of the~$\lambda_0$ tuning parameter: $0\le \lambda_0\lesssim \sigma^2\log(p)/n$.
\end{corollary}

As we note above, condition~(\ref{cond.BA}) of Theorem~\ref{th.BA} is automatically satisfied in the univariate case under Assumption~\ref{assumptionb}. The following result considers the high-dimensional case and illustrates that condition~(\ref{cond.BA}) is satisfied under suitable assumptions in the random design setting.

\begin{proposition}
\label{prop.rand}
    Suppose that the categorical predictors are independently generated from a discrete uniform distribution with~$L$ outcomes. If $s^*\le C$ for some constant~$C$, then there exists a  constant~$c$, such that~(\ref{cond.BA}) holds with high probability when $n\ge c\log(p)$. Here, $L$ and~$p$ are allowed to be arbitrarily large.
\end{proposition}

\subsubsection{Comparison with Existing Results}
\label{sec.clust.th.comp}
We now discuss the most relevant existing cluster recovery results for the linear regression problem with categorical predictors that have many levels. The SCOPE estimator of~\citet{scope} bears similarities to our approach with some notable differences: it uses a minimax concave penalty on the differences of the ordered regression coefficients (rather than our $\ell_0$-type penalty) to encourage clustering; it also does not include a penalty that encourages sparsity in the coefficients. 

In the univariate setting with just one categorical predictor, \citet{scope} show that a minimum separation condition of the form $\Delta_{\min}\gtrsim \sigma\sqrt{p\log(p)/n}$ leads to correct cluster recovery, which is in agreement with our result in Corollary~\ref{univariate.corollary}. As \citet{scope} discuss, this level of the required minimum separation is minimax optimal in the univariate setting \citep[see also][]{lu2016statistical}. Due to our use of the sparsity inducing $\ell_0$-penalty, we allow the sizes of the true zero clusters to grow arbitrarily large without adversely affecting the minimum separation level required (see our Corollary~\ref{univariate.corollary}). This is not the case for the estimator in \citet{scope}, which requires that the sizes of \emph{all} the true clusters are of the same order (see their Theorem~5; also see our Corollary~\ref{cor.univ.dense} which gives the corresponding result for our estimator without the~$\ell_0$ penalty). We also note that in the case $p\leq n$,  results similar to those in~\citet{scope} are  established in~\citet{she2022supervised} for the linear regression problem with continuous predictors.

In the multivariate case, \citet{scope} provide an optimisation algorithm for their estimator based on block coordinate descent. They show that with high probability and under suitable separation conditions, there exists a stationary point of this algorithm that has the correct clustering pattern -- the algorithm does not move if it reaches this \emph{oracle} stationary point. However, due to the nonconvexity of the optimisation problem, there might be many stationary points and it is unclear if their algorithm would reach the oracle one. In contrast, our results in Theorem~\ref{clust.recov.thm} are presented for the global solution to Problem~\eqref{clusteringproblem-reg}. As we show in Section~\ref{mip-sec}, Problem~\eqref{clusteringproblem-reg} can be written as a MIP. Moreover, our numerical experiments in Section~\ref{sec:time-bench} demonstrate that Problem~\eqref{clusteringproblem-reg} can be solved to certifiable (near) optimality for moderately-sized instances in minutes (where $p$ is in the thousands). This property highlights a strength of our MIP-based approach, which allows us to present theoretical guarantees for global solutions that can be obtained in practice.

\section{Algorithms}\label{sec:algs}

Here we study exact and approximate algorithms for our estimator, \ourmethodlz. First, we derive a novel MIP formulation for Problem~\eqref{clusteringproblem-reg}, demonstrating that we can obtain 
globally optimal solutions to our problem.
We also propose specialised modifications that improve the runtime of the MIP solver over off-the-shelf solvers.  In addition to exact methods, we present a new approximate algorithm for our estimator to deliver high-quality solutions quickly. This algorithm is of interest to obtain high-quality warm-starts for exact MIP solvers and to tune hyperparameters.
\subsection{MIP Formulation and Solver}\label{mip-sec}
\subsubsection{A MIP Formulation}
 We consider binary indicator variables $z_i=\1(\beta_i\neq 0)$ for $i \in [p]$ that encode whether the $\B\beta$ regression coefficients are zero or not. Note that $\|\B\beta\|_0=\sum_{i=1}^p z_i$. 
 We now discuss how we can use binary variables to encode the clustering pattern of the regression coefficients. To this end, let us study a special case with one categorical variable ($q=1,N=0$). For $i\neq k\in[p]$, we define indicator variables $z_{i,k}=\1(\beta_i\neq \beta_k)$ that describe the clustering pattern of the categorical predictor.
  If $z_{i,k}=1$ for all $k=1,\cdots,i-1$, then we have $\beta_i\notin \{\beta_1,\cdots,\beta_{i-1}\}$. In other words,
  $\1(\beta_i\notin \{\beta_1,\cdots,\beta_{i-1}\})=\1\left(\sum_{k=1}^{i-1} z_{i,k}=i-1\right)$.
On the other hand, we also have $    \left\vert\{\beta_i:i\in[p]\}\right\vert = 1+\sum_{i=2}^p\1(\beta_i\notin \{\beta_1,\cdots,\beta_{i-1}\})$.
 Putting these together, we derive 
 $ \left\vert\{\beta_i:i\in[p]\}\right\vert =1+\sum_{i=2}^p\1\left(\sum_{k=1}^{i-1} z_{i,k}=i-1\right)$,  showing that the number of distinct levels in $\B\beta$ can be represented using the binary variables~$\{z_{i,k}\}_{i,k\in[p]}$.

Next, we turn our attention to the general case. Using the notation from Section~\ref{prob-form-sec}, suppose that we have $\cI_j=\{s_j,\cdots,s_j+p_j-1\}$ for some $s_j\geq 1$ and all $j\in[q]$. Building on the discussion above, we consider the following MIP:

\begin{align}
   \obj:=  \min_{\B\beta,\alpha,\B{z},\B{l}} \quad &   \tfrac{1}{n}\|\B y-\B{X\beta}-\alpha\B{1}\|_2^2 + \lambda_0 \sum_{i=1}^p z_i + \lambda \sum_{i=1}^p l_i \label{clusteringproblem-mip} \\
     \text{s.t.} \quad & z_i\in\{0,1\}, l_i\in\{0,1\} ~\forall i\in[p],\nonumber \\
     & z^j_{i,k}\in\{0,1\}~\forall j\in[q], s_j\leq i\leq s_j+p_j-1, s_j\leq k\leq  i-1,  \nonumber \\
     & |\beta_i|\leq Mz_i, ~\forall i\in[p],\nonumber\\
 & \left\vert\beta_k-\beta_i\right\vert \leq 2M z^j_{i,k},~\forall s_j\leq k\leq i-1,\, s_j\leq i\leq s_j+p_j-1,j\in[q] ,\nonumber\\
     & \sum\nolimits_{k=s_j}\nolimits^{i-1} z^j_{i,k} - (i-s_j-1)\leq l_i, ~\forall s_j\leq i\leq s_j+p_j-1,j\in[q]. \nonumber
\end{align}
In Problem~\eqref{clusteringproblem-mip}, the binary variables $\{z_i\}_i$ encode the sparsity pattern of a solution $\B\beta$ through the constraints $|\beta_i|\leq M z_i$ for all $i$. Here $M>0$ is a pre-specified large positive constant often known as the Big-M parameter–see~\citet{l0-modern} for further details for the best subset selection problem. Similarly, for a fixed $j$, the binary variables
$z^j_{i,k}$ (for all $i,k$) describes the clustering pattern for the $j$-th categorical variable. Specifically, the variables $\{z_{i,k}^j\}_{i,j,k}$ 
specify which regression coefficients are fused together through the constraints $|\beta_k-\beta_i|\leq 2Mz_{i,k}^j$ for all $(i,j,k)$-tuples. The decision variable $l_i$ appearing in the last constraint in~\eqref{clusteringproblem-mip} indicates whether the regression coefficient of the $i$-th level is distinct from the coefficients of the prior levels of the predictor $j$. Hence, $\sum_{i=s_j}^{s_j+p_j-1}
l_i$ represents the number of clusters for the $j$-th categorical predictor,
 which is subsequently penalised in the objective. The following result addresses equivalence of Problems~\eqref{clusteringproblem-reg} and~\eqref{clusteringproblem-mip}.
\begin{proposition}\label{mip-prop}
    For a sufficiently large~$M$, Problems~\eqref{clusteringproblem-reg} and~\eqref{clusteringproblem-mip} are equivalent.
\end{proposition}
Problem~\eqref{clusteringproblem-mip} is a convex MIP with a convex quadratic objective and mixed integer linear constraints. Therefore, off-the-shelf MIP solvers such as Gurobi can be used to obtain optimal solutions to this problem for small/moderate-scale instances where~$p$ is around one hundred or smaller. 

\subsubsection{A Specialised Row Generation Algorithm}\label{subsec:rowgen}
Problem~\eqref{clusteringproblem-mip} is 
a MIP that can be solved to optimality by off-the-shelf solvers, such as Gurobi, for small/moderate problem sizes. However, based on our experiments, such solvers can face difficulties for problems with $p$ in the hundreds or larger. One key challenge of formulation~\eqref{clusteringproblem-mip} that makes it different from usual $\ell_0$ sparse regression~\citep{l0-modern} is that it can involve up to $\mathcal{O}(p^2)$-many constraints and binary variables, posing computational difficulties. To this end, we design a row generation algorithm that can offer improvements over off-the-shelf solvers.

The main idea behind our row generation algorithm is the following observation: Suppose that $i,k\in\cI_j$ for some $j\in[q]$. Then, if $z_i=z_k=0$, the constraint $|\beta_i-\beta_k|\leq 2Mz_{i,k}^j$ holds regardless of the value of $z_{i,k}^j$, as $\beta_k=\beta_i=0$. Therefore, all such constraints can be effectively `ignored' for this specific value of $\B z$. Intuitively speaking, removing these constraints can result in a formulation with fewer constraints which can be faster to solve, but will still provide suitable optimality guarantees (see Proposition~\ref{mip-lower-prop}).  More specifically, suppose that $\B\beta^{(1)},\cdots,\B\beta^{(T)}$ are $T$ feasible solutions to Problem~\eqref{clusteringproblem-mip}  [we discuss selecting $\{\B\beta^{(t)}\}_t$ below]. For $t\in[T]$, we let $S_t=\{j\in[q]: \beta^{(t)}_i\neq 0~\text{ for some } i\in\cI_j\}$ 
be the set of categorical predictors that have nonzero coefficients in $\B\beta^{(t)}$ and consider a simplified version of Problem~\eqref{clusteringproblem-mip}:
\begin{align}\label{l0clustering-mip-relaxed}
  \widehat\obj_{T+1}=   \min_{(\B\beta,\alpha,\B{z},\B{l})\in\widehat{C}_T} \quad &   \frac{1}{n}\|\B y-\B{X\beta}-\alpha\B{1}\|_2^2 + \lambda_0 \sum_{i=1}^p z_i + \lambda \sum_{i=1}^p l_i, 
\end{align}
where~$\widehat{C}_T$ be the constraint set for $(\B\beta,\alpha,\B{z},\B{l})$ in~\eqref{clusteringproblem-mip} with one key modification: in the inequalities $|\beta_k-\beta_i|\leq 2M z^j_{i,k}$ we restrict~$j$ to a smaller set~$\cup_{t\in[T]} S_t$.
We now show that Problem~\eqref{l0clustering-mip-relaxed} can deliver optimality guarantees for Problem~\eqref{clusteringproblem-mip}.
\begin{proposition}\label{mip-lower-prop}
    Suppose that $\hat{\B\beta}$ is an optimal solution to Problem~\eqref{l0clustering-mip-relaxed}. Then, \\
    \noindent\textbf{(1)} The optimal objectives of~\eqref{clusteringproblem-mip} and~\eqref{l0clustering-mip-relaxed} satisfy $\widehat\obj_{T+1}\leq \obj$.\\
    \noindent\textbf{(2)} If $S(\hat{\B \beta})=S( \B\beta^{(t)})$ for some $t\in[T]$, then  $\hat{\B \beta}$ is also optimal for Problem~\eqref{clusteringproblem-mip}.
\end{proposition}

Motivated by Proposition~\ref{mip-lower-prop}, we consider the following procedure. After obtaining a warm-start such as $\B\beta^{(1)}$, we solve~\eqref{l0clustering-mip-relaxed} to obtain ${\B\beta}^{(2)}$. If we cannot certify the optimality of this solution for~\eqref{clusteringproblem-mip} using Proposition~\ref{mip-lower-prop}, we add new constraints to~\eqref{l0clustering-mip-relaxed} based on the support of ${\B\beta}^{(2)}$ and solve the resulting problem. We repeat this process until an optimal solution to the original MIP~\eqref{clusteringproblem-mip} is found. This procedure is summarised in Algorithm~\ref{alg:rowgen}.
Problem~\eqref{l0clustering-mip-relaxed} has significantly 
fewer constraints compared to Problem~\eqref{clusteringproblem-mip}, especially if $\B\beta^{(t)}$s are sparse. 
Hence Problem~\eqref{l0clustering-mip-relaxed} is more desirable from a computational perspective. 
As we discuss in Section~\ref{sec:time-bench}, in our numerical experiments we can solve the MIP~\eqref{clusteringproblem-mip} to (near) global optimality using our row generation procedure for $q\approx250, p\approx 4500$ within 15 minutes using a laptop. We refer to Table~\ref{table:runtimebenchmark-exact} in the appendix for more details.

\begin{algorithm}
\caption{Row generation for \ourmethodlz~in Problem~\eqref{clusteringproblem-mip}}\label{alg:rowgen}
\begin{algorithmic}
    \Require Data matrix $\B X$, response vector $\B{y}$, $\lambda,\lambda_0\geq 0$, initial $\B\beta^{(1)} \in \R^p$
\State $T\gets 1$
\While{not converged}

\State Obtain $\B\beta^{(T+1)}$ by solving Problem~\eqref{l0clustering-mip-relaxed}

  \State $S_{T+1}\gets S(\B\beta^{(T+1)})$
  \If{ $S_{T+1}= S_t$ for some $t\in[T]$}\Comment{Optimal solution to~\eqref{clusteringproblem-mip} found}
  \State Break
  \EndIf
  \State $T \gets T+1$
  \EndWhile
\end{algorithmic}

\end{algorithm}

\subsection{Approximate Solutions for Problem~\eqref{clusteringproblem-reg}}\label{subsec:bcd}
From a practical perspective, it is often desirable to obtain good feasible solutions to our estimator quickly. They are useful for tuning hyperparameters, warm-starting a MIP solver, or when we prefer to obtain good solutions quickly over optimality certificates.
Here we present an approximate algorithm for Problem~\eqref{clusteringproblem-reg}.

Motivated by the recent work on $\ell_0$-sparse linear regression~\citep{hazimeh2020fast,hazimeh2023grouped}, we explore a Block Coordinate Descent (BCD) algorithm~\citep{tseng2001convergence}. SCOPE~\citep{scope} also use a BCD procedure. Due to the challenging presence of both the clustering and the sparsity penalties in the objective function, our BCD algorithm goes well beyond being a simple application, or extension, of the earlier work. We use problem structure to improve the efficiency of our algorithm. For example, as illustrated by the numerical experiments, 
 our approximate solver can be up to 500 times faster than the approximate solver of SCOPE in certain settings.\\
 Given $j_0\in [q]$, our BCD algorithm partially minimises the objective of Problem~\eqref{clusteringproblem-reg} with respect to the block of variables $\{\beta_i:i\in\cI_{j_0}\}$ while keeping other coefficients fixed. We cycle through these blocks of variables and update them. Formally, for a block of variables $\{\beta_i:i\in\cI_{j_0}\}$, the corresponding optimisation problem is given by
\begin{equation}\label{blockcd-1}
    \min_{\beta_i:i\in\cI_{j_0}} \tfrac{1}{n} \|\B{y}-\B{X\beta}-\alpha\B{1}\|_2^2 + \lambda_0 |\{i:\beta_i\neq 0, i\in\cI_{j_0}\}|+ \lambda |\{\beta_i: i\in\cI_{j_0}\}|.
\end{equation}
A key workhorse of our BCD algorithm is an efficient solver for Problem~\eqref{blockcd-1}. To this end, we discuss how Problem~\eqref{blockcd-1} can be reduced to a univariate problem with a single categorical predictor.

\subsubsection{Solving Problem~\eqref{blockcd-1}}
\noindent\textbf{Reduction to a univariate problem.}
Suppose $C_{j_0}$ for some $j_0\in[q]$ is a categorical predictor.  We can write the quadratic loss $ \|\B{y}-\B{X\beta}-\alpha\B{1}\|_2^2$ as
\begin{equation*}
    \Big\Vert \B{y}-\sum_{\substack{j\in[q] \\ j\neq j_0}}\sum_{i\in\cI_j} \B{X}_{i}\beta_i- \sum_{i=p-N+1}^p \B{X}_{i}\beta_i -\sum_{i\in\cI_{j_0}}\B{X}_{i}\beta_i-\alpha\B{1}\Big\Vert_2^2  = \sum_{i\in\cI_{j_0}}\sum_{k\in\cJ^{j_0}_i} \left(\tilde{y}_k-\beta_{i}\right)^2,
\end{equation*}
with $\cJ_i^{j}$ defined in~\eqref{J_k_j} and $\tilde{\B{y}}= \B{y}-\sum_{j\in[q], j\neq j_0}\sum_{i\in\cI_j} \B{X}_{i}\beta_i- \sum_{i=p-N+1}^p \B{X}_{i}\beta_i-\alpha\B{1}.$
As a result, Problem~\eqref{blockcd-1} can be equivalently written as 
\begin{align}\label{blockcd-2}
 \min_{\B{\beta} \in \R^{p_{j_0}} } \sum_{i=1}^{p_{j_0}}\sum_{k\in\cJ^{j_0}_i} \left(\tilde{y}_k-\beta_{i}\right)^2+ n\lambda_0 \|\B{\beta}\|_0 + n\lambda |\{\beta_1,\cdots,\beta_{p_{j_0}}\}|.
\end{align}
Furthermore, we note that Problem~\eqref{blockcd-2} is equivalent to the univariate problem
\begin{align}\label{univariate-1}
 \min_{\B{\beta} \in \R^{p_{j_0}} } \sum_{i=1}^{p_{j_0}} |\cJ^{j_0}_i| \bigg(\beta_i - \frac{1}{|\cJ^{j_0}_i|}\sum_{k\in\cJ^{j_0}_i}\tilde{y}_k\bigg)^2 + n\lambda_0 \|\B{\beta}\|_0 + n\lambda |\{\beta_1,\cdots,\beta_{p_{j_0}}\}|.
\end{align}

\noindent\textbf{Solving the univariate problem.}
The main building block for our BCD-based solver is Problem~\eqref{univariate-1}. Therefore, we consider a general form of this problem as
\begin{equation}\label{one-dim-general}
    \min_{\B\beta\in\R^p} \frac{1}{2} \sum_{i=1}^{p}n_i(\beta_i-\bar{y}_i)^2 + \tilde \lambda_0 \|\beta\|_0 + \tilde \lambda |\{\beta_1,\cdots,\beta_p\}|,
\end{equation}
for some $\bar{\B y}\in\R^p$, $\tilde{\lambda}_0,\tilde{\lambda}\geq 0$, and $n_i\geq 1$ for $i\in[p]$.
Without loss of generality, we assume $\bar{y}_1\geq\cdots\geq \bar{y}_p$, which implies  $\beta_1\geq\cdots\geq \beta_p$ at optimality. Hence, we rewrite~\eqref{one-dim-general} as the following discrete optimization problem:
\begin{align}\label{one-dim-jumps}
  \min_{\B\beta} \frac{1}{2} \sum_{i=1}^{p}n_i(\beta_i-\bar{y}_i)^2 + \tilde \lambda_0 \|\B\beta\|_0 + \tilde \lambda \sum_{i=1}^{p-1} \1(\beta_i \neq \beta_{i+1}). 
\end{align}
We design an algorithm that obtains globally optimal solutions to this problem. We denote the optimal solution to~\eqref{one-dim-jumps} by DpSegPen-$L_0(\B{y},\{n_i\}_i,\tilde{\lambda}_0,\tilde{\lambda})$.
Our algorithm for solving~\eqref{one-dim-jumps}, based on  Dynamic Programming (DP), 
extends the DP framework of~\citet{johnson2013dynamic} which was developed in the context of piecewise constant signal approximation [aka signal segmentation]. \citet{johnson2013dynamic} considers a special case of our problem with $\tilde \lambda_0 = 0$ and $n_i = 1$  for all $i \in [p]$. We provide the details of our DP framework in the Appendix~\ref{app:dp}. However, we note that due to the presence of the term $\tilde\lambda_0 \|\B\beta\|_0$ in Problem~\eqref{one-dim-jumps}, certain convexity structures used by~\citet{johnson2013dynamic} to speed up their method do not apply to our case. Thus, we need to develop a new approach that can efficiently deal with such non-convex structures.

Algorithm~\ref{alg:bcd} below summarises our BCD algorithm for Problem~\eqref{clusteringproblem-reg}.
\begin{algorithm}
\caption{\ourmethodlz~ block CD for Problem~\eqref{clusteringproblem-reg}}\label{alg:bcd}
\begin{algorithmic}
    \Require Data $\B X=[\B X_{\text{cat}},\B X_{\text{cont}}]\in\R^{n\times p}$ with continuous features $\B X_{\text{cont}}$ and $\B{y}\in\R^n$, $\lambda_0, \lambda \geq 0$, initial $\B\beta^0 \in \R^p$ 
\State $\B\beta \gets \B\beta^0$\;
\While{not converged}
  \For{$j = 1$ to $q$}
      \State  $\tilde{\B{y}} \gets \B{y}-\B{X\beta} + \sum_{i\in\cI_j} \B{X}_{i}\beta_i-\alpha\B{1}$\; 

           \State $\bar{y}_i\gets \frac{1}{|\cJ^{j}_i|}\sum_{k\in\cJ^{j}_i}\tilde{y}_k~~\forall~i\in\cI_{j}$\Comment{$\bar{\B{y}}\in\R^{|\cI_{j}|}$}

      \State  $\sigma \gets \text{SortDescending}(\bar{\B{y}})$ \Comment{$\bar{{y}}_{\sigma(1)} \geq \bar{{y}}_{\sigma(2)} \geq \ldots \geq \bar{{y}}_{\sigma(|\cI_{j_0}|)}$}
       \State $\left(\B\beta_{\cI_{j}}\right)_\sigma \gets \text{DpSegPen-$L_0$}(\bar{\B{y}}_\sigma,\{|\cJ^{j}_{\sigma(i)}|\}_i,n\lambda_0/2,n\lambda/2)$\Comment{Solve~\eqref{one-dim-jumps} and update the subvector of $\B\beta$}
    
    \EndFor
   
      \For{$j = 1$ to $N$}\Comment{Update the continuous predictors via coordinate descent}
      \State $\beta_{p-N+j}\gets \argmin_{b}\|\B{y}-\alpha\B{1}-\sum_{i\neq p-N+j}\B{X}_i\beta_i-\B{X}_{p-N+j}b\|_2^2+n\lambda_0\1(b\neq 0)$
      \EndFor
     
\State $\alpha \gets \argmin_{a\in\R} \|\B y - \B{X\beta} - a \B{1}\|_2^2$
 \EndWhile
\end{algorithmic}

\end{algorithm}

\textbf{Active Sets:} To further speed up the convergence of our BCD algorithm and anticipating a sparse solution to Problem~\eqref{clusteringproblem-reg}, we make use of active sets. Starting with an initial active set $\cA\subseteq[p]$, we apply Algorithm~\ref{alg:bcd} with the additional constraint $\beta_i=0$, $\forall i\in\cA,$
which effectively reduces the problem size from $p$ to $|\cA|$. After Algorithm~\ref{alg:bcd} converges and returns a solution, we run an iteration of Algorithm~\ref{alg:bcd} on $[p]$, i.e.,  without an active set, and update~$\cA$ as follows: $\cA\gets \cA\cup\{i\in[p]:\beta_i\neq 0\}.$
We repeat this procedure until the active set stabilises.

\subsubsection{Extension to Binary Classification}\label{sec:logistic} 
We note that our BCD algorithm discussed above can also be applied to binary classification tasks with, for example, logistic loss (cf~\citet{10.5555/3546258.3546393} for $\ell_0$-penalised classification problems). 
 Suppose that we observe $y^{(i)}\in\{-1,1\},\B{x}^{(i)}\in\R^p$ for $i\in[n]$, where $\B{x}^{(i)}$ contains the values of the~$p$ dummy variables for the $i$-th individual. The logistic loss for the $i$-th datapoint is $L(y^{(i)},\B{x}^{(i)};\B\beta,\alpha) = \log(1+\exp(-y^{(i)}[\B\beta^\top\B{x}^{(i)}+\alpha]))$. 
Under this loss function, we consider the estimator
\begin{equation}\label{logstic-reg}
    \min_{\B\beta\in\R^p,\alpha\in\R} \frac{1}{n}\sum_{i=1}^n L(y^{(i)},\B x^{(i)};\B\beta,\alpha) + \lambda_0 \|\B{\beta}\|_0 + \lambda \sum_{j=1}^{q} \left\vert\left\{\beta_i:i\in\cI_j\right\}\right\vert,
\end{equation}
which is similar to~\eqref{clusteringproblem-reg}, but uses the logistic loss. The following result provides a quadratic upper bound to the logistic loss function around $\B\beta_0$.

\begin{lemma}\label{classification-lemma}
Suppose $\B\beta,\B\beta_0\in\R^p$ and $\alpha,\alpha_0\in\R$. Let $\tilde{\B{y}}=\alpha_0\B{1} + \B{X}\B{\beta}_0-4\B{g}$ with $$g_i =\frac{-y^{(i)}\exp(-y^{(i)} [\B\beta_0^\top\B{x}^{(i)}+\alpha_0])}{1+\exp(-y^{(i)} [\B\beta_0^{\top}\B{x}^{(i)}+\alpha_0])}~\forall i\in[n].$$
Then, $ \sum_{i=1}^n \{L(y^{(i)},\B{x}^{(i)};\B\beta,\alpha)-L(y^{(i)},\B{x}^{(i)};\B\beta_0,\alpha_0)\}  \leq \frac{1}{8}\|\tilde{\B{y}}-\B{X}\B\beta-\alpha\B{1}\|_2^2 - 2\|\B g\|_2^2$. 
\end{lemma}
It follows from Lemma \ref{classification-lemma}  that the objective in~(\ref{logstic-reg}) can be upper-bounded by $\tilde{P}(\B\beta_0,\alpha_0)+P(\B\beta,\B\beta_0,\alpha,\alpha_0)$, 
where $\tilde{P}(\B\beta_0,\alpha_0)=(1/n)\sum_{i=1}^n L(y^{(i)},\B{x}^{(i)};\B\beta_0,\alpha_0)$ and
$P(\B\beta,\B\beta_0,\alpha,\alpha_0) = \|\tilde{\B y}-\B{X\beta}-\alpha\B{1}\|_2^2/(8n)+\lambda_0 \|\B{\beta}\|_0 + \lambda \sum_{j=1}^{q} \left\vert\left\{\beta_i:i\in\cI_j\right\}\right\vert$. 
We extend our BCD algorithm, developed in Section~\ref{subsec:bcd}, to the binary classification setting by cyclically minimising our upper bound with respect to the block of coordinates $\{\beta_i:i\in\cI_{j}\}$ for $j\in \{1,\cdots,q\}$ and $\{\beta_i: p-N+1\leq i\leq p\}$, which are the coefficients corresponding to continuous features. Our approach for obtaining good feasible solutions to Problem~\eqref{logstic-reg} is summarised in Algorithm~\ref{alg:bcd_classification} of the Appendix.

\section{Numerical Experiments}\label{sec:numerical}
We investigate the performance of our estimator on several synthetic and real-world datasets. The code for implementing \ourmethodlz~can be accessed at \href{https://github.com/mazumder-lab/ClusterLearn}{https://github.com/mazumder-lab/ClusterLearn}.
With the exception of Section~\ref{sec:time-bench}, we use the approximate solver for our estimator in all of the experiments. In Appendix~\ref{app:numerical}, we present additional numerical experiments, as well as further details of our experimental setup and some resources.

\subsection{Synthetic data}
\subsubsection{Experimental Setup}
We start by generating a series of observations with categorical features $C_1,\ldots,C_q$. To this end, following~\citet{scope}, we draw $\zeta_1,\cdots,\zeta_q$ from a zero-mean $q$-dimensional normal distribution with pairwise correlation of $\rho\geq 0$. Throughout this section, we set $\rho=0.2$. Then, we pass $\zeta_1,\cdots,\zeta_q$ through the normal CDF to obtain $q$ uniformly distributed random variables $\tilde\zeta_1,\cdots,\tilde\zeta_q\sim\text{Unif}[0,1]$. The correlation among $\tilde\zeta_1,\cdots,\tilde\zeta_q$ is controlled by~$\rho$. We bin $\tilde\zeta_1,\cdots,\tilde\zeta_q$ into discrete values to obtain the categorical features $C_1,\cdots,C_q$.
We generate $3n$ different observations, which we split equally into a training set used to estimate the coefficients, a validation set used to find the best hyperparameters, and a test set used to evaluate the performance (the details of hyperparameter tuning can be found in Appendix~\ref{app:numerical}). We generate the response vector $\B{y}$ as $\B y = \B X\B\beta^* + \B\epsilon$ where $\B X$ is as defined in Section~\ref{prob-form-sec}. The noise vector $\B\epsilon\sim\cN(\B 0,\sigma^2 \B I)$ is generated independently of~$\B X$. We define $\text{SNR}=\|\B{X\beta}^*\|_2^2/\|\B\epsilon\|_2^2$.

\subsubsection{Alternative Methods and Metrics}
We compare the performance of our estimators with the following approaches:\\
    \noindent\textbf{(1)} Elastic Net~\citep{elasticnet}, which imposes the penalty $\lambda_1 \|\B\beta\|_1 + \lambda_2\|\B\beta\|_2^2$. Lasso is a special case of Elastic Net with $\lambda_2=0$. We use the R package \texttt{glmnet}.\\
    \noindent\textbf{(2)}  Iterative Hard Thresholding \citep[IHT,][]{iht0}, which imposes the penalty $\lambda_0\|\B\beta\|_0$. We use our implementation of IHT.\\
    \noindent\textbf{(3)}  The state-of-the-art approach SCOPE~\citep{scope}  which performs clustering of levels of categorical variables\footnote{They use a minimax concave penalty on the differences between the order statistics of the coefficients for a categorical variable.}.
    We use the R package \texttt{CatReg}.

The details of hyperparameter tuning can be found in Appendix~\ref{app:numerical}. To measure the prediction performance of each estimator, we use the coefficient of determination ($R^2$) computed on the test set. To quantify the cluster recovery performance of the estimators, we report 
their clustering purity~\citep{puritydef}, computed for the categorical predictors with nonzero true regression coefficients\footnote{For the predictors with zero true coefficients, the purity is always one by definition.}. Specifically, we let 
$A = \{j\in[q]\;\text{s.t.}\; \bbeta^*_{\cI_j}\neq \B 0\}\quad \text{and}\quad \nu = \sum_{j\in G}|\cI_j|$, where $A$ indexes the ``true'' predictors and~$\nu$ counts the total number of regression coefficients for the predictors in $A$, and define purity as:
$$\text{Purity}(\B\beta)=\frac{1}{\nu}\sum_{k\in A} \; \sum_{a\in\{\beta_j:\,j\in\cI_k\}} \max_{b\in\{\beta_j^*:\,j\in\cI_k\}}\big|\{l\in\cI_k: {\beta}_l=a,\,\beta^*_l=b\}\big|.$$
Purity is complementary to our impurity measure $\mathcal{E}(\mathcal{G})$, introduced in Section~\ref{sec.clust.th.prelim}.
Larger values of purity generally indicate better clustering performance. When the number of clusters is correctly estimated, the maximal purity of~$1$ corresponds to perfect cluster recovery. In addition to purity, we report the total number of the estimated regression coefficient clusters,
averaging the results over 50 replications.

\subsubsection{Statistical Comparisons}\label{results-set1}
We explore the statistical properties of various estimators. For given values of $q,q_s,r_1,r_2>0$, we let
\begin{equation}\label{eq:beta_star_setting_2}
    \B\beta^{*}_{\cI_j} = \left\{
    \begin{array}{ll}
        (\overbrace{-2,\ldots ,-2}^{r_1\text{ times}},\overbrace{0,\ldots ,0}^{ r_2 \text{ times}},\overbrace{2,\ldots ,2}^{r_1\text{ times}} ) & \mbox{if } j \in \{1,\cdots,q_s\} \vspace{0.2cm}\\
        (\overbrace{0,\ldots ,0}^{2r_1+r_2\text{ times}}) & \mbox{otherwise.}
    \end{array}
\right.
\end{equation}
First, we let $r_1=4,r_2=12, q=20, q_s=3$. This simulates a case where each categorical predictor has a large zero cluster. We also set $\sigma=2$ ($\text{SNR}=1.23$ on average). The results for this case are shown in Figure~\ref{fig:set1-n} for several values of $n$. We do not include Elastic Net in the two rightmost panels of Figure~\ref{fig:set1-n} as it produces significantly more clusters (around 133 on average) than the other methods.

We observe that \ourmethodlz~works better than SCOPE and \ourmethod~in terms of $R^2$ for a wide range of $100\leq n\leq 400$. For larger values of $n$, all methods seem to have a similar prediction performance. \ourmethodlz~also shows a good clustering performance with a high purity value and a small number of clusters.
Interestingly,  \ourmethodlz~does better than \ourmethod~(Problem~\eqref{clusteringproblem-reg} with $\lambda_0=0$), illustrating the benefits of $\ell_0$ penalisation.

Next, we set $r_1=r_2=10,q=20,q_s=5$ and $\sigma=2$ (average
$\text{SNR}=3.70$). This simulates a setting where nonzero predictors do not have a dominating zero cluster. Similarly, we show the results over a wide range of $n$ for this setup in Figure~\ref{fig:set2-n}. On average, Elastic Net produced 377 clusters, so we omitted it in the purity and the number of clusters plots. We see that \ourmethodlz~and \ourmethod~seem to outperform other methods in terms of $R^2$ when $n$ is smaller. In particular, \ourmethodlz~performs well when $200\leq n\leq 500$. For larger values of $n$, SCOPE performs better in terms of $R^2$. However, we see that in terms of clustering, \ourmethodlz~has the best performance, followed by \ourmethod.

In addition to the experiments presented here, our  experiments in Appendix~\ref{supp:synthetic} investigate the effect of changing SNR and consider other data generation settings. We observe that \ourmethodlz~can perform well for a wide range of SNR values. We also present an example of a data generation setting where \ourmethod~can outperform \ourmethodlz.
In summary, \ourmethodlz~(or \ourmethod) mostly appears to display excellent overall performance in our synthetic data experiments, demonstrating the benefits of our proposed estimator.

\begin{figure}[t!]
\setlength{\tabcolsep}{1pt}
     \centering
    \scalebox{1.1}{\begin{tabular}{ccc}
      Test $R^2$ & Purity & \# of Clusters \\
       \includegraphics[width=0.28\linewidth]{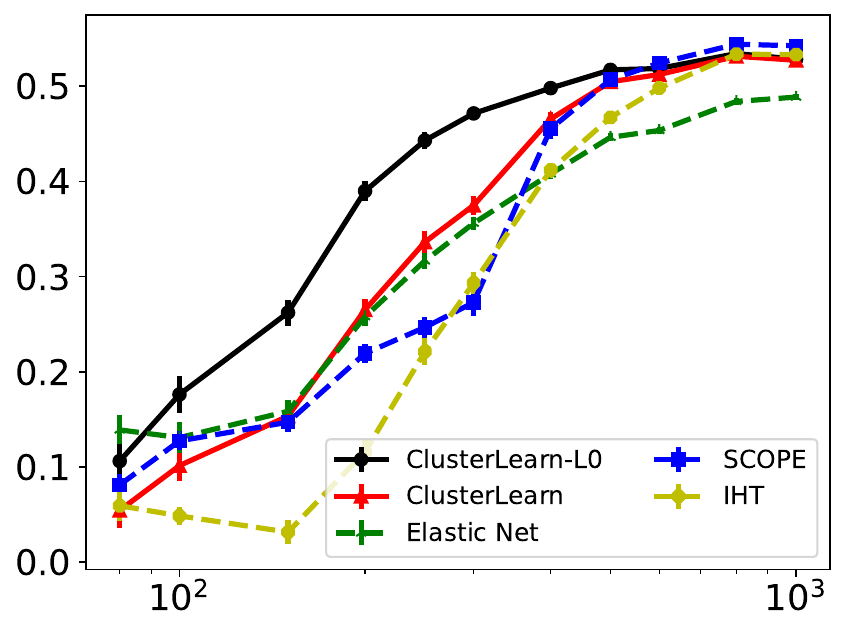}  &  \includegraphics[width=0.28\linewidth]{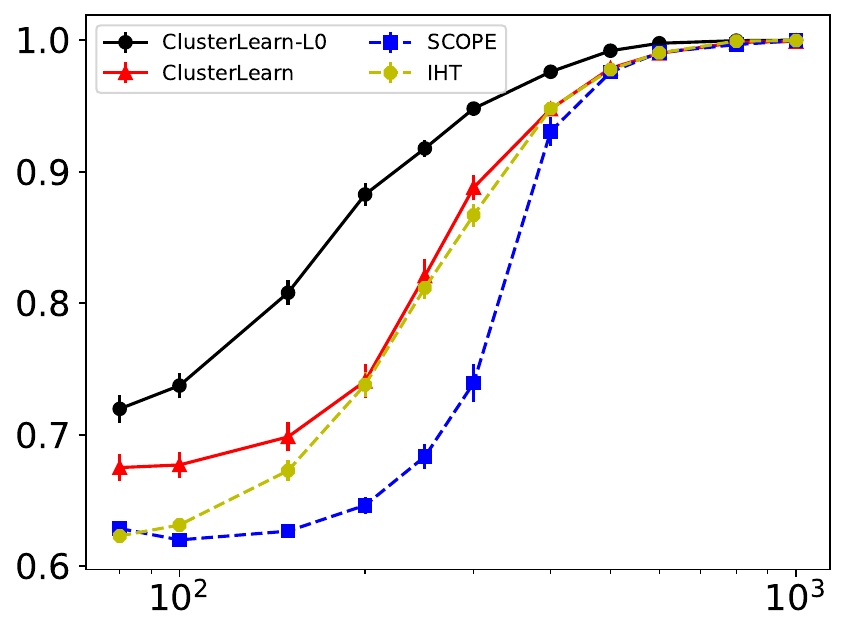} & 
       \includegraphics[width=0.28\linewidth]{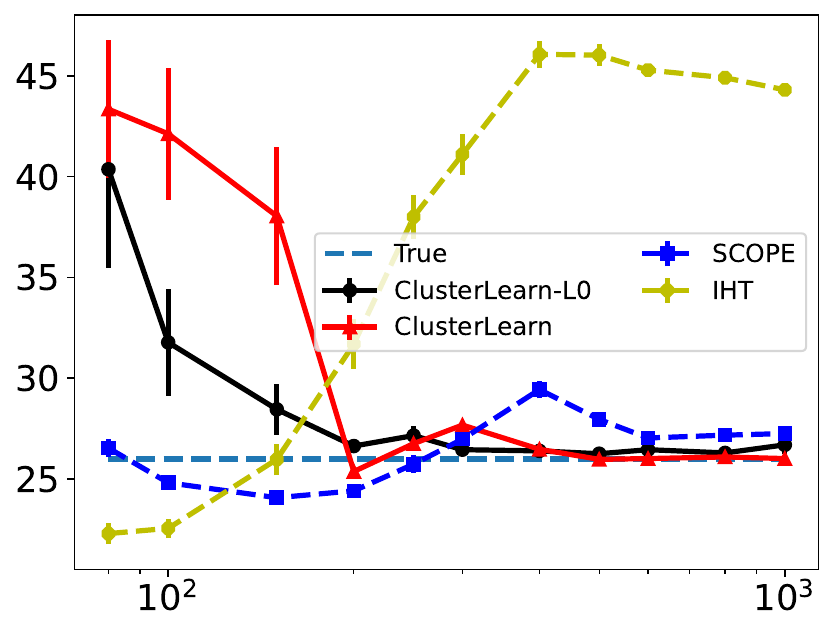} \\
       $n$ & $n$ & $n$
    \end{tabular}}
    \captionof{figure}{Experiments from Section~\ref{results-set1} with $r_1=4,r_2=12, q=20, q_s=3$ and $\rho=0.2,\sigma=2$. We do not plot the purity and the number of clusters for Elastic Net as it produces a large number of clusters. The vertical bars at each point indicate the corresponding standard errors.}
        \label{fig:set1-n}
\end{figure}

\begin{figure}[t!]
 \setlength{\tabcolsep}{1pt}
    \centering
    \scalebox{1.1}{\begin{tabular}{ccc}
      Test $R^2$ & Purity & \# of Clusters \\
       \includegraphics[width=0.28\linewidth]{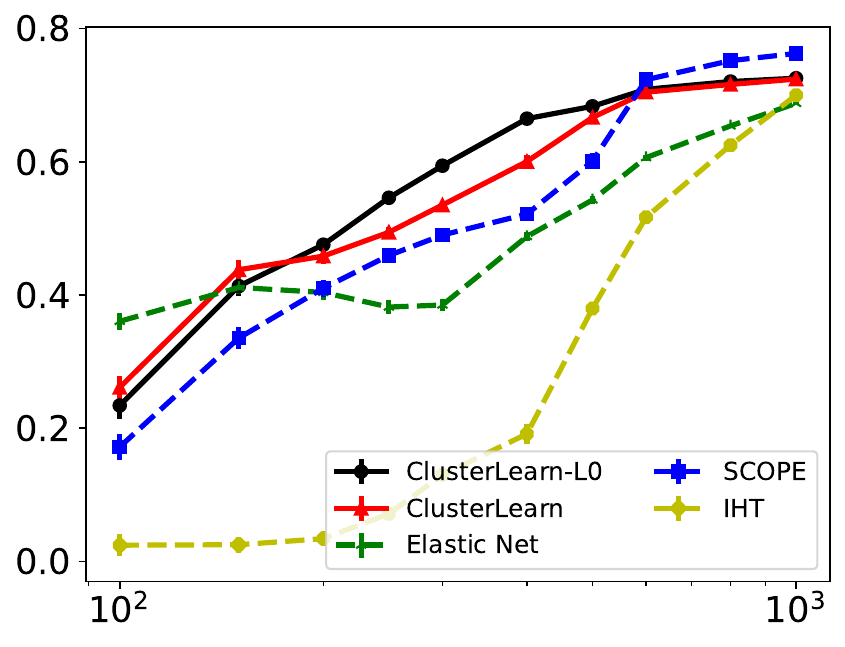}  &  \includegraphics[width=0.28\linewidth]{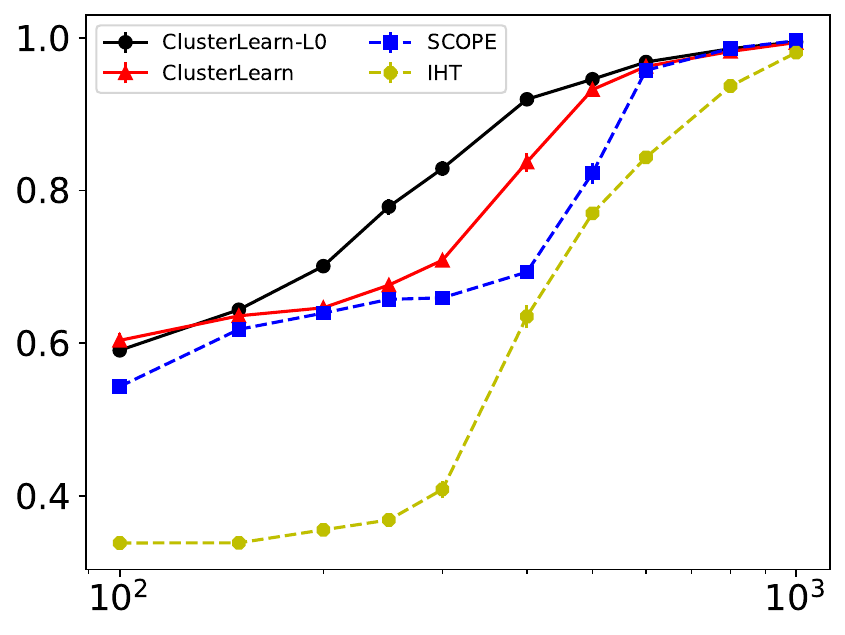} &
       \includegraphics[width=0.28\linewidth]{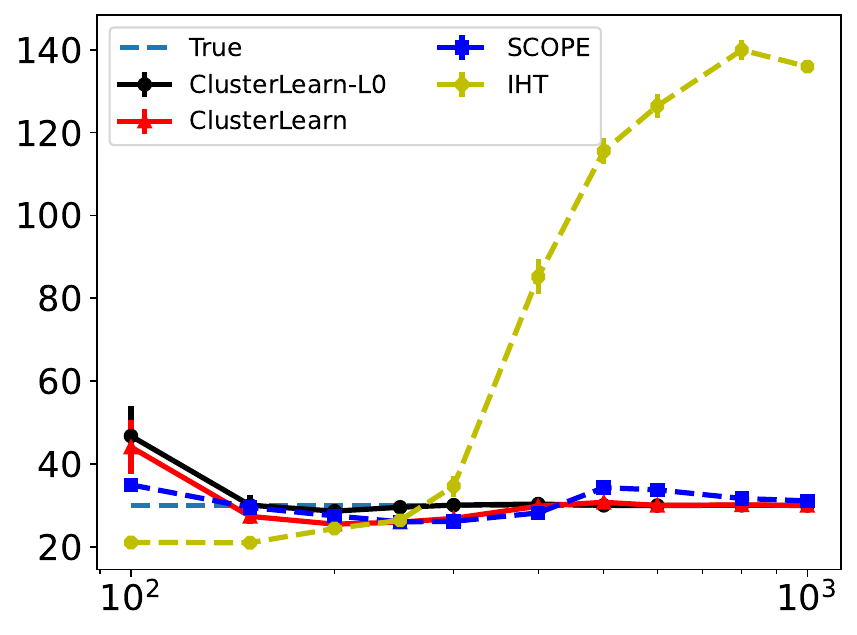} \\
       $n$ & $n$ & $n$
    \end{tabular}}
    \captionof{figure}{Experiments from Section~\ref{results-set1} with $r_1=r_2=10,q=20,q_s=5, \rho=0.2$ and $\sigma=2$ with varying $n$. We do not plot purity and the number of clusters for Elastic Net as it produces a large number of clusters. The vertical bars at each point indicate the corresponding standard errors. }
    \label{fig:set2-n}

\end{figure}

\subsubsection{Algorithm Runtime Comparisons}\label{sec:time-bench}
We study the computational performance of our estimator with respect to runtime and scalability. To this end, in~\eqref{eq:beta_star_setting_2} we set $r_1=2$ and consider varying values of $r_2>1$, $q>1$ and $q_s>1$.
 We also set $\sigma=1$ in this section.
 
 \noindent\textbf{Approximate Algorithms:} We first compare our approximate algorithm from Section~\ref{subsec:bcd} with the state-of-the-art (approximate) solver for SCOPE. We report the runtime of fitting a path of solutions for 100 different tuning hyperparameters (for \ourmethodlz,~this corresponds to a $10\times 10$ grid of $\lambda_0,\lambda$-values). For these experiments, we set $q_s=2,\sigma=1$ and $n=2000$.
 Table~\ref{table:runtimebenchmark} in the Appendix 
 summarises the results for different values of~$q$ and~$r_2$. We observe that for large problems with~$r_2$ in the hundreds, \ourmethod~is up to 500 times faster than SCOPE, and \ourmethodlz~(which includes an additional sparsity penalisation)~is about 100 times faster than SCOPE. These results, as well as the results in Table~\ref{table:runtimebenchmark-approx2}, demonstrate the superior computational efficiency/scalability of our approximate algorithms on large problems with numerous categorical levels.

\noindent\textbf{Exact Algorithms:} We study exact MIP algorithms for our estimator --- see Section~\ref{numerical-results-runtime} in the Appendix for the results. We compare our custom row generation framework discussed in Section~\ref{subsec:rowgen} vs Gurobi's off-the-shelf solver. (Recall that the SCOPE method does not discuss exact algorithms for multiple predictors). 
As expected, the exact algorithms are slower than the approximate solvers due to their focus on certifying global optimality (via branch-and-bound), but they are able to address reasonably large problems with $n \approx 5000$ and $p\approx4500$ within 15 minutes or so (we refer to Appendix~\ref{app:numerical} and Table~\ref{table:runtimebenchmark-exact} for further results and details). We observe that our specialised row-generation algorithm is promising and can be faster than Gurobi's off-the-shelf solvers in certifying optimality. 

Additional experiments in Appendix~\ref{numerical-results-runtime} demonstrate that in terms of optimisation, our exact and approximate solvers achieve close objective values, confirming that our approximate solver can obtain high-quality solutions. We also illustrate the scalability of our specialized solver DpSegPen-$L_0$ in Table~\ref{table:runtimebenchmark-oneidm}.

\subsection{Real Datasets}\label{sec:realdata}
\subsubsection{Bike Sharing Dataset}\label{bikedata}
We consider the bike sharing dataset~\citep{misc_bike_sharing_275} from the UCI data repository, where the goal is to predict the number of bike rentals in a given time period (we use the hourly version of the dataset). This is a regression dataset with 8 categorical and 4 continuous features.  We take $n=100$ training data points, 100 validation data points, and use the rest of the data as a test set. This leads to a problem with about $p=60$ regression coefficients (the exact number varies slightly depending on the train/test split after we remove the unused categorical levels). In Table~\ref{tab:bike}, we report the test~$R^2$, the runtime of each method, and the total number of regression coefficient levels, $\sum_{j\in[q]} |\{\beta_i:i\in\cI_j\}|$. We observe that \ourmethodlz~has the best test $R^2$, improving upon both \ourmethod~and SCOPE. \ourmethodlz~and \ourmethod~also have fewer coefficient levels  than SCOPE. In terms of runtime, our methods are competitive with the alternatives.

\begin{table}[]
    \centering
    \scriptsize
        \caption{Mean performance of different algorithms on the bike sharing dataset in Section~\ref{bikedata}, averaged across 100 runs. We report the average and the standard error (in parenthesis).}
    \label{tab:bike}
    \begin{tabular}{cccc}
\toprule
  Method & Test $R^2$ & Number of levels & Runtime (seconds) \\
\midrule
		
IHT &  $0.267(\pm 0.078)$ & $16.1(\pm 1.1)$  & $0.41(\pm 0.02)$  \\
Elastic Net & $0.437(\pm	0.007)$	& $27.5(\pm	0.9)$     &  $\mathbf{0.22(\pm0.00)}$ \\
SCOPE & $0.472(\pm 0.005)$ &	$48.6(\pm	0.6)$    & $1.07(\pm0.02)$  \\
\ourmethod & $0.514(\pm	0.007)$ &	$\mathbf{15.4 (\pm	0.5)}$    &  $0.48 (\pm 0.01)$ \\
\ourmethodlz &  $\mathbf{0.532 (\pm0.004)}$ &	$17.2 (\pm 0.7)$   & $1.44(\pm0.01)$  \\
\bottomrule
\end{tabular}

\end{table}

\subsubsection{Insurance Dataset}\label{sec:insurance}
We consider the Prudential Life Insurance Assessment
dataset~\citep{prudential-life-insurance-assessment}, where the goal is to predict the risk level of life insurance applicants. The original dataset's response is an ordinal measure of risk that has~8 levels. To adapt the data to our context, we generate a new binary response indicating whether the original risk measure falls into the highest 4 risk levels.

We use $30,400$ observations from the original data, and use $n=1000$ observations as a training set, 2000 as a validation set, and the rest as a test set. We utilise the full set of 13 continuous and 60 categorical features but drop any feature with more than $80\%$ missing values, which results in the model dimension of $p=723$. Similar to the previous dataset, we compare with SCOPE, Elastic Net, and Lasso. For SCOPE, we use the concavity parameter of 100 that is recommended for this data in~\cite{scope}. As illustrated by Table~\ref{tab:insurance}, \ourmethod~and \ourmethodlz~achieve the best prediction performance.

\begin{table}[]
    \centering
    \scriptsize
        \caption{Mean performance of different algorithms on the insurance dataset in Section~\ref{sec:insurance}, averaged across 100 runs.  We report the average and the standard error (in parenthesis).}
    \label{tab:insurance}
    \begin{tabular}{cccc}
\toprule
  Method & Accuracy & Number of levels & Runtime (seconds) \\
\midrule
Lasso & $0.6800 (\pm0.0006)$ & $184.0 (\pm6.4)$ & $2.8 (\pm0.1)$ \\
Elastic Net & $0.6877(\pm0.0005)$ & $296.5(\pm0.7)$ & $2.6(\pm0.1)$\\
SCOPE-100 & $0.6846 (\pm0.0008)$ & $\mathbf{82.9 (\pm1.1)}$ & $281.2 (\pm11.0)$ \\
\ourmethod & $\mathbf{0.6906 (\pm0.0004)}$ & $143.7 (\pm2.9)$ & $\mathbf{2.4 (\pm0.1)}$ \\
\ourmethodlz & $\mathbf{0.6908 (\pm0.0004)}$ & $122.1 (\pm2.7)$ & $5.5 (\pm0.1)$ \\
\bottomrule
\end{tabular}

\end{table}

\section*{Acknowledgements}
Kayhan Behdin contributed to this work while he was a PhD student at MIT. Riade Benbaki contributed to this work while he was a Master student at MIT.
Authors would like to thank MIT SuperCloud~\citep{reuther2018interactive} for providing the computational resources that made this work possible. Our research was funded by grants from the Office of Naval Research (ONR N000142512504).

\bibliographystyle{plainnat}
\bibliography{ref.bib}

\clearpage
\appendix
\numberwithin{equation}{section}
\numberwithin{lemma}{section}
\numberwithin{proposition}{section}
\numberwithin{figure}{section}
\numberwithin{table}{section}
\numberwithin{algorithm}{section}

\section*{Additional Numerical Experiments, Proofs and Technical Details}

\section{Details of DpSegPen-$L_0$}\label{app:dp}

We reformulate Problem~\eqref{one-dim-jumps} as a maximisation problem:

\begin{align}\label{one-dim-jumps-max}
     \max_{\B\beta\in\R^p} \sum_{i=1}^{p}e_i(\beta_i) - \tilde \lambda \sum_{i=1}^{p-1} \1(\beta_i \neq \beta_{i+1}),
\end{align}
where $e_i(x) = {-n_i}(x-\bar{y}_i)^2/2 + \tilde \lambda_0 \1(x = 0)$. We note that~\citet{johnson2013dynamic} considers a special case of~\eqref{one-dim-jumps-max} with $\tilde{\lambda}_0 = 0$. Following derivations in~\citet{johnson2013dynamic}
Problem~\eqref{one-dim-jumps-max} can be solved using the following recursion:
\begin{equation}\label{dp-recursions}
    \begin{aligned}
        \delta_1(b) &= e_1(b) \\
        b^*_k &= \argmax_{\tilde{b}} \delta_{k-1}(\tilde{b})\quad \forall k \in \{2,\ldots,p\}\\
        f_k(b) &= \max_{\tilde b} \delta_{k-1}(\tilde b) - \tilde \lambda \1(b \neq \tilde{b}) = \max \left\{\delta_{k-1}(b),\delta_{k-1}(b^*_k) -\tilde \lambda \right\}~~\forall k \in \{2,\ldots,p\}\\
    \delta_k(b)& = e_k(b) + f_k(b) \quad \forall k \in \{2,\ldots,p\}.
    \end{aligned}
\end{equation}
Once the functions $f_k$ and $\delta_k$ in~\eqref{dp-recursions} are computed, a solution $\hat{\B\beta}$ of~\eqref{one-dim-jumps-max} can be obtained by:
\begin{equation}\label{dp-recursion-2}
\begin{aligned}
    \hat{\beta}_p &= \argmax_b \delta_p(b)\\
    \hat{\beta}_k &= \argmax_b \delta_k(b) - \tilde \lambda \1(b \neq \hat{\beta}_{k+1}).
\end{aligned}    
\end{equation}
We refer to~\citet{johnson2013dynamic} for the derivation of~\eqref{dp-recursion-2}.
However, we note that as $\tilde\lambda_0=0$ in the work of~\citet{johnson2013dynamic}, the functions $e_i$ are concave quadratics, while in our case, the functions $e_i$ are non-concave, which shows that their algorithm does not directly extend to our problem. Nevertheless, we discuss that~\eqref{dp-recursions} and~\eqref{dp-recursion-2} can still be efficiently calculated in our case. To this end, let us define the set of piece-wise quadratic functions as
\begin{equation}
    \cQ = \bigg\{\sum_{i=1}^{m_1+1} \1(x \in (x^{(1)}_{i-1},x^{(1)}_i]) \left( a_i x^2 + b_i x + c_i \right):-
    \infty=x_0^{(1)}\leq x_1^{(1)}\leq\cdots\leq x_{m_1+1}^{(1)}=\infty,m_1\geq 0\bigg\},
\end{equation}
and the set of positive ``spike'' functions as
\begin{equation}\label{csdef}
    \cS = \left\{\sum_{i=1}^{m_2}s_i\1(x=x^{(2)}_i):x_i^{(2)}\in\R,~s_i>0~\forall i\in[m_2], m_2\geq 0\right\}.
\end{equation}
We define the set of functions that can be decomposed into a sum of a piece-wise quadratic function and a spike function as
\begin{equation}
    \cF = \cQ + \cS
\end{equation}
where $A+B=\{a+b:a\in A, b\in B\}$ denotes the Minkowski sum. Our first result in this section is to show that the functions appearing in~\eqref{dp-recursions} belong to~$\cF$.
\begin{lemma}\label{opt-lemma}
    For every $k \geq 1$, functions $f_k,\delta_k$ defined in~\eqref{dp-recursions} belong to~$\cF$.
\end{lemma}
\begin{proof}[Proof of Lemma~\ref{opt-lemma}]
    We prove the lemma by induction. To this end, note that $\delta_1=e_1\in\cF$.
    We next show the following two steps: (1) For $2\leq k\leq n$, assuming $f_k\in\cF$  we have $\delta_k\in\cF$, (2) For $2\leq k\leq n$, assuming $\delta_{k-1}\in\cF$, we have $f_k\in\cF$.
    These two steps together complete the induction argument.

    \textbf{Proof of Step (1)}: Note that $\cF$ is closed under addition, $\cF=\cF+\cF$. 
    Since $e_k\in\cF$, and $f_k\in\cF$ (assumed for this step), we have that $\delta_k=e_k+f_k$ lies in $\cF$. 
    
    \textbf{Proof of Step (2)}: Let $\delta_{k-1}=Q+S$ be the piece-wise quadratic and spike decomposition of $\delta_{k-1}\in\cF$ with $Q\in\cQ,S\in\cS$. Define 
    $$\tilde{Q}(b)=\max\{Q(b),\delta_{k-1}(b_k^*)-\tilde\lambda\},~~\forall b\in\R.$$
    As discussed by~\citet{johnson2013dynamic}, $\tilde{Q}$ has a piece-wise quadratic representation. Moreover, for $b\in\R$ such that $S(b)=0$,
    $$\max\{\delta_{k-1}(b),\delta_{k-1}(b_k^*)-\tilde\lambda\}=\tilde{Q}(b).$$
    Therefore, for such $b$, $f_k(b)=\tilde{Q}(b)$. Hence, for at most finitely many points, $f_k-\tilde{Q}$ can be nonzero which shows $f_k-\tilde{Q}$ is a spike function. We still need to show that the spikes of $f_k-\tilde{Q}$ are positive to demonstrate $f_k-\tilde{Q}\in\cS$ defined in~\eqref{csdef}. However, for any $b\in\R$ such that $S(b)>0$,
    $$f_k(b) \stackrel{(a)}{=} \max\{\delta_{k-1}(b),\delta_{k-1}(b_k^*)-\tilde\lambda\} \stackrel{(b)}{\geq} \max\{Q(b),\delta_{k-1}(b_k^*)-\tilde\lambda\}=\tilde{Q}(b)$$
    where $(a)$ is by the definition of $f_k$ in~\eqref{dp-recursions} and $(b)$ is true as for $b\in\R$, $\delta_{k-1}(b)=Q(b)+S(b)\geq Q(b)$.
    This completes the proof.
\end{proof}
Given Lemma~\ref{opt-lemma}, it remains to show that we can find the maximum of functions from $\cF$ efficiently---this is required in~\eqref{dp-recursions} (for $b_k^*$) and~\eqref{dp-recursion-2} (for $\hat{\beta}_k$).

Let $f=Q+S\in\cF$ and let $\{x_1^{(2)},\cdots,x_{m_2}^{(2)}\}$ be the set of spikes for $S$, as defined in~\eqref{csdef}. Note that
\begin{equation}\label{max-twopart}
    \max_x f(x) = \max_{x\in\{x_i^{(2)}\}} f(x) \lor \sup_{x\notin\{x_i^{(2)}\}} Q(x).
\end{equation}
We note that the maximum in~\eqref{max-twopart} exists and is well-defined, as the spikes of function in $\cS$ are always positive.
The second optimisation in rhs of~\eqref{max-twopart} requires maximising a piece-wise quadratic function, which can be done as described by~\citet{johnson2013dynamic}. The feasible set of the first problem in the rhs of~\eqref{max-twopart} is a finite set, so $\max_{x\in\{x_i^{(2)}\}} f(x)$ can be computed by calculating the function value $f(x)$ for finitely many values of $x$. 
Based on the above discussion, an optimal solution to Problem~\eqref{max-twopart} and hence  steps~\eqref{dp-recursions} and~\eqref{dp-recursion-2} can be computed efficiently---this results in our algorithm DpSegPen-$L_0$.

\section{Proof of Theorem~\ref{estimation_thm}}
We present the proof for the special case where $\hat\alpha=\alpha=0$. The general proof follows by repeating the same steps after appending a column of all ones to~$\B X$ and appending the intercept~$\alpha$ to~$\B\beta$. The proof of Theorem~\ref{estimation_thm} is based on the following lemma.
\begin{lemma}\label{thm1-lem1}
For any $t>0$ and $p\geq s\geq 0$, define 
\begin{equation}
    E_1(t,s)=\left\{\frac{4}{n}\sup_{\|\B\beta\|_0=s}\left(\B{\epsilon}^\top\frac{\B{X}(\hat{\B{\beta}}-\B{\beta})}{\|\B{X}(\hat{\B{\beta}}-\B{\beta})\|_2}\right)^2 -\lambda_0 \|\hat{\B\beta}\|_0 >\frac{t}{n}\right\}.
\end{equation}

Under the conditions of Theorem~\ref{estimation_thm}, 
\begin{equation}
     \p\left( E_1(t,s)\right)
         \stackrel{}{\leq}   4\exp\left(-\frac{t}{a\sigma^2} + 2s\log(3ep)\right)
\end{equation}   
for some universal constant $a>0$. 
\end{lemma}
\begin{proof}[Proof of Lemma~\ref{thm1-lem1}]
   Fix $\B\beta$ such that $\|\bbeta\|_0=s$ and let $\hat{S}=\{i\in[p]:\hat{\beta}_i\neq 0\}$ and $S(\bbeta)=\{i\in[p]:\beta\neq 0\}$. Moreover, for $S\subseteq[p]$, 
 let $\B\Phi_{S}\in\R^{n\times |S|}$ be an orthonormal basis for the column span of $\B X_{S}$, the submatrix of $\B X$ with columns indexed by $S$. Then,
\begin{align}
\frac{4}{n} \left(\B{\epsilon}^\top\frac{\B{X}(\hat{\B{\beta}}-\B{\beta})}{\|\B{X}(\hat{\B{\beta}}-\B{\beta})\|_2}\right)^2 -\lambda_0 \|\hat{\B\beta}\|_0
    & \leq \frac{4}{n}\sup_{\B{v}\in\R^{|\hat{S}\cup S(\bbeta)|}}\left(\B{\epsilon}^\top\frac{\B X_{\hat{S}\cup S(\bbeta)} \B{v}}{\|\B X_{\hat{S}\cup S(\bbeta)} \B{v}\|_2}\right)^2 -\lambda_0 \|\hat{\B\beta}\|_0\nonumber \\
    & = \frac{4}{n} \sup_{\substack {\B{v}\in\R^{|\hat{S}\cup S(\bbeta)|} \\ \|\B{v}\|_2=1}}\left(\B{\epsilon}^\top\B{\Phi}_{\hat{S}\cup S(\bbeta)}\B{v}\right)^2 -\lambda_0 |\hat{S}|\nonumber \\
     & \leq \frac{4}{n} \sup_{\substack {\B{v}\in\R^{|\hat{S}\cup S(\bbeta)|} \\ \|\B{v}\|_2=1}}\left(\B{\epsilon}^\top\B{\Phi}_{\hat{S}\cup S(\bbeta)}\B{v}\right)^2 -\lambda_0 |\hat{S}\setminus S(\bbeta)|\nonumber \\
    & \leq \max_{0\leq k\leq p-s}\left[\frac{4}{n}\max_{\substack{S\subseteq[p]\\|S|=k+s}}\sup_{\substack { \|\B{v}\|_2=1}}\left((\B{\Phi}_{S}^\top\B{\epsilon})^\top\B{v}\right)^2 -\lambda_0 k\right].\label{ineq2}
\end{align}
As a result,

\begin{equation}
    \frac{4}{n}\sup_{\|\B\beta\|_0=s}\left(\B{\epsilon}^\top\frac{\B{X}(\hat{\B{\beta}}-\B{\beta})}{\|\B{X}(\hat{\B{\beta}}-\B{\beta})\|_2}\right)^2 -\lambda_0 \|\hat{\B\beta}\|_0 \leq \max_{0\leq k\leq p-s}\left[\frac{4}{n}\max_{\substack{S\subseteq[p]\\|S|=k+s}}\sup_{\substack { \|\B{v}\|_2=1}}\left((\B{\Phi}_{S}^\top\B{\epsilon})^\top\B{v}\right)^2 -\lambda_0 k\right].\label{ineq2-new}
\end{equation}

 Next, for a given $S\subseteq[p]$ with $|S|=k+s$, we have
\begin{align}
     \p\left( \frac{4}{n}\sup_{\substack { \|\B{v}\|_2=1}}\left((\B{\Phi}_{S}^\top\B{\epsilon})^\top\B{v}\right)^2 -\lambda_0 k>\frac{t}{n}\right)& \stackrel{(a)}{\leq} 2\exp\left(-\frac{t+n\lambda_0 k}{a\sigma^2}\right)6^{|S|}\nonumber \\
     & \leq  2\exp\left(-\frac{t+n\lambda_0 k}{a\sigma^2}+ k\log(3ep) + s\log(3ep) \right)
     \label{epsnet}
\end{align}
for some sufficiently large universal constant $a>0$, where $(a)$ is true by an $\epsilon$-net argument~\citep[Chapter 4]{vershynin2018high} as $(\B{\Phi}_{S}^\top\B{\epsilon})\sim\cN(\B{0},\sigma^2\B{I}_{|S|})$. In the next step, we use union bound on $S$. To this end, from~\eqref{ineq2-new},
\begin{align}
    & \p\left( \frac{4}{n}\sup_{\|\B\beta\|_0=s}\left(\B{\epsilon}^\top\frac{\B{X}(\hat{\B{\beta}}-\B{\beta})}{\|\B{X}(\hat{\B{\beta}}-\B{\beta})\|_2}\right)^2 -\lambda_0 \|\hat{\B\beta}\|_0 >\frac{t}{n}\right)\nonumber\\
   \stackrel{(a)}{\leq} & \sum_{k=0}^{p-s}\sum_{\substack{S\subseteq[p]\\|S|=k+s}} \p\left( \frac{4}{n}\sup_{\substack { \|v\|_2=1}}\left((\B\Phi_{S}^\top\B\epsilon)^\top v\right)^2-\lambda_0 k>\frac{t}{n}\right) \nonumber \\
    \stackrel{(b)}{\leq} &  2\sum_{k=0}^{p-s} {p\choose k+s}\exp\left(-\frac{t+n\lambda_0 k}{a\sigma^2}+ s\log(3ep)+k\log(3ep)\right)\nonumber \\
    \stackrel{(c)}{\leq} &  2\sum_{k=0}^{p-s} \exp\left(-\frac{t+n\lambda_0 k}{a\sigma^2}+ 2s\log (3ep)+2k\log (3ep)\right)\nonumber \\
    \stackrel{(d)}{\leq} &  2\sum_{k=0}^{p-s} \exp\left(-\frac{t}{a\sigma^2}-(c_{\lambda_0}/a)k\log(ep) + 2s\log(3ep)+2k\log (3ep)\right)\nonumber \\
        \stackrel{(e)}{\leq} &  2\sum_{k=0}^p \exp\left(-\frac{t}{a\sigma^2}-2k\log(ep) + 2s\log (3ep)\right)\nonumber \\
         \stackrel{}{\leq} &  2\exp\left(-\frac{t}{a\sigma^2} + 2s\log (3ep)\right)\sum_{k=0}^{\infty} (ep)^{-2k}\nonumber \\
         \stackrel{}{\leq} &  4\exp\left(-\frac{t}{a\sigma^2} + 2s\log (3ep)\right)
\label{after-g-size}
\end{align}
where $(a)$ is by union bound and~\eqref{ineq2-new}, $(b)$ is by~\eqref{epsnet}, $(c)$ is true as ${p \choose k}\leq (ep)^{k}$, $(d)$ is by the lower-bound $\lambda_0\geq c_{\lambda_0}\frac{\sigma^2 \log (ep)}{n}$, and $(e)$ is by taking $c_{\lambda_0}/a$ to be sufficiently large. 
\end{proof}

\begin{lemma}\label{lemmab.2-new}
    For $p\geq s\geq 0$ define
    \begin{equation}
        E_2(s)=\left\{ \frac{1}{n}\|\mathbf{f}^*-\B{X}\hat{\B{\beta}}\|_2^2 + \lambda \hat{K} \gtrsim {\inf_{\|\B\beta\|_0=s}  \left[\frac{1}{n}\|\mathbf{f}^*-\B X\B\beta\|_2^2 + \lambda_0 \|\bbeta\|_0+ \lambda K(\bbeta)\right]}+\frac{\sigma^2 \log(ep)}{n}\right\}
    \end{equation}
    with a sufficiently large multiplicative constant. Then, under the assumptions of Theorem~\ref{estimation_thm}, 
    \begin{equation}
        \p(E_2(s))\leq 4\exp(-10(s\lor 1)\log(ep)). 
    \end{equation}
\end{lemma}

\begin{proof}
   For an arbitrary $\B\beta$ such that $\|\B\beta\|_0=s$ define $\B{\omega}\in\R^n$ to be 
\begin{equation}\label{w-def}
    \B{\omega}=\mathbf{f}^*-\B{X\beta}.
\end{equation}
Note that as $\lambda_0>0$, at least one of the regression coefficients of $\hat{\B\beta}$ corresponding to each categorical predictor is zero:
$$0\in\{\hat\beta_i:i\in\cI_j\}~~\forall~j\in[q].$$
Therefore,
$$\sum_{j=1}^{q} \left\vert\left\{\hat\beta_i:i\in\cI_j\right\}\right\vert = q + \sum_{j=1}^{q} \left\vert\left\{\hat\beta_i:i\in\cI_j, \hat\beta_i\neq 0\right\}\right\vert = q + \hat{K}.$$
Similarly, because we focus on the sparsest representation of~$\bbeta$, and hence every categorical predictor has one cluster with zero $\bbeta$-coefficients, we derive:
$$\sum_{j=1}^{q} \left\vert\left\{\beta_i:i\in\cI_j\right\}\right\vert \leq q + K(\B\beta).$$

Then, by optimality of $\hat{\B{\beta}}$ and the feasibility of $\B{\beta}$ for Problem~\eqref{clusteringproblem-reg},
\begin{align}
    &\frac{1}{n}\|\B{y}-\B{X}\hat{\B{\beta}}\|_2^2+ \lambda \hat{K} + \lambda_0 \|\hat{\B\beta}\|_0  \leq \frac{1}{n}\|\B{y}-\B{X\beta}\|_2^2+ \lambda K(\bbeta) + \lambda_0 \|\B\beta\|_0 \nonumber \\
    \stackrel{(a)}{\Rightarrow} & \frac{1}{n} \|\B{X}(\B{\beta}-\hat{\B{\beta}})\|_2^2 \leq \frac{2}{n}(\B{\epsilon}+\B{\omega})^\top\B{X}(\hat{\B{\beta}}-\B{\beta}) + \lambda (K(\bbeta)-\hat{K}) + \lambda_0 (\|\B\beta\|_0 -\|\hat{\B\beta}\|_0)\nonumber\\
    \stackrel{(b)}{\Rightarrow} & \frac{1}{2n} \|\B{X}(\B{\beta}-\hat{\B{\beta}})\|_2^2 + \lambda \hat{K}\leq \frac{2}{n}\left((\B{\epsilon}+\B{\omega})^\top\frac{\B{X}(\hat{\B{\beta}}-\B{\beta})}{\|\B{X}(\hat{\B{\beta}}-\B{\beta})\|_2}\right)^2 +\lambda K(\bbeta) + \lambda_0 (\|\B\beta\|_0 -\|\hat{\B\beta}\|_0)\nonumber\\
     \stackrel{(c)}{\Rightarrow} & \frac{1}{2n} \|\B{X}(\B{\beta}-\hat{\B{\beta}})\|_2^2 +\lambda \hat{K}\leq \frac{4}{n}\left(\B{\epsilon}^\top\frac{\B{X}(\hat{\B{\beta}}-\B{\beta})}{\|\B{X}(\hat{\B{\beta}}-\B{\beta})\|_2}\right)^2 + \lambda K(\bbeta)+ \lambda_0 (\|\B\beta\|_0 -\|\hat{\B\beta}\|_0) +\frac{4\|\B{\omega}\|_2^2}{n},\label{thm-helper1.4}
\end{align}
where $(a)$ is by substituting $\B{y}=\B{X\beta}+\B{\epsilon}+\B{\omega}$, $(b)$ is true as
\begin{align*}
    2(\B{\epsilon}+\B{\omega})^\top\B{X}(\hat{\B{\beta}}-\B{\beta}) & = 2\frac{(\B{\epsilon}+\B{\omega})^\top\B{X}(\hat{\B{\beta}}-\B{\beta})}{\|\B{X}(\hat{\B{\beta}}-\B{\beta})\|_2}\|\B{X}(\hat{\B{\beta}}-\B{\beta})\|_2 \\ &\leq 2\left((\B{\epsilon}+\B{\omega})^\top\frac{\B{X}(\hat{\B{\beta}}-\B{\beta})}{\|\B{X}(\hat{\B{\beta}}-\B{\beta})\|_2}\right)^2+\frac{1}{2}\|\B{X}(\hat{\B{\beta}}-\B{\beta})\|_2^2
\end{align*}
and $(c)$ is true as
\begin{align*}
    \left((\B{\epsilon}+\B{\omega})^\top\frac{\B{X}(\hat{\B{\beta}}-\B{\beta})}{\|\B{X}(\hat{\B{\beta}}-\B{\beta})\|_2}\right)^2 & \leq 2\left(\B{\epsilon}^\top\frac{\B{X}(\hat{\B{\beta}}-\B{\beta})}{\|\B{X}(\hat{\B{\beta}}-\B{\beta})\|_2}\right)^2 + 2\left(\B{\omega}^\top\frac{\B{X}(\hat{\B{\beta}}-\B{\beta})}{\|\B{X}(\hat{\B{\beta}}-\B{\beta})\|_2}\right)^2 \\
    & \leq 2\left(\B{\epsilon}^\top\frac{\B{X}(\hat{\B{\beta}}-\B{\beta})}{\|\B{X}(\hat{\B{\beta}}-\B{\beta})\|_2}\right)^2 + 2\|\B{\omega}\|_2^2\left\Vert\frac{\B{X}(\hat{\B{\beta}}-\B{\beta})}{\|\B{X}(\hat{\B{\beta}}-\B{\beta})\|_2}\right\Vert_2^2.
\end{align*}
Take
$$t=c_t\sigma^2(s\lor 1)\log(ep)$$
in Lemma~\ref{thm1-lem1} for some sufficiently large absolute constant $c_t>0$. Then, from Lemma~\ref{thm1-lem1} and the event $E_1(t,s)$, with probability at least
\begin{equation}\label{thm1-prob-old}
    1-4\exp\left(-10(s\lor 1)\log (ep)\right)
\end{equation}
we have
\begin{equation}\label{thm1-finalbound}
     \frac{4}{n}\left(\B{\epsilon}^\top\frac{\B{X}(\hat{\B{\beta}}-\B{\beta})}{\|\B{X}(\hat{\B{\beta}}-\B{\beta})\|_2}\right)^2 -\lambda_0\|\hat{\B\beta}\|_0\lesssim \frac{\sigma^2 (s\lor 1)\log (ep)}{n}.
\end{equation}
Finally,
\begin{align}
   \|\mathbf{f}^*-\B{X}\hat{\B{\beta}}\|_2^2 \lesssim  \|\mathbf{f}^*-\B{X}{\B{\beta}}\|_2^2 + \|\B{X}{\B{\beta}}-\B{X}\hat{\B{\beta}}\|_2^2 =\|\B{\omega}\|_2^2+ \|\B{X}{\B{\beta}}-\B{X}\hat{\B{\beta}}\|_2^2.\label{thm1-final-helper}
\end{align}
We complete the proof of the lemma by substituting~\eqref{thm1-finalbound} and~\eqref{thm-helper1.4} into~\eqref{thm1-final-helper}.
\end{proof}

\begin{proof}[Proof of Theorem~\ref{estimation_thm}]

The proof follows from Lemma~\ref{lemmab.2-new}. In particular,
\begin{align}
    \p\left(\bigcap_{s=0}^p E_2^c(s)\right) & \geq 1 - \sum_{s=0}^p \p(E_2(s)) \nonumber\\
    & \geq 1- \sum_{s=0}^p 4\exp(-10(s\lor 1)\log(ep)) \nonumber \\
    & \geq 1-4(ep)^{-9}.\label{thm1-prob}
\end{align}

\end{proof}

\section{Proof of Theorem~\ref{estimation_thm2}}
When $\lambda_0=0$, there are infinitely many solutions to our optimisation problem~\eqref{clusteringproblem-reg}, because we can shift all the coefficients of a categorical predictor by a constant, and adjust the intercept accordingly, without changing either the vector of predictions or the clustering pattern of the solution. Because all  these solutions have the same prediction error, we can focus on any one of them without loss of generality. We will focus on the \emph{sparsest} solution for convenience. Thus, the largest estimated cluster for each categorical predictor will have zero estimated coefficients.

Following our approach in Theorem~\ref{estimation_thm}, we present the proof for the setting where $\alpha=\hat\alpha=0$. The general proof follows after minor adjustments to account for the intercept.

For $\B\beta\in\R^p$, define 
$$\bG(\B\beta)=\{j\in[q]:\beta_k\neq 0 \text{ for some } k\in\cI_j\},~~g(\B\beta)=|\bG(\B\beta)|.$$
\begin{lemma}\label{thm2-lem1}
For $q\geq \tilde g\geq 0$ and $t>0$ define
\begin{equation}
    E_3(t,\tilde g)=\left\{\frac{4}{n}\sup_{g(\bbeta)=\tilde g}\left(\B{\epsilon}^\top\frac{\B{X}(\hat{\B{\beta}}-\B{\beta})}{\|\B{X}(\hat{\B{\beta}}-\B{\beta})\|_2}\right)^2 -\lambda  \hat{K}>\frac{t}{n}\right\}.
\end{equation}
Under the conditions of Theorem~\ref{estimation_thm2} with $\lambda\geq c_{\lambda}\sigma^2 p_{\max} \log(eq)/n$, 
\begin{equation}
     \p\left( E_3(t,\tilde g)\right)
         \stackrel{}{\leq}   4\exp\left(-\frac{t}{a\sigma^2} + \tilde{g}p_{\max}\log (eq)\right)
\end{equation}   
for some universal constant $a>0$.
\end{lemma}

\begin{proof}[Proof of Lemma~\ref{thm2-lem1}]
Fix $\B\beta$ such that $\tilde g=g(\bbeta)$ and let $\hat{g}=g(\hat\bbeta)$. For any $\mathbb{G}\subseteq[q]$, let us define 
$$S_{\mathbb{G}} = \{i\in\cI_j: j\in\mathbb{G}\}.$$ 
Under this notation, $\B{X}(\hat{\B{\beta}}-\B{\beta})$ lies in the column span of $\B X_{{S}_{\mathbb{G}(\hat\bbeta)}\cup {S}_{\mathbb{G}(\bbeta)}}$.
Then, following the same steps leading to~\eqref{ineq2}, we obtain: 
    \begin{align}
   \frac{4}{n} \left(\B{\epsilon}^\top\frac{\B{X}(\hat{\B{\beta}}-\B{\beta})}{\|\B{X}(\hat{\B{\beta}}-\B{\beta})\|_2}\right)^2 -\lambda \hat g & = \frac{4}{n} \left(\B{\epsilon}^\top\frac{\B{X}(\hat{\B{\beta}}-\B{\beta})}{\|\B{X}(\hat{\B{\beta}}-\B{\beta})\|_2}\right)^2 -\lambda |\bG(\hat\bbeta)| \nonumber\\
   & \leq \frac{4}{n} \left(\B{\epsilon}^\top\frac{\B{X}(\hat{\B{\beta}}-\B{\beta})}{\|\B{X}(\hat{\B{\beta}}-\B{\beta})\|_2}\right)^2 -\lambda |\bG(\hat\bbeta)\setminus \bG(\bbeta)| \nonumber \\
   & \leq \max_{0\leq k\leq q-\tilde g}\left[\frac{4}{n}\max_{\substack{\mathbb{G}\subseteq[q]\\|\mathbb{G}|=\tilde g+k}}\sup_{\substack { \|\B{v}\|_2=1}}\left((\B{\Phi}_{S_{\mathbb{G}}}^\top\B{\epsilon})^\top\B{v}\right)^2 -\lambda k\right].\label{ineq2_2-old}
\end{align}
Hence,
\begin{equation}
       \frac{4}{n} \sup_{g(\bbeta)=\tilde g}\left(\B{\epsilon}^\top\frac{\B{X}(\hat{\B{\beta}}-\B{\beta})}{\|\B{X}(\hat{\B{\beta}}-\B{\beta})\|_2}\right)^2 -\lambda \hat g \leq \max_{0\leq k\leq q-\tilde g}\left[\frac{4}{n}\max_{\substack{\mathbb{G}\subseteq[q]\\|\mathbb{G}|=\tilde g+k}}\sup_{\substack { \|\B{v}\|_2=1}}\left((\B{\Phi}_{S_{\mathbb{G}}}^\top\B{\epsilon})^\top\B{v}\right)^2 -\lambda k\right].\label{ineq2_2}
\end{equation}

Additionally, as $\hat{K}\geq \hat{g}$, 
    \begin{align}
   \frac{4}{n} \sup_{g(\B\beta)=\tilde g}\left(\B{\epsilon}^\top\frac{\B{X}(\hat{\B{\beta}}-\B{\beta})}{\|\B{X}(\hat{\B{\beta}}-\B{\beta})\|_2}\right)^2 -\lambda \hat K
    \leq \frac{4}{n} \sup_{g(\B\beta)=\tilde g}\left(\B{\epsilon}^\top\frac{\B{X}(\hat{\B{\beta}}-\B{\beta})}{\|\B{X}(\hat{\B{\beta}}-\B{\beta})\|_2}\right)^2 -\lambda \hat g.\label{ineq2_3}
\end{align}

By union bound over $\bG$ and $k$ in~\eqref{ineq2_2} (similar to the one leading to~\eqref{after-g-size}), and using the fact 
$|S_{\mathbb{G}}|\leq  p_{\max}|\bG| $
we obtain 
$$     \p\left( \frac{4}{n}\sup_{g(\bbeta)=\tilde g}\left(\B{\epsilon}^\top\frac{\B{X}(\hat{\B{\beta}}-\B{\beta})}{\|\B{X}(\hat{\B{\beta}}-\B{\beta})\|_2}\right)^2 -\lambda  \hat{K}>\frac{t}{n}\right)
         \stackrel{}{\leq}   4\exp\left(-\frac{t}{a\sigma^2} + \tilde{g}p_{\max}\log (eq)\right).$$
\end{proof}

\begin{lemma}\label{lemma-c2-new}
     For $q\geq \tilde g\geq 0$ define
    \begin{equation}
        E_4(\tilde g)=\left\{ \frac{1}{n}\|\mathbf{f}^*-\B{X}\hat{\B{\beta}}\|_2^2  \gtrsim {\inf_{g(\bbeta)=\tilde g}  \left[\frac{1}{n}\|\mathbf{f}^*-\B X\B\beta\|_2^2 + \lambda_0 \|\bbeta\|_0+ \lambda (K(\bbeta)\lor 1)\right]}\right\}
    \end{equation}
    with a sufficiently large multiplicative constant. Then, under the assumptions of Theorem~\ref{estimation_thm2} with $\lambda\geq c_{\lambda}\sigma^2 p_{\max} \log(eq)/n$,
    \begin{equation}
        \p(E_4(\tilde g))\leq 4\exp(-10(\tilde g \lor 1)p_{\max}\log(eq)). 
    \end{equation}
\end{lemma}
\begin{proof}
Take
   $$t\gtrsim \sigma^2 (\tilde{g} \lor 1)p_{\max}\log (eq) $$
   in Lemma~\ref{thm2-lem1}. Then, we have that
   \begin{equation}
       \frac{4}{n}{\sup_{g(\bbeta)=\tilde g}}\left(\B{\epsilon}^\top\frac{\B{X}(\hat{\B{\beta}}-\B{\beta})}{\|\B{X}(\hat{\B{\beta}}-\B{\beta})\|_2}\right)^2 -\lambda  \hat{K} \lesssim \frac{\sigma^2 (\tilde g\lor 1)p_{\max}\log (eq)}{n}.\label{lemmanewc2-helper}
   \end{equation}
   with probability greater than
      \begin{equation}
       1-4\exp(-10(\tilde g\lor 1)p_{\max}\log (eq)).
   \end{equation}
   Fix $\bbeta$ such that $g(\bbeta)=\tilde g$. Noting that $\tilde g\lesssim K(\bbeta)$, from~\eqref{lemmanewc2-helper} we obtain
   \begin{equation}\label{lemmanewc2-helper2}
       \frac{4}{n}\left(\B{\epsilon}^\top\frac{\B{X}(\hat{\B{\beta}}-\B{\beta})}{\|\B{X}(\hat{\B{\beta}}-\B{\beta})\|_2}\right)^2 -\lambda  \hat{K} \lesssim \frac{\sigma^2 (K(\B\beta)\lor 1)p_{\max}\log (eq)}{n}.
   \end{equation}
Under $\lambda\geq c_{\lambda}\sigma^2 p_{\max} \log(eq)/n$, note that without loss of generality we can assume $0\in\{\hat{\beta}_i:i\in\cI_j\}$ for $j\in[q]$ because of the intercept, and therefore,~\eqref{thm-helper1.4} holds. This completes the proof with~\eqref{lemmanewc2-helper2}.
\end{proof}

\begin{proof}[Proof of Theorem~\ref{estimation_thm2}]
Note that when the lower-bound on~$\lambda_0$ holds, the theorem is a direct result of Theorem~\ref{estimation_thm}.  Under lower-bound on~$\lambda$, the proof follows from Lemma~\ref{lemma-c2-new}. In particular,
\begin{align}
    \p\left(\bigcap_{\tilde g=0}^q E_4^c(\tilde g)\right) & \geq 1 - \sum_{\tilde g=0}^q \p(E_4(\tilde g)) \nonumber\\
    & \geq 1- \sum_{\tilde g=0}^q 4\exp(-10(\tilde g\lor 1)p_{\max}\log(eq)) \nonumber \\
    & \geq 1-8(eq)^{-9p_{\max}}.\label{thm1.5-0-prob}
\end{align}
Therefore, by union bound over the two possible cases, i.e., with the lower-bound on~$\lambda_0$ or the lower-bound on~$\lambda$, the theorem holds with probability at least
      \begin{equation}\label{thm2new-prob}
       1-8(ep)^{-9}-8(eq)^{-9p_{\max}}.
   \end{equation}
\end{proof}

\section{Proofs of the Cluster Recovery Results (Section~\ref{sec:clusterrecovery})}
\label{sec.proofs}

\subsection{Preliminaries}
\label{sec.prf.prelim}

To formalise our results, we first define the clustering pattern of $\B\beta\in\R^p$.

\begin{definition}[Clustering Pattern]\label{cg-def}
    We write $\cG$ or $\cG(\B{\beta})$ for the clustering pattern of $\B\beta$, which we define as the collection of non-empty subsets of $[p]$ such that: \\
\noindent\textbf{(a)} If $G\in\cG(\B{\beta})$, then $\beta_j=\beta_k$ for all $j,k\in G$.\\
\noindent\textbf{(b)}  If $G\in\cG(\B{\beta})$, then there exists $j\in [q]$ such that $G\subseteq\cI_j$.\\
\noindent\textbf{(c)} If $G_1,G_2\in\cG(\B{\beta})$ and $G_1,G_2\subseteq \cI_j$, then  $\beta_{k_1}\neq \beta_{k_2}$ for all $k_1\in G_1$ and $k_2\in G_2$.\\
\noindent\textbf{(d)} $\cup_{G\in\cG(\B{\beta})}G=[p]$.
\end{definition}

In Definition~\ref{cg-def}, condition (a) means that the $\bbeta$-coefficients in each cluster are equal; (b) means that the coefficients in each cluster correspond to the same categorical variable; (c) means that the clusters are maximal; and (d) means that the clustering is exhaustive.

\begin{definition}
\label{DefinitionB} Given $\mathcal{G}=\mathcal{G}(\bbeta)$ with $\bbeta\in\mathbb{R}^p$, we define~$r({\mathcal{G}})$ as the \textit{smallest possible} number of reassignments in stage S1 of the following general procedure for converting~$\mathcal{G}$ into~$\mathcal{G^*}$.
\end{definition}
We note that the conversion procedure and the associated quantities depend on both~$\cG$ and~$\bbeta$. However, to simplify the expressions we suppress the dependence on~$\bbeta$ in our notation.

The three stages of the procedure (S1 - S3) are implemented in sequence. The procedure starts with clustering pattern~$\mathcal{G}$ and a partition into the ``zero'' and the ``nonzero'' clusters (made on the basis of the $\bbeta$-coefficients) and ends with clustering pattern~$\mathcal{G}^*$. Furthermore, the zero and nonzero clusters at the end of the procedure exactly match the corresponding true clusters determined by~$\cG^*$.   Even though the description below does not specify a unique procedure, the quantity~$r({\mathcal{G}})$ is well-defined. 
\begin{enumerate}

\item[S1.] \textit{Reassignment}: change the cluster assignment for some of the variables (increasing the number of clusters is allowed).

\item[S2.]  \textit{Zeroing-out}: Convert some of the nonzero clusters into zero clusters.

\item[S3.]
\textit{Merging}: merge some of the nonzero clusters.
\end{enumerate}

We will write~$\widehat{r}$ for~$r(\widehat{\mathcal{G}})$. There may exist many procedures described in Definition~\ref{DefinitionB} such that the number of stage S1 reassignments is equal to~$\widehat{r}$. For the rest Section~\ref{sec.proofs}, we focus on one such procedure.

We note that it is possible to increase the number of merges in stage~S3 of the procedure without increasing the number of reassignments in stage~S1; this can be achieved by replacing a reassignment to an existing cluster with a reassignment to a new cluster coupled with a merge. We will avoid this situation by focusing on the procedure with~$\widehat{r}$ reassignments and the \textit{smallest} possible number of merges in stage~S3.

Similarly, we could increase the number of coefficients zeroed-out in stage~S2 without increasing the number of reassignments in stage~S1; this can be achieved by replacing a zeroing-out action with a reassignment to a new nonzero cluster coupled with a conversion into a zero cluster. We will avoid this situation by focusing on the procedure with~$\widehat{r}$ reassignments and the \textit{smallest} possible number of coefficients zeroed out in stage~S2.

We let~$\widehat{\bbeta}_j$ and~${\bbeta}^*_j$ be the sub-vectors of~$\widehat{\bbeta}$ and~${\bbeta}^*$, respectively, corresponding to the $j$-th categorical predictor.
We define $\widehat{r}_j$ as the number of reassignments corresponding to the $j$-th predictor that are implemented  in stage S1 of the procedure.  Note that $\widehat{r}=\sum_{j=1}^{q}\widehat{r}_j$. Similarly, we let~$\widehat{t}_j$ denote the number of the $\widehat{\bbeta}_j$-coefficients zeroed-out in stage S2 and define $\widehat{t}=\sum_{j=1}^{q}\widehat{t}_j$. We write~$\widehat{m}_j$ for the corresponding number of stage S3 merges and let $\widehat{m}=\sum_{j=1}^{q}\widehat{m}_j$. We also let $\widehat{s}=\|\widehat{\bbeta}\|_0$.

We define matrix $\B{X}_{\cG}\in\R^{n\times |\cG|}$ so that each of its columns corresponds to a distinct cluster $G\in\cG$ and is given by
$\sum_{i\in G}\B{X}_{i}$.  We also define $\bP_{\bX_{\mathcal{G}}}$ as an orthonormal matrix for projecting onto the column span of~$\B{X}_{\cG}$, and write~$\bI_n$ for the $n\times n$ identity matrix.

We will need the following two lemmas, which  are proved in Section~\ref{sec:prf.lemmas}.
\begin{lemma}
\label{max.ineq.lem1}
Let $\widehat{w}=w(\widehat{\cG})$ and suppose that Assumption~\ref{assumptiona} holds. Then, inequalities
\begin{enumerate}

\item[(i)] \; $\frac{\bepsilon^{\top}}{\sigma}(\bP_{\bX_{\widehat{\mathcal{G}}}}-\bP_{\bX_{\mathcal{G}^*}})\frac{\bepsilon}{\sigma}\lesssim \big[(\widehat{r}+\widehat{t})\log(p)+\widehat{m}c^*\log(K^*\vee 2)\big]\wedge\big[(\widehat{s}+s^*)\log(p)\big]+\log(1/\epsilon_0)$ \; and\\

\item[(ii)] \, $\big|\frac{\bepsilon^{\top}}{\sigma}(\bI_n-\bP_{\bX_{\widehat{\mathcal{G}}}})\bff^*\big|\lesssim \sqrt{\big[(\widehat{r}+\widehat{t})\log(p)+\widehat{m}c^*\log(K^*\vee 2)\big]\wedge\big[(\widehat{s}+s^*)\log(p)\big]+\log(1/\epsilon_0)}\sqrt{n\widehat{w}}$

\end{enumerate}
hold uniformly over $\epsilon_0\in(0,1/2)$, with probability at least $1-\epsilon_0$.
\end{lemma}

\begin{lemma}
\label{lem.BA}
Suppose that Assumption~\ref{assumptiona} holds. Then, $\mathcal{E}(\cG)\gtrsim r(\cG)$.
\end{lemma}

We note that Lemma~\ref{lem.BA} and Assumption~\ref{assumptiona} imply 
\begin{equation}
\label{mu.bnd2}
w(\cG)\gtrsim \frac{r(\cG)\sigma^2\log(p)}{n}~~\text{for all~} \kappa s^*\text{-sparse\ clustering patterns~} \cG.
\end{equation}
Moreover, the multiplicative constant in the above inequality can be made arbitrarily large by increasing the corresponding multiplicative constant in bound~(\ref{assumptionA.bnd}) of Assumption~\ref{assumptiona}. 

\subsection{Proof of Theorem~\ref{clust.recov.thm}}
\begin{proof}[Proof of Theorem~\ref{clust.recov.thm}]
To simplify the exposition, we suppose that every nonzero cluster in $\mathcal{G}^*$ has the same size: $c^*$. The general case follows with minor modifications. To simplify the expressions we also assume, without loss of generality, that $\sigma=1$.

For the rest of the proof, we restrict our attention to the high-probability event on which the inequalities in Lemma~\ref{max.ineq.lem1} are satisfied, with~$\epsilon_0$ specified at the end.

We write~$L^*_j$ for the number of true clusters (including the zero cluster) for the $j$-th predictor and let $L^*=\sum_{j=1}^{q} L^*_j$. We define~$\widehat{L}_j$ and~$\widehat{L}$ similarly.

Given a clustering pattern~$\cG$ we define
\begin{equation*}
  \CR_{\cG}(\bY)=  \min_{\alpha\in\mathbb{R}, \bgamma\in\mathbb{R}^{|\cG|}} \,\frac{1}{n}\|\bY-\alpha\B{1}-\bX_{\cG}\gamma\|_2^2= \frac{1}{n}\bY^{\top}(\bI_n-\bP_{\bX_{\cG}})\bY,
\end{equation*}
and observe that 
$$
\CR_{\widehat{\mathcal{G}}}(\bY)-\CR_{\cG^*}(\bY)=\widehat{w} +\frac2n\bepsilon^{\top}(\bI_n-\bP_{\bX_{\widehat{\mathcal{G}}}})\bff^* +
\frac1n\bepsilon^{\top}(\bP_{\bX_{\mathcal{G}^*}}-\bP_{\bX_{\widehat{\mathcal{G}}}})\bepsilon.
$$
Using Lemma~\ref{max.ineq.lem1} to bound the last two terms, and noting that
\begin{equation}
\label{main.ineq.th1.prf}
\CR_{\widehat{\mathcal{G}}}(\bY)-\CR_{\cG^*}(\bY)+\lambda(\widehat{L}-L^*)+\lambda_0(\widehat{s}-s^*)\le0,
\end{equation}
we derive

\begin{align}
\widehat{w}+\lambda(\widehat{L}-L^*)+\lambda_0(\widehat{s}-s^*) \; \lesssim & 
\sqrt{\frac{[\widehat{r}+\widehat{t}]\log(p)+\widehat{m}c^*\log(K^*\vee 2)-\log(\epsilon_0)}{n}} \sqrt{\widehat{w}}
\nonumber \\
\nonumber\\
& + \frac{[\widehat{r}+\widehat{t}]\log(p)+\widehat{m}c^*\log(K^*\vee 2)-\log(\epsilon_0)}{n}, \nonumber
\end{align}
which implies
\begin{equation}
\label{prof.th1.ineq3}
\widehat{w}+\lambda\widehat{L}- \lambda L^* +\lambda_0\widehat{s}-\lambda_0s^* \lesssim  \frac{[\widehat{r}+\widehat{t}]\log(p)+\widehat{m}c^*\log(K^*\vee 2)-\log(\epsilon_0)}{n}.
\end{equation}

We note that the conversion of zero coefficients to nonzero happens only in stage~S1 of the procedure, while the conversion of nonzero coefficients to zero happens only in stage~S2. Consequently,
$$
\|\widehat{\bbeta}_j\|_0-\|\bbeta^*_j\|_0\ge \widehat{t}_j-\widehat{r}_j \quad \text{for all}\;j\in[q],\;\text{and hence} \quad \widehat{s}-s^*\ge \widehat{t}-\widehat{r}.
$$
It follows that by taking the universal constant in the lower bound for~$\lambda_0$ sufficiently large, we can cancel the $\widehat{t}$ term on the right-hand side of~\eqref{prof.th1.ineq3} and still have a $\lambda_0\widehat{t}/2$ term on the left-hand side. Thus,
we can rewrite inequality~\eqref{prof.th1.ineq3} as
\begin{equation*}
\widehat{w}+\lambda\big(\widehat{L}- L^*\big) +\frac{\lambda_0\widehat{t}}{2} \lesssim  \frac{\widehat{r}\log(p)+\widehat{m}c^*\log(K^*\vee 2)-\log(\epsilon_0)}{n}.
\end{equation*}
In other words,
\begin{align}
\widehat{w}+\sum_{j:\,\widehat{L}_j\ge L^*_j}\lambda\big(\widehat{L}_j- L^*_j\big) +\sum_{j:\,\widehat{L}_j < L^*_j}\lambda\big(\widehat{L}_j- L^*_j\big) +\frac{\lambda_0\widehat{t}}{2}\nonumber\\
 \le \tilde{a}\frac{[\widehat{r}\log(p)+\widehat{m}c^*\log(K^*\vee 2)-\log(\epsilon_0)]}{n}
 \label{T1.prf.main.mu.ineq} 
 \end{align}
for some positive constant~$\tilde{a}$.

We now analyze the second and third terms on the left-hand side of inequality~(\ref{T1.prf.main.mu.ineq}) separately.

\textbf{Second term}:  $\sum_{j:\,\widehat{L}_j\ge L^*_j}\lambda\big(\widehat{L}_j- L^*_j\big)$. Because we chose a procedure with~$\widehat{r}$ reassignments but the smallest possible number of merges, we have $\widehat{m}_j \le \widehat{L}_j- L^*_j$ for all~$j$ with $\widehat{L}_j\ge L^*_j$.
Thus, letting $\tilde{b}=2a$, where~$a$ is the multiplicative constant in the lower bound on~$\lambda$, we derive
\begin{eqnarray}
\sum_{j:\,\widehat{L}_j\ge L^*_j}\lambda\big(\widehat{L}_j- L^*_j\big) &\ge& \lambda \widehat{m}\nonumber\\
&\ge& \frac{\lambda \widehat{m}}{2} + \tilde{b} \frac{\widehat{m}c^*\log(K^*\vee 2)}{n}.
\label{T1.prf.main.mu.ineq2}
\end{eqnarray}

\bigskip

\textbf{Third term}: $\sum_{j:\,\widehat{L}_j < L^*_j}\lambda\big(\widehat{L}_j- L^*_j\big)$.
Note that $\widehat{r}_j\ge c^*(L^*_j - \widehat{L}_j)$, because for each additional true cluster, there are~$c^*$ coefficients that need to be reassigned to a different cluster to get from~$\widehat{\mathcal{G}}_j$ to~$\mathcal{G}^*_j$. Consequently, recalling that~$b$ be the multiplicative constant in the upper bound on~$\lambda$ specified in the statement of the theorem, we derive
\begin{eqnarray}
\sum_{j:\,\widehat{L}_j < L^*_j}\lambda\big(\widehat{L}_j- L^*_j\big) &\ge&
 - b\frac{c^*\log(p)}{n}\sum_{j:\,\widehat{L}_j < L^*_j}\big(\widehat{L}_j- L^*_j\big)\nonumber\\
   &\ge& - b\frac{\widehat{r}\log(p)}{n}.
\label{T1.prf.main.mu.ineq3}
\end{eqnarray}

\bigskip

\textbf{Combining the terms together} via inequalities~(\ref{T1.prf.main.mu.ineq}), (\ref{T1.prf.main.mu.ineq2}) and~(\ref{T1.prf.main.mu.ineq3}), we arrive at
\begin{equation}
\label{prof.th1.ineq4}
\widehat{w}+\frac{\lambda\widehat{m}}{2}+\frac{\lambda_0\widehat{t}}{2} \lesssim  \frac{\widehat{r}\log(p)-\log(\epsilon_0)}{n}.
\end{equation}

Coming back to~(\ref{main.ineq.th1.prf}), and applying Lemma~\ref{max.ineq.lem1} again, this time using the $(\widehat{s}+s^*)\log(p)$ term in the bounds, we can rewrite inequality~(\ref{prof.th1.ineq3}) as follows:
\begin{equation*}
\lambda_0\widehat{s}-\lambda_0s^* \lesssim  \frac{(\widehat{s}+s^*)\log(p)-\log(\epsilon_0)}{n}.
\end{equation*}
We can ensure that $\widehat{s}\le \kappa s^*$ by using a sufficiently large multiplicative factor in the lower bound $\lambda_0\gtrsim \log(p)/n$ and setting $\epsilon_0=p^{-\tilde{c}}$ for a sufficiently small positive~$\tilde{c}$.
Consequently, we have $\widehat{w}\gtrsim  \frac{\widehat{r}\log(p)}{n}$ by inequality~(\ref{mu.bnd2}), and the multiplicative constant can be made arbitrarily large by increasing its counterpart in bound~(\ref{assumptionA.bnd}) . Thus, if the multiplicative constant in~(\ref{assumptionA.bnd}) is sufficiently large, inequality~(\ref{prof.th1.ineq4}) yields
$$
\frac{\widehat{r}\log(p)}{n}+\frac{\lambda\widehat{m}}{2}+\frac{\lambda_0\widehat{t}}{2} \lesssim  \frac{\log(1/\epsilon_0)}{n}.
$$
Let $\lambda_l$ denote the smallest possible value of~$\lambda$. If we take~$\epsilon_0=\max\{\exp(-cn\lambda_l),p^{-\tilde{c}}\}$ for a sufficiently small positive~$\tilde{c}$, we can rewrite the above inequality as
$$
\frac{\widehat{r}\log(p)}{n}+\frac{\lambda\widehat{m}}{2}+\frac{\lambda_0\widehat{t}}{2}\le  \min\Big\{\frac{\log(p)}{2n},\frac{\lambda_0}{4},\frac{\lambda_l}{4}\Big\},
$$
implying~$\widehat{r}=0$, $\widehat{m}=0$ and $\widehat{t}=0$ with probability at least $1-\epsilon_0$.

If $\lambda\asymp c^*\log(p)/n$, then $\lambda_l\gtrsim \log(p)/n$ and $\epsilon_0\le p^{-c}$ for some $c>0$.\\
On the other hand, if
$\lambda\gtrsim c^*\log(K^*\vee 2)/n$, then $\epsilon_0\le (K^*\vee 2)^{-bc^*}\vee p^{-c}$.
\end{proof}

\subsection{Proof of Theorem~\ref{th.BA}}
\label{sec:prf.thm.suff}

\begin{proof}[Proof of Theorem~\ref{th.BA}]

Consider a $\kappa s^*$-sparse clustering pattern~$\mathcal{G}$. Let $\bff$ be the prediction vector such that $w(\cG)=(1/n)\|\bff^*-\bff\|^2$, and note that $\bff\in\mathcal{F}(\kappa s^*)$. Let~$\tilde{\bbeta}$ and $\tilde{\bbeta}^*$ be the vectors from the statement of Theorem~\ref{th.BA}, such that 
\begin{equation}
\label{cond.BA.repeat}
\frac{\|\bff^*-\bff\|_2}{\|\tilde{\bbeta}^*-\tilde{\bbeta}\|_2}\gtrsim \sqrt{n^*},
\end{equation}
which leads to 
\begin{equation}
\label{lwr.bnd.mu}
w(\mathcal{G})\gtrsim \frac{n^*}{n}\|\tilde{\bbeta}-\tilde{\bbeta}^*\|^2.
\end{equation}

Breaking $\mathcal{E}(\cG)$ down by the categorical predictors  and by the values of $\tilde{\beta}_j$ , we write
$$\mathcal{E}(\cG)=\sum_{k=1}^{q}\sum_{a\in\{\tilde{\beta}_j:\,j\in\cI_k\}}\mathcal{E}_{k,a},$$
where
$$
\mathcal{E}_{k,a}=\big|\{l\in\cI_k:\, {\tilde{\beta}}_l=a\}\big|-\max_{b\in\{\tilde{\beta}_j^*:\,j\in\cI_k\}}\big|\{l\in\cI_k:\, {\tilde{\beta}}_l=a,\,\tilde{\beta}^*_l=b\}\big|.
$$
Let~$\tilde{\bbeta}_{k,a}$ and~$\tilde{\bbeta}^*_{k,a}$ be the sub-vectors of~$\tilde{\bbeta}$ and~$\tilde{\bbeta}^*$, respectively, indexed by the set $\{l\in\cI_k:\, {\tilde{\beta}}_l=a\}$. We will now establish the following lower bound for every~$k$ and~$a$:
\begin{equation}
\label{lwr.bnd.minority}
\|\tilde{\bbeta}_{k,a}-\tilde{\bbeta}^*_{k,a}\|^2\ge \mathcal{E}_{k,a}\Delta_{\text{min}}^2/4.
\end{equation}

Define~$B_a$ as the set of \textit{distinct} elements in $\{{\beta}^*_j:\,j\in\cI_k, \tilde{\beta}_j=a\}$, and note that all the elements of~$B_a$ are separated by at least~$\Delta_{\text{min}}$. Consequently, there exists at most \textit{one} element~$b$ in the set~$B_a$ that satisfies $|b-a|<\Delta_{\text{min}}/2$. If such an element exists, we denote it by~$b_a$, otherwise we can set~$b_a$ equal to an arbitrary element of~$B_a$. Because all elements of~$\tilde{\bbeta}_{k,a}$ are equal to~$a$, we then have
\begin{eqnarray*}
\|\tilde{\bbeta}_{k,a}-\tilde{\bbeta}^*_{k,a}\|^2&\ge& \frac{\Delta_{\text{min}}^2}{4}\big( |\{l\in\cI_k:\, {\tilde{\beta}}_l=a\}|-|\{l\in\cI_k:\, {\tilde{\beta}}_l=a,\,\tilde{\beta}^*_l=b_a\}| \big)\\
&\ge& \frac{\Delta_{\text{min}}^2}{4}\mathcal{E}_{k,a},
\end{eqnarray*}
and hence we have derived inequality~(\ref{lwr.bnd.minority}). Aggregating both sides of this inequality over all $a\in\{\tilde{\beta}_j:\,j\in\cI_k\}$ and $k\in[q]$, we arrive at
\begin{equation}
\label{lwr.bnd.minority2}
\|\tilde{\bbeta}-\tilde{\bbeta}^*\|^2\ge \mathcal{E}(\cG)\Delta_{\text{min}}^2/4.
\end{equation}
The conclusion of Theorem~\ref{th.BA} follows directly from~(\ref{lwr.bnd.mu}) and~(\ref{lwr.bnd.minority2}).
\end{proof}

\subsection{Proof of Lemma~\ref{max.ineq.lem1} and Lemma~\ref{lem.BA}}
\label{sec:prf.lemmas}

\begin{proof}[Proof of Lemma~\ref{max.ineq.lem1}]

To simplify the exposition, we again suppose that every nonzero cluster in $\mathcal{G}^*$ has the same size: $c^*$. The general case follows with minor modifications. To simplify the expressions we also assume, without loss of generality, that $\sigma=1$.

In what follows, we establish that each inequality in Lemma~\ref{max.ineq.lem1} holds with probability at least $1-\epsilon_0/2$.
We conduct most of our arguments for a general clustering pattern $\cG=\cG(\bbeta)$ rather than $\widehat{\cG}=\cG(\widehat{\bbeta})$ and establish the probability inequalities  uniformly over~$\cG$.
We consider a procedure for converting of~$\mathcal{G}$ into~$\mathcal{G}^*$ according to Definition~\ref{DefinitionB}. We focus on a procedure such that the number of reassignments in stage~S1 is $r(\mathcal{G})$, while both the number of merges in stage~S3 and the number of coefficients zeroed-out in stage~S2 are chosen as small as possible, as described in the beginning of Section~\ref{sec.proofs}. 

\textbf{Inequality (i) of Lemma~\ref{max.ineq.lem1}}. We write~$t(\mathcal{G})$ for the number of coefficients zeroed-out in stage~S2 of the procedure. We define~$\tilde{\mathcal{G}}$ for the collection of nonzero clusters after stages~S1 and~S2 of the procedure are completed. Note that we can recover~$\mathcal{G}^*$ by merging some (or none) of the clusters in~$\tilde{\mathcal{G}}$. We write~$\mathcal{I}$ for the index set of the coefficients that get reassigned and zeroed out in the stages~S1 and~S2 of the procedure. We observe that $|\mathcal{I}|=r(\mathcal{G})+t(\mathcal{G})$,  $m(\mathcal{G})=|\tilde{\mathcal{G}}|-|\mathcal{G}|$, and
$$
\text{span}(\bX_{\mathcal{G}})\subseteq\text{span}(\bX_{\tilde{\mathcal{G}}},\bX_{\mathcal{I}}).
$$
Consequently,
\begin{eqnarray}
\label{eps.proj.reform}
\bepsilon^{\top}(\bP_{\bX_{\mathcal{G}}}-\bP_{\bX_{\mathcal{G}^*}})\bepsilon &=& \|\bP_{\bX_{\mathcal{G}}}\bepsilon\|^2-\|\bP_{\bX_{\mathcal{G}^*}}\bepsilon\|^2 \nonumber \\
&=& \|\bP_{\bX_{\mathcal{G}}}\bepsilon\|^2-\|\bP_{\bX_{\tilde{\mathcal{G}}}}\bepsilon\|^2+\|\bP_{\bX_{\tilde{\mathcal{G}}}}\bepsilon\|^2-\|\bP_{\bX_{\mathcal{G}^*}}\bepsilon\|^2 \nonumber \\
&\le& \Big(\|\bP_{\bX_{\tilde{\mathcal{G}}\cup\mathcal{I}}}\bepsilon\|^2-\|\bP_{\bX_{\tilde{\mathcal{G}}}}\bepsilon\|^2\Big)+\Big(\|\bP_{\bX_{\tilde{\mathcal{G}}}}\bepsilon\|^2-\|\bP_{\bX_{\mathcal{G}^*}}\bepsilon\|^2\Big),
\end{eqnarray}
where we write $\bP_{\bX_{\tilde{\mathcal{G}}\cup\mathcal{I}}}$ for the orthogonal projection onto $\text{span}(\bX_{\tilde{\mathcal{G}}},\bX_{\mathcal{I}})$. We will deal with the two terms on the right hand side of inequality~(\ref{eps.proj.reform}) separately.

\medskip

\textbf{Term 1}: $\|\bP_{\bX_{\tilde{\mathcal{G}}\cup\mathcal{I}}}\bepsilon\|^2-\|\bP_{\bX_{\tilde{\mathcal{G}}}}\bepsilon\|^2$.

Let $\mathcal{A}$ be the orthogonal complement of the space spanned by the columns of $\bX_{\tilde{\mathcal{G}}}$ within the space spanned by the columns of $\bX_{\tilde{\mathcal{G}}\cup\mathcal{I}}$ and note that
$$
\|\bP_{\bX_{\tilde{\mathcal{G}}\cup\mathcal{I}}}\bepsilon\|^2-\|\bP_{\bX_{\tilde{\mathcal{G}}}}\bepsilon\|^2=\|\bP_{\mathcal{A}}\bepsilon\|^2.
$$

Suppose that $r(\mathcal{G})+t(\mathcal{G})=d$ for some fixed~$d$, which implies $|\mathcal{I}|=d$. Suppose also that $m(\mathcal{G})=m$ and both~$\tilde{\mathcal{G}}$ and~$\mathcal{I}$ are fixed. Then, the space~$\mathcal{A}$ is fixed as well, and its dimension is at most~$d$. Hence, random variable $\|\bP_{\mathcal{A}}\bepsilon\|^2$ has a Chi-square distribution with at most~$d$ degrees of freedom. When $d=0$, we have $\|\bP_{\mathcal{A}}\bepsilon\|=0$ with probability one. When $d\ge 1$, standard Chi-square tail bounds imply that
\begin{equation}
\label{eps.proj.bnd}
\|\bP_{\mathcal{A}}\bepsilon\|^2\lesssim d(1+a)
\end{equation}
with probability at least $1-\exp(-2da)$, for each positive~$a$.

We now extend the above bound to all possible $\tilde{\mathcal{G}}$ and~$\mathcal{I}$ for which $r(\mathcal{G})+t(\mathcal{G})=d$ and $m(\mathcal{G})=m$, focusing on the case $d\ge1$. The total number of such sets~$\mathcal{I}$ is $p \choose d$, which is bounded above by $(ep/d)^d$. We note that each~$\tilde{\mathcal{G}}$ can be obtained by~$m$ successive binary splits of the~$K^*$ nonzero clusters in~$\mathcal{G}^*$. Thus, the total number of relevant clustering patterns~$\tilde{\mathcal{G}}$ is bounded above by $(K^*)^m (2^{c^*})^m$, where $(K^*)^m$ corresponds to the sequences of the nonzero $\mathcal{G}^*$-clusters involved in the~$m$ splits, and $2^{c^*}$ corresponds to the binary splits of a nonzero $\mathcal{G}^*$-cluster or its subset. Applying the union bound, we deduce that inequality~(\ref{eps.proj.bnd}) holds with probability at least $1-\exp\big(-2da+d\log(ep/d)+mc^*\log(K^*\vee 2)\big)$, uniformly over $\tilde{\mathcal{G}}$ and~$\mathcal{I}$ for which $r(\mathcal{G})+t(\mathcal{G})=d$ and $m(\mathcal{G})=m$. Taking $a=\log(ep/d)+d^{-1}mc^*\log(K^*\vee 2)+[2d]^{-1}\log(16/\epsilon_0)$ for a small positive~$\epsilon_0<1/2$, we derive that
\begin{equation*}
\max_{\mathcal{G}:\,r(\mathcal{G})=d,\,m(\mathcal{G})=m}\|\bP_{\mathcal{A}}\bepsilon\|^2\lesssim d\log(p)+mc^*\log(K^*\vee 2)+\log(1/\epsilon_0)
\end{equation*}
with probability at least $1-(d/[ep])^d\big(1/[K^*\vee 2]\big)^{mc^*}(\epsilon_0/16)$.

Applying the union bound over the different values of $d=1,..., p$ and $m=0,..., p$, and recalling that $\|\bP_{\mathcal{A}}\bepsilon\|\equiv0$ when $d=0$, we then conclude that
\begin{equation}
\label{Term1.bnd}
\|\bP_{\bX_{\tilde{\mathcal{G}}\cup\mathcal{I}}}\bepsilon\|^2-\|\bP_{\bX_{\tilde{\mathcal{G}}}}\bepsilon\|^2 = \|\bP_{\mathcal{A}}\bepsilon\|^2\lesssim [r(\mathcal{G})+t(\mathcal{G})]\log(p)+m(\mathcal{G})c^*\log(K^*\vee 2)+\log(1/\epsilon_0)
\end{equation}
uniformly over~$\mathcal{G}$, with probability at least
$$
1-\sum_{m=0}^p\sum_{q=1}^p \big(q/[ep]\big)^q\big([1/K^*]\wedge [1/2]\big)^{mc^*}(\epsilon_0/16) \ge 1- \sum_{q=1}^{\infty} (1/e)^q\,\sum_{m=0}^{\infty}(1/2)^m\,(\epsilon_0/16) \ge 1-\epsilon_0/8.
$$

\smallskip

\textbf{Term 2}: $\|\bP_{\bX_{\tilde{\mathcal{G}}}}\bepsilon\|^2-\|\bP_{\bX_{\mathcal{G}^*}}\bepsilon\|^2$.

Suppose that $m(\mathcal{G})=m$ and recall that we can obtain~$\tilde{\mathcal{G}}$ from~$\mathcal{G}^*$ by~$m$ successive binary splits of the nonzero clusters in~$\mathcal{G}^*$.

Term 2 can be handled similarly to Term 1 by (a) noting that, for a fixed~$\tilde{\mathcal{G}}$, random variable $\|\bP_{\bX_{\tilde{\mathcal{G}}}}\bepsilon\|^2-\|\bP_{\bX_{\mathcal{G}^*}}\bepsilon\|^2$ can be bounded above by a Chi-square with at most $mc^*$ degrees of freedom; (b) using the previously derived bound on the total number of possible~$\tilde{\mathcal{G}}$ corresponding to $m(\mathcal{G})=m$; and (c) applying the union bound twice. We can conclude that
\begin{equation}
\label{Term2.bnd}
\|\bP_{\bX_{\tilde{\mathcal{G}}}}\bepsilon\|^2-\|\bP_{\bX_{\mathcal{G}^*}}\bepsilon\|^2\lesssim m(\mathcal{G})c^*\log(K^*\vee 2)+\log(1/\epsilon_0)
\end{equation}
uniformly over~$\mathcal{G}$, with probability at least $1-\epsilon_0/8$.

Consequently, combining bounds~(\ref{Term1.bnd}) and~(\ref{Term2.bnd}) for Terms~1 and~2 together with inequality~(\ref{eps.proj.reform}), we deduce that
\begin{equation}
\label{term1.final.bnd1}
\bepsilon^{\top}(\bP_{\bX_{\widehat{\mathcal{G}}}}-\bP_{\bX_{\mathcal{G}^*}})\bepsilon\lesssim (\widehat{r}+\widehat{t})\log(p)+\widehat{m}c^*\log(K^*\vee 2)+\log(1/\epsilon_0)
\end{equation}
holds with probability at least $1-\epsilon_0/4$.

We let~$\mathcal{J}$ denote the index set of the joint support of~$\bbeta$ and~$\bbeta^*$, i.e., $\mathcal{J}=S(\bbeta)\cup S(\bbeta^*)$. Then, revisiting inequality~(\ref{eps.proj.reform}), we note that
\begin{equation}
\label{eps.proj.reform2}
\bepsilon^{\top}(\bP_{\bX_{\mathcal{G}}}-\bP_{\bX_{\mathcal{G}^*}})\bepsilon \le \|\bP_{\bX_{\mathcal{G}\cup\mathcal{J}}}\bepsilon\|^2-\|\bP_{\bX_{\mathcal{G}^*}}\bepsilon\|^2.
\end{equation}
Let~$s=\|\bbeta\|_0$. We can now follow the argument we used for Term~1 by (a) noting that, for a fixed~$\mathcal{J}$, random variable $\|\bP_{\bX_{\mathcal{G}\cup\mathcal{J}}}\bepsilon\|^2-\|\bP_{\bX_{\mathcal{G}^*}}\bepsilon\|^2$ can be bounded above by a Chi-square with at most $s+s^*$ degrees of freedom; (b) bounding the total number of possible~$\mathcal{J}$ by $(ep/s)^s$; and (c) applying the union bound twice. We conclude that
\begin{equation}
\label{term1.final.bnd2}
\bepsilon^{\top}(\bP_{\bX_{\widehat{\mathcal{G}}}}-\bP_{\bX_{\mathcal{G}^*}})\bepsilon\lesssim (\widehat{s}+s^*)\log(p)+\log(1/\epsilon_0)
\end{equation}
holds with probability at least $1-\epsilon_0/4$.

Combining bounds~(\ref{term1.final.bnd1}) and~(\ref{term1.final.bnd2}), we deduce that inequality (i) of Lemma~\ref{max.ineq.lem1} holds with probability at least $1-\epsilon_0/2$.

\medskip

\textbf{Inequality (ii) of Lemma~\ref{max.ineq.lem1}}. We let
$$
\widehat{\balpha}=\frac{(\bI_n-\bP_{\bX_{\widehat{\mathcal{G}}}})\bff^*}{\|(\bI_n-\bP_{\bX_{\widehat{\mathcal{G}}}})\bff^*\|}
$$
and note that $\|(\bI_n-\bP_{\bX_{\widehat{\mathcal{G}}}})\bff^*\|^2=n\widehat{w}$. It follows that
\begin{equation*}
|\bepsilon^{\top}(\bI_n-\bP_{\bX_{\widehat{\mathcal{G}}}})\bff^*|\le |\bepsilon^{\top}\widehat\balpha| \sqrt{n\widehat{w}}.
\end{equation*}
We will consider vectors~$\balpha$ of the form
$$
\balpha=\frac{(\bI_n-\bP_{\bX_{\mathcal{G}}})\bff^*}{\|(\bI_n-\bP_{\bX_{\mathcal{G}}})\bff^*\|}
$$
and uniformly bound $|\bepsilon^{\top}\balpha|$ over clustering patterns~$\mathcal{G}$.

If~$\mathcal{G}$ is fixed, then~$\balpha$ is fixed as well, and random variable $|\bepsilon^{\top}\balpha|^2$ has a Chi-square distribution with one degree of freedom, satisfying a chi-square tail bound. From the Term~1 part in the proof of Lemma~\ref{max.ineq.lem1} we have the following bound on the total number of clustering patterns~$\mathcal{G}$ with $r(\mathcal{G})+t(\mathcal{G})=d$ and $m(\mathcal{G})=m$ for fixed~$d$ and~$m$:
$$
\left(\frac{ep}{d}\right)^d \times \left(\big[K^*\big] \big[2^{c^*}\big]\right)^m.
$$
Proceeding as in the Term~1 part in the proof of Lemma~\ref{max.ineq.lem1}, and repeating the corresponding argument with minor modifications, we deduce that
\begin{equation*}
|\bepsilon^{\top}\widehat\balpha|\lesssim \sqrt{\big(\widehat{r}+\widehat{t}\big)\log(p)+\widehat{m}c^*\log(K^*\vee 2)+\log(1/\epsilon_0)}
\end{equation*}
and, hence,
\begin{equation}
\label{lem1.ineq2.fin.bnd1}
|\bepsilon^{\top}(\bI_n-\bP_{\bX_{\widehat{\mathcal{G}}}})\bff^*|\le  \sqrt{\big(\widehat{r}+\widehat{t}\big)\log(p)+\widehat{m}c^*\log(K^*\vee 2)+\log(1/\epsilon_0)} \sqrt{n\widehat{w}}
\end{equation}
with probability at least $1-\epsilon_0/4$.

Following the argument below inequality~(\ref{term1.final.bnd1}), we also derive 
\begin{equation}
\label{lem1.ineq2.fin.bnd2}
|\bepsilon^{\top}(\bI_n-\bP_{\bX_{\widehat{\mathcal{G}}}})\bff^*|\le  \sqrt{\big(\widehat{s}+s^*\big)\log(p)+\log(1/\epsilon_0)} \sqrt{n\widehat{w}}
\end{equation}
with probability at least $1-\epsilon_0/4$.

Combining bounds~(\ref{lem1.ineq2.fin.bnd1}) and~(\ref{lem1.ineq2.fin.bnd1}), we conclude that inequality (ii) of Lemma~\ref{max.ineq.lem1} holds with probability at least $1-\epsilon_0/2$.
\end{proof}

\begin{proof}[Proof of Lemma~\ref{lem.BA}]

We will show that $\mathcal{E}_j\gtrsim r_j$ uniformly over $j\in[q]$. Recall that we break $\mathcal{E}(\cG)$ down by the categorical predictors, writing $\mathcal{E}(\cG)=\sum_{j=1}^{q}\mathcal{E}_j$,  where
$$
\mathcal{E}_j=\sum_{a\in\{\tilde{\beta}_j:\,j\in\cI_j\}} \Big(\big|\{l\in\cI_j: {\tilde{\beta}}_l=a\}\big|-\max_{b\in\{\tilde{\beta}_j^*:\,j\in\cI_j\}}\big|\{l\in\cI_j: {\tilde{\beta}}_l=a,\,\tilde{\beta}^*_l=b\}\big|\Big).
$$

To simplify the exposition, we first focus on the case $S(\bbeta_j)=S(\bbeta^*_j)$, where~$\bbeta_j$ and~$\bbeta^*_j$ are the sub-vectors of~${\bbeta}$ and~${\bbeta}^*$, respectively, corresponding to the $j$-th categorical predictor. Afterwards, we discuss the extension to the general case.

For concreteness, suppose that $\cG=\cG(\bbeta)$ and $\cG^*=\cG(\bbeta^*)$.
Let~$\mathcal{G}_j$ and~$\mathcal{G}^*_j$ be the subsets of~$\mathcal{G}$ and~$\mathcal{G}^*$, respectively, corresponding to the $j$-th categorical predictor. Consider the following procedure for converting $\mathcal{G}_j$ into $\mathcal{G}^*_j$. Note that Stage~2 is not required in the special case that we focus on, as the zero clusters in~$\cG$ match the ones in~$\cG^*$ and hence we will ignore them. For convenience, we reverse the order of stages~S1 and~S3 in the description below. However, we can easily represent this procedure in the form of Definition~\ref{DefinitionB} without changing the counts of merges and reassignments.
\begin{itemize}

\item[(A)] If $|\mathcal{G}_j|>|\mathcal{G}^*_j|$, we start by merging two (nonzero) clusters of $\mathcal{G}_j$ that have the same (nonzero) $\mathcal{G}_j^*$-cluster in the majority (note that we can do this when $|\mathcal{G}_j|>|\mathcal{G}^*_j|$); we continue this merging process until $|\mathcal{G}_j|=|\mathcal{G}^*_j|$. Note that the number of variables in the minority, i.e., $\mathcal{E}_j$, does not change after these merges. Thus, we can focus on the case the case $|\mathcal{G}_j|\le|\mathcal{G}^*_j|$ for the remainder of the proof.

\item[(B)] $|\mathcal{G}_j|\le|\mathcal{G}^*_j|$. \textbf{First} step: we reassign each variable into the $\mathcal{G}_j$-cluster where the $\mathcal{G}_j^*$-cluster of this variable is in the majority (if several such $\mathcal{G}_j$-clusters exist, then we choose the one with the largest majority). \textbf{Second} step: we repeat this process until no further reassignments are possible. Note that the first two steps may result in some empty clusters.
    With the final (\textbf{third}) step, we take all the variables that remain in the minority and reassign them into new clusters (either empty or completely new), separating the variables by their $\mathcal{G}_j^*$-cluster membership.

\end{itemize}

For illustration, consider the following example. Suppose that~$\mathcal{G}_j^*$ has four clusters which we refer to as ``$a$'',``$b$'',``$c$'', and ``$d$''. Suppose that~$\mathcal{G}_j$ also has four clusters, with the following $\mathcal{G}_j^*$-cluster membership:
\begin{equation}
\label{example.start.alloc}
T\mathcal{G}_j: \qquad \{\textbf{4a},2b\},\,\{{3a},2c,1b\},\,\{{2d},1c\},\,\{\textbf{3d},1b,1c\}.
\end{equation}
The number in front of each letter refers to the number of variables with the corresponding $\mathcal{G}_j^*$-cluster membership. For example, the first cluster in~$\mathcal{G}_j^*$ has $4$ variables in the``$a$'' cluster of~$\mathcal{G}_j^*$ and~$2$ variables in the ``$b$'' cluster. Note that we used the {\bf{bold}} font for the variables in the (largest) majority with respect to the $\mathcal{G}_j^*$-cluster membership, i.e., the variables that will attract other variables from the same $\mathcal{G}_j^*$-cluster (if any) at the next step of the procedure.

Clustering pattern~(\ref{example.start.alloc}) is the starting point of our procedure for converting~$\mathcal{G}_j$ into~$\mathcal{G}_j^*$.  The two steps in part~(B) of the procedure  yield the following result:
$$
~~~~~~~~~~~~~~~\text{After step 1:} \qquad \{\textbf{7a},2b\},\,\{\textbf{2c},1b\},\,\{{1c}\},\,\{\textbf{5d},1b,1c\}.
$$

$$
~~~~~~~\text{After step 2:} \qquad \{\textbf{7a},2b\},\,\{\textbf{4c},1b\},\,\emptyset,\,\{\textbf{5d},1b\}.
$$

$$
\text{After step 3:} \qquad \{\textbf{7a}\},\,\{\textbf{4c}\},\,\{\textbf{4b}\},\,\{\textbf{5d}\}.
$$

We denote the number of reassignments in the above procedure by~$\tilde{r}_j$, and observe that $\tilde{r}_j\ge r_j$ because, by definition, $r_j$ is the smallest possible number of reassignments.

Thus, it is only left to show that $\tilde{r}_j\lesssim \mathcal{E}_j$, which is done below.

{\bf Proving} $\mathbf{\tilde{r}_j}\boldsymbol{\lesssim \mathcal{E}_j}$. We make the following general observations about our conversion procedure: the variables that are originally in (any) majority with respect to the $\mathcal{G}_j^*$-cluster membership may get reassigned during step~1 but do not get reassigned afterwards. In the example, these are the ``$a$'' and the ``$d$'' variables, and they only get reassigned during step~1.

We note that we can bound~$\tilde{r}_j$ by the sum of~$\mathcal{E}_j$ and the number of variables (say, $\tilde{s}_j$) that are in the majority at the start of the procedure but still get reassigned during the first step.

Consider a cluster, say $G$, where the variables in the majority get reassigned during step~1. The number of such reassignments for cluster~$G$ is at most $c_{\text{max}}$, by which we denote the size of the largest true nonzero cluster. The variables that end up in~$G$ at the end of the procedure are either (a) already in (the minority of) cluster~$G$ at the start or (b) get reassigned after step~1 and hence are not in the majority in any cluster of~$\mathcal{G}_j$. Thus, all these variables are included in the~$\mathcal{E}_j$ count. Note that the number of such variables is at least $c_{\text{min}}$, by which we denote the size of the smallest true cluster.

Let's say there are~$N$ such clusters~$G$ at the start of the procedure. Then, according to the paragraph above, there are at most~$Nc_{\text{max}}$ majority variables that get reassigned during the procedure, however, there are also at least~$Nc_{\text{min}}$ minority variables. Consequently,
$$
 \tilde{r}_j\le \mathcal{E}_j+\tilde{s}_j \le \mathcal{E}_j + Nc_{\text{max}}\le  \mathcal{E}_j +  \frac{c_{\text{max}}}{c_{\text{min}}} Nc_{\text{min}} \le  \mathcal{E}_j +  \frac{c_{\text{max}}}{c_{\text{min}}} \mathcal{E}_j \lesssim \mathcal{E}_j.
$$
We conclude that $\tilde{r}_j\lesssim \mathcal{E}_j$ by noting that the multiplicative constant in the last inequality is determined by $c_{\text{max}}/c_{\text{min}}$ and, hence, is independent of~$j$.

The general case, where we may have $S(\bbeta_j)\ne S(\bbeta^*_j)$ follows by analogous arguments, accounting for the fact that some variables may need to be reassigned from and to the zero cluster. The main change to the proof occurs in part (A) of our conversion procedure, where we should not merge two $\mathcal{G}_j$-clusters that have the true zero cluster in the majority -- these clusters will be zeroed out in Stage~S2 of the original procedure. However, this change has no effect on the number of reassignments and, hence, is inconsequential.
\end{proof}

\section{Proof of Corollary~\ref{cor.univ.dense}}

\begin{proof}[Proof of Corollary~\ref{cor.univ.dense}]
    As we have shown after the statement of Theorem~\ref{th.BA}, in the univariate case condition~(\ref{cond.BA})  and (consequently) condition~(\ref{assumptionA.bnd}) are satisfied. Thus, we just need to check that Theorem~\ref{clust.recov.thm} holds when~$0\le\lambda_0\lesssim\log(p)/n$ and the sizes of \emph{all} the clusters are of order~$c^*$ (we note that the true zero clusters are at least as large as the corresponding nonzero clusters). We can verify this fact by repeating the proof of Theorem~\ref{clust.recov.thm} with the following simplifications:
    (a) remove stage S2 from the conversion procedure in Section~\ref{sec.prf.prelim}; (b) remove the word ``nonzero'' from stage S3 of the same procedure, thus allowing merges of any clusters; and (c)~set $\widehat{t}=0$ in the inequalities of Lemma~\ref{max.ineq.lem1}, and noting that Lemma~\ref{lem.BA} does not require sparsity and continues to hold under the changes in~(a) and~(b) as shown in the main part of its proof where we assume $S(\bbeta)=S(\bbeta^*)$.
 \end{proof}

\section{Proof of Proposition~\ref{prop.rand}}

\begin{proof}[Proof of Proposition~\ref{prop.rand}]
We consider two possible cases: $L\ge (4s^*)^2$ and $L< (4s^*)^2$. 

\textbf{Case 1:} $L\ge (4s^*)^2$.

Write $D_{1j},...,D_{Lj}$ for the dummy variables corresponding to predictor $j$, and let
 $$
\tilde{\B x} = \sqrt{L}\Big(\frac1{\sqrt{L}},D_{11},...,D_{L1},\ldots,D_{1q},...,D_{Lq}\Big)^\top.
 $$
Note that
\begin{align*}
&E[D_{kj}]=E[D_{kj}^2]=1/L & \text{for} & \; k\in[L], \; j\in[q]\\
&E[D_{kj}D_{lj}]=0 & \text{for}& \; k,l\in[L] \; \text{s.t.} \;  k\ne l, \; j\in[q]\\
&E[D_{kj}D_{lm}]=1/L^2 & \text{for}& \; k,l\in[L],\; j,m\in[q] \;\text{s.t.}\, j\ne m.
\end{align*}
Hence, letting $A=E[\tilde{\B x} \tilde{\B x}^\top]$, we conclude that $\text{diag}(A)=\B 1$; all the off-diagonal elements of~$A$ in both the first row and the first column equal $1/\sqrt{L}$; and all the other off-diagonal elements of~$A$ are either zero or $1/L$. Because $2s^*/\sqrt{L}\le 1/2$ under the Case 1 setting, we conclude that~$A^{1/2}$ satisfies the following condition:
\begin{equation}
\label{SE.cond}
\|A^{1/2}\balpha\|_2^2\ge (1/2)\|\balpha\|_2^2 \qquad \text{for all} \; \balpha\in\mathbb{R}^{p+1} \;\text{s.t.}\; \|\balpha\|_0\le 2s^*+1.
\end{equation}
Let~$\tilde{\B X}_{n\times (p+1)}$ be a random matrix whose rows are independent realisation of the vector~$\tilde{\B x}$.
By~(\ref{SE.cond}) and the random matrix theory results in, for example, \cite{lecue2017sparse}, condition
\begin{equation}
\label{SE.cond.sam}
\|\tilde{\B X}\bgamma\|_2^2\gtrsim n\|\bgamma\|_2^2 \qquad \text{for all} \; \bgamma\in\mathbb{R}^{p+1} \;\text{s.t.}\; \|\bgamma\|_0\le 2s^*+1
\end{equation}
holds with high probability when $n\gtrsim \log(p)$.

Let
 $$
 \B {x} = \Big(D_{11},...,D_{L1},\ldots,D_{1q},...,D_{Lq}\Big)^\top
 $$
and note that random matrix~$\B {X}_{n\times p}$ whose rows are independent realisation of the vector~$\B {x}$ is exactly the usual model matrix~$\B X$ in the random setting of the proposition. Applying~(\ref{SE.cond.sam}) and noting that $n=n^*L$, we derive
$$
\|\tau\B1+\B X\bbeta\|_2^2\gtrsim n\|\bbeta/\sqrt{L}\|_2^2 =n^*\|\bgamma\|_2^2 \qquad \text{for all} \; \tau \in\mathbb{R},\; \bbeta\in\mathbb{R}^{p} \;\text{s.t.}\; \|\bbeta\|_0\le 2s^*.
$$
The claim of the proposition follows directly.

\textbf{Case 2:} $L< (4s^*)^2$.

Note that both $s^*$ and~$L$ are bounded in this setting. We remove the dummy variable for one arbitrarily chosen (baseline) level for each categorical variable, and then center each random dummy variable by subtracting its expected value. The covariance matrix for the reduced set of dummy variables is block-diagonal, where each block is a well-behaved $(L-1)\times (L-1)$ matrix with the smallest eigen value bounded away from zero and the largest bounded above by one. 

Note that we can still represent the original linear predictions using the reduced set of (centered) dummy variables and the intercept. Coefficient vectors that were originally $s^*$-sparse become denser in the new representation, however, their $\ell_0$ norms remain uniformly bounded -- this is because categorical predictors that were originally excluded from the prediction (at least $q-s^*$ of them) remain excluded in the new representation and hence the $\ell_0$ norms of the new coefficient vectors are at most~$s^*L$, which is bounded in Case~2. Thus, arguing analogously to Case 1, we conclude that the claim of the proposition holds.
\end{proof}

\section{Other Proofs}

\subsection{Proof of Proposition~\ref{mip-prop}}
\begin{lemma}\label{mip-lemma}
    Let $\B\beta\in\R^p$. Then,
    $$|\{\beta_1,\cdots,\beta_p\}|=1+\sum_{i=2}^p \1(\beta_i\notin\{\beta_1,\cdots,\beta_{i-1}\}).$$
\end{lemma}
\begin{proof}[Proof of Lemma~\ref{mip-lemma}]
Suppose $|\{\beta_1,\cdots,\beta_p\}|=K$. Let $\beta_{(1)},\cdots,\beta_{(K)}$ be the $K$ distinct values that $\B\beta$ takes, i.e. $\{\beta_1,\cdots,\beta_p\}=\{\beta_{(1)},\cdots,\beta_{(K)}\}$. Moreover, let $i_k$ for $k\in[K]$ be the first index where $\beta_{i_k}=\beta_{(k)}$.

Note that 
\begin{align}
  &  \sum_{i=2}^p \1(\beta_i\notin\{\beta_1,\cdots,\beta_{i-1}\})\nonumber \\
    =&\sum_{k=2}^K \underbrace{\1(\beta_{i_k}\notin\{\beta_1,\cdots,\beta_{i_k-1}\})}_{=1} + \sum_{\substack {i\in[p] \\ i\notin \{i_k\}}}\underbrace{ \1(\beta_i\notin\{\beta_1,\cdots,\beta_{i-1}\})}_{=0} = K-1.
\end{align}
Therefore,
$$|\{\beta_1,\cdots,\beta_p\}| = K = 1+\sum_{i=2}^p \1(\beta_i\notin\{\beta_1,\cdots,\beta_{i-1}).$$
\end{proof}

\begin{proof}[\textbf{{Proof of Proposition~\ref{mip-prop}}}.]
    The proof consists of two parts. \\
    
    \textbf{Part 1:} First, suppose $\alpha,\B\beta,\B{z},\B{l}$ is feasible for Problem~\eqref{clusteringproblem-mip} implying $\B\beta$ is feasible for~\eqref{clusteringproblem-reg}. We show that 
    $$\lambda_0\|\B\beta\|_0+\lambda \sum_{j=1}^q |\{\beta_i:i\in\cI_j\}|\leq \lambda_0\sum_{i=1}^p z_i+\lambda\sum_{i=1}^p l_i.$$
    Note that if $|\beta_i|>0$ for some $i\in[p]$, then we must have $z_i=1$ which shows $\sum_{i=1}^p z_i\geq \|\B\beta\|_0$. Next, fix $j\in[q]$. Note that 
    $$l_{s_j}\geq 1\Rightarrow l_{s_j}=1.$$
    For $i\in\{s_j+1,\cdots,s_j+p_j-1\}$ and $k\in\{s_j,\cdots,i-1\}$, if $\beta_k\neq \beta_i$, then $z^j_{i,k}=1$. Therefore, 
    if $\beta_i\notin\{\beta_{s_j},\cdots,\beta_{i-1}\}$, then $z^j_{i,k}=1$ for all $k=s_j,\cdots,i-1$, and we have:
    $$\sum_{k=s_j}^{i-1}z^j_{i,k}-(i-s_j-1)=i-s_j-(i-s_j-1)=1=\1(\beta_i\notin\{\beta_{s_j},\cdots,\beta_{i-1}\}).$$
    This implies $l_i\geq \1(\beta_i\notin\{\beta_{s_j},\cdots,\beta_{i-1}\})$ as $l_i\geq \sum_{k=s_j}^{i-1}z^j_{i,k}-(i-s_j-1)$ and $l_i\in\{0,1\}$. Hence,
    \begin{align}
        |\{\beta_i:i\in\cI_j\}| & \stackrel{(a)}{=} 1 + \sum_{i=s_{j}+1}^{s_j+p_j-1}\1(\beta_i\in\{\beta_{s_j}\cdots,\beta_{i-1}\})\nonumber \\
        & \leq \sum_{i=s_j}^{s_j+p_j-1}l_i.
    \end{align}
       where $(a)$ is by Lemma~\ref{mip-lemma}. Thus,
        \begin{align}
       \sum_{j\in[q]} |\{\beta_i:i\in\cI_j\}| \leq \sum_{j\in[q]}\sum_{i\in\cI_j}l_i \leq \sum_{i=1}^{p}l_i.
    \end{align}
    This shows that for any feasible solution of~\eqref{clusteringproblem-mip}, there exists a feasible solution to~\eqref{clusteringproblem-reg} with the same or lower objective.\\
    \textbf{Part 2:} Suppose $\hat\alpha\in\R$ and $\hat{\B\beta}\in\R^p$ is an optimal solution to~\eqref{clusteringproblem-reg} and take $M\geq \max_{i\in [p]}|\hat{\beta}_i|$. Let $\hat{z}_i=1$ if $\hat{\beta}_i\neq 0$ (and $\hat{z}_i=0$ otherwise). Similarly, let $\hat{z}^j_{i,k}=1$ if $\hat{\beta}_i\neq \hat{\beta}_k$ (and $\hat{z}^j_{i,k}=0$ otherwise). Finally, if $i\in\cI_j$ for some $j\in[q]$, let 
    $$\hat{l}_i = 0\lor \left\{\sum_{k=s_j}^{i-1}\hat{z}^j_{i,k}-(i-s_j-1)\right\}\in\{0,1\}.$$
    Otherwise, let $\hat{l}_i=0$. By construction, $\hat{\B\beta},\hat{\B{z}},\hat{\B{l}}$ is feasible for~\eqref{clusteringproblem-mip}. 
    We have that $\hat{l}_i=1$ if $\hat{\beta}_i\notin\{\hat{\beta}_{s_j},\cdots,\hat{\beta}_{i-1}\}$ implying $\hat{l}_i=\1(\hat{\beta}_i\notin\{\hat{\beta}_{s_j},\cdots,\hat{\beta}_{i-1}\})$. Hence, by Lemma~\ref{mip-lemma},
    $$\sum_{i=1}^p\hat{l}_i=\sum_{j=1}^q\sum_{i=s_j}^{s_j+p_j-1}\hat{l}_i=\sum_{j=1}^q\sum_{i=s_j}^{s_j+p_j-1}\1(\hat{\beta}_i\notin\{\hat{\beta}_{s_j},\cdots,\hat{\beta}_{i-1}\})=\sum_{j=1}^q|\{\hat{\beta}_i:i\in\cI_j\}|.$$
    Moreover, we note that $\sum_{i=1}^p\hat{z}_i=\|\hat{\B\beta}\|_0$.
    Therefore,
        $$\lambda_0\|\hat{\B\beta}\|_0+\lambda \sum_{j=1}^r |\{\hat{\beta}_i:i\in\cI_j\}|= \lambda_0\sum_{i=1}^p \hat{z}_i+\lambda\sum_{i=1}^p \hat{l}_i.$$
        This shows that for any optimal solution to~\eqref{clusteringproblem-reg}, there exists a feasible solution to~\eqref{clusteringproblem-mip} with the same objective.

       From the two parts above, we have that~\eqref{clusteringproblem-reg} and~\eqref{clusteringproblem-mip} are equivalent as long as $M\geq \max_{i\in [p]}|\hat{\beta}_i|$.
\end{proof}

\subsection{Proof of Proposition~\ref{mip-lower-prop}}
\begin{proof}[Proof of Proposition~\ref{mip-lower-prop}]
  Note that to obtain Problem~\eqref{l0clustering-mip-relaxed}, we have removed constraints from Problem~\eqref{clusteringproblem-mip}, which completes the proof of the first part.

  We now focus on the second part of the proposition. Suppose $\hat{\B l},\hat{\B z},\hat{\B \beta},\hat{\alpha}$ is the optimal solution to~\eqref{l0clustering-mip-relaxed} for some $T\geq 1$. From the first part we have that
  $\widehat\obj_{T+1}\leq \obj$. Suppose $S(\hat{\B\beta})=S(\B\beta^{(t)})$ for some $t\in[T]$. We claim that $\hat{\B l},\hat{\B z},\hat{\B \beta},\hat{\alpha}$ is also feasible for Problem~\eqref{clusteringproblem-mip}. To show feasibility, we only need to make sure $\hat{\B l},\hat{\B z},\hat{\B \beta},\hat{\alpha}$ satisfies the constraints we removed from~\eqref{clusteringproblem-mip} to obtain~\eqref{l0clustering-mip-relaxed}. Specifically, for $j\notin S_t$ and $i,k\in\cI_j$,
  $$|\hat{\beta}_i-\hat{\beta}_k|=0\leq 2M\hat{z}^j_{i,k}.$$
  This completes the proof of feasibility, implying $\widehat\obj_{T+1}\leq \obj \leq \widehat\obj_{T+1}$ that completes the proof.
\end{proof}

\subsection{Proof of Lemma~\ref{classification-lemma}}
\begin{proof}[Proof of Lemma~\ref{classification-lemma}]
Let 
$$h(x)=\log(1+\exp(-x)).$$
Then, 
$$h'(x) = -\frac{\exp(-x)}{1+\exp(-x)},~~~~~~~~~h''(x)=\frac{\exp (x)}{(\exp(x)+1)^2}>0.$$
Therefore, $h(x)$ is convex and $0<h''(x)\leq 1/4$. Invoking the descent lemma~\citep{bertsekas2016nonlinear}, for any $u,v$
\begin{align}\label{eq:descent-logstic}
    h(u) -h(v) & \leq  -\frac{\exp(-v)}{1+\exp(-v)}(u-v) + \frac{1}{8}(u-v)^2 \\
    & = \frac{1}{8}\left(u - \left[v +\frac{4\exp(-v)}{1+\exp(-v)}\right]\right)^2 - \frac{1}{8}\left(\frac{4\exp(-v)}{1+\exp(-v)}\right)^2.
\end{align}

Therefore, 
    for the logistic loss $ L(y^{(i)},\B{x}^{(i)};\B\beta,\alpha) = \log(1+\exp(-y^{(i)}[\B\beta^\top\B{x}^{(i)}+\alpha]))$ with $\alpha=\alpha_0=0$, by taking $u= y^{(i)}\B\beta^\top \B x^{(i)}$ and $v= y^{(i)}\B\beta^\top_0 \B x^{(i)}$, we have
    \begin{align}
      &  L(y^{(i)},\B{x}^{(i)};\B\beta)-L(y^{(i)},\B{x}^{(i)};\B\beta_0)\nonumber \\
         = & h(u)-h(v) \nonumber \\
        \leq &\frac{1}{8}\left(\B\beta^\top x^{(i)} -\left[\B\beta_0^\top\B x^{(i)}+\frac{4y^{(i)} \exp(-y^{(i)}\B\beta_0^\top\B x^{(i)})}{1+\exp(-y^{(i)}\B\beta_0^\top\B x^{(i)})}\right]\right)^2 - \frac{1}{8}\left(\frac{4y^{(i)} \exp(-y^{(i)}\B\beta_0^\top\B x^{(i)})}{1+\exp(-y^{(i)}\B\beta_0^\top\B x^{(i)})}\right)^2.
    \end{align}
The proof is complete by summing over $i\in[n]$. The general case with $\alpha,\alpha_0\neq 0$ is similar.
\end{proof}

\section{Additional Details for the Numerical Experiments}\label{app:numerical}

We tune the methods as follows:

\begin{itemize}
    \item~SCOPE: We tune the parameter $\lambda$ over 100 different values in the range $[10^{-5},10]$. Unless stated otherwise, we use the default parameter $\gamma$.
    \item~Elastic Net: We tune $\lambda_1,\lambda_2$ over a $10\times 10$ grid, over the range $[10^{-5},10]$.
    \item~Lasso: We tune $\lambda_1$ over a grid of size 25 over the range $[10^{-5},10]$.
    \item~IHT: We tune the parameter $\lambda_0$ over 100 different values in the range $[10^{-5},10]$.
    \item~\ourmethod: We tune the parameter $\lambda$ over 100 different values in the range $[10^{-5},10]$.
    \item~\ourmethodlz: We tune $\lambda,\lambda_0$ over a $10\times 10$ grid, over the range $[10^{-5},10]$.
\end{itemize}

For the experiments in Section~\ref{sec:time-bench}, we set the Big-M value as $M=1.2\max_{i\in[p]}|\hat{\beta}_i|$, where $\hat{\B\beta}$ is a solution available from our approximate solver.

In terms of computational resources, the experiments in Appendices~\ref{numerical-results-runtime} and~\ref{supp:dpseg-bench} are performed on an M1 MacBook Air, while the experiments in Section~\ref{sec:realdata} on performed on a machine equipped with Intel Xeon Platinum 8260 CPU using 4GB of RAM.

\subsection{Numerical Results from Section~\ref{sec:time-bench}}\label{numerical-results-runtime}

Here, we present numerical results corresponding to the runtime and algorithm discussion in Section~\ref{sec:time-bench}. 

Table~\ref{table:runtimebenchmark} shows runtime for the approximate methods/algorithms.  For the experiments in this table, we observe that the final solutions from our methods (\ourmethodlz~and \ourmethod) have an average purity higher than 0.99.

Next, under the same setting, we fix $r_2=550$ and vary the value of $q$ until the runtime of \ourmethodlz~reaches around 5 minutes. We report the runtime and model quality metrics for this experiment in Table~\ref{table:runtimebenchmark-approx2}. We see that we can fit a path of solutions with $q$ to around 25 in this case using \ourmethodlz~(as well as with \ourmethod~which is even faster).

Next, we perform numerical experiments to demonstrate the quality of the solutions from our approximate solver. To this end, we consider the same setup as above, however, we set $\lambda=\lambda_0=0.05$. We compare the objective value of the solution from our approximate (BCD) solver to the objective value from exactly solving Problem~\eqref{clusteringproblem-mip} using Gurobi. More specifically, we set a large value of Big-M, $M=10$ (to ensure the equivalence of Problems~\eqref{clusteringproblem-mip} and~\eqref{clusteringproblem-reg}) and solve Problem~\eqref{clusteringproblem-mip} to an optimality gap of $0.1\%$. We report the average of the objective value from these two solutions (and their standard error) in Table~\ref{table:benchmarkquality}. We also report the average relative objective gap between approximate and exact solutions, $|\text{obj}_{\text{approx}}-\text{obj}_{\text{exact}}|/\text{obj}_{\text{approx}}$. We see that in general, approximate and exact solutions have close objective values, which demonstrates that our approximate solver is able to deliver high quality solutions. We also compared the prediction errors for approximate and exact solutions, and we observed the prediction errors were within $0.1\%$ of each other on average.

For the experiments with exact solvers (Table~\ref{table:runtimebenchmark-exact}), we set $n=500$, $r_2=14$ and consider the best hyperparameter setting selected via our approximate solver (available by validation tuning).
If the MIP solver is not able to achieve the optimality gap of $0.5\%$ in under~15 minutes, we stop the solver and report the final optimality gap. We note that our row generation procedure is guaranteed to converge after finitely many iterations. In practice, for the experiments in Table~\ref{table:runtimebenchmark-exact}, we observe that our row generation algorithm terminates after two iterations at most, most often after just one iteration. Providing convergence guarantees (beyond the worst case) for our row generation remains an open question.

\begin{table}[h!]
    \centering
    \caption{{\bf Approximate Algorithm Runtimes}: Running time benchmarks (in seconds) from Section~\ref{sec:time-bench} for approximate solvers. We compare our approximate solver with SCOPE. $q$ is the number of (categorical) predictors, each having $(r_2+4)$-many levels.  We report the average and the standard error (in the parenthesis).}
    \label{table:runtimebenchmark}
    {\scalebox{.99}{\begin{tabular}{ccc|c|c|c}
    \toprule
  $r_2$ & $q$ & $p$ & SCOPE & \ourmethod & \ourmethodlz  \\
  & & & Time [s]& Time [s] & Time [s]  \\
\midrule
\multirow{2}{*}{200} & 2 & 408 & $9.4(\pm0.1)$    & $<1$   & $<1$   \\
  & 3 & 612 & $18.3(\pm0.6)$ & $<1$ & $<1$  \\
  \midrule
\multirow{2}{*}{300} & 2 & 608 &  $30.8(\pm2.0)$  & $<1$  & $1.2(\pm0.0)$ \\
  & 3 & 912 & $63.8(\pm3.1)$  & $1.3 (\pm 0.2)$ & $2.4(\pm0.1)$ \\
    \midrule
\multirow{2}{*}{400} & 2 & 808 & $81.4(\pm2.7)$  & $1.1(\pm0.2)$  & $2.4(\pm0.2)$  \\
  & 3 & 1212 &  $189(\pm7.0)$& $1.9(\pm0.2)$  & $4.4(\pm0.3)$ \\
    \midrule
\multirow{2}{*}{500}   & 2 & 1008 & $184(\pm 8)$  & $1.3(\pm0.1)$  & $3.9(\pm0.4)$  \\
& 3 & 1512 &  $580(\pm44)$  &  $2.1(\pm 0.3)$  & $7.8(\pm0.4)$   \\
    \midrule
\multirow{2}{*}{550} & 2 & 1108 &  $259(\pm10)$  & $1.6(\pm0.2)$  & $5.1(\pm0.2)$  \\
  & 3 & 1662 &  $1121(\pm48)$ & $2.0(\pm0.1)$ & $10.8(\pm0.4)$ \\
\bottomrule
\end{tabular}}}
\end{table}

\begin{table}[h!]

    \centering
    \caption{\textbf{Scalability of Approximate Algorithms}: We report the runtime of \ourmethod~and \ourmethodlz~for several values of $q$ while fixing $r_2=550$. The rest of the setup is similar to the one in Table~\ref{table:runtimebenchmark}. We see that \ourmethodlz~can fit a path of solutions in around 5 minutes for problems with $q=25$. We report the average and the standard error (in the parenthesis).}
    \label{table:runtimebenchmark-approx2}
    {\scalebox{.99}{\begin{tabular}{cc|c|c}
    \toprule
  $q$ & $p$ &  \ourmethod & \ourmethodlz  \\
  & & Time [s] & Time [s] \\

\midrule
  5 & 2770 & $15.69(\pm2.33)$ & $32.75(\pm1.74)$ \\
    10 & 5540 & $39.28(\pm2.3)$& $76.10(\pm3.57)$ \\
  15 & 8310 & $67.92(\pm3.89)$ & $129.62(\pm7.26)$  \\
    20 & 11080& $95.24(\pm5.17)$ & $182.58(\pm8.50)$ \\
    25 & 13850& 155.84($\pm11.96$) & 289.85($\pm18.89$)   \\
\bottomrule
\end{tabular}}}
\end{table}

\begin{table}[h!]

    \centering
    \caption{ {\bf Quality of Approximate Solutions}: We report the objective value of the solution available from our approximate solver, as well as the exact solution obtained through solving the MIP~\eqref{clusteringproblem-mip}. Additionally, we also report the average objective gap between the exact and approximate solutions (the standard error for all cases was less than $0.001\%$ and hence is not reported). We set $\lambda=\lambda_0=0.05$ and $M=10$. The setup here is similar to Table~\ref{table:runtimebenchmark}. We report the average and the standard error (in the parenthesis). }
    \label{table:benchmarkquality}
    {\scalebox{.99}{\begin{tabular}{ccc|c|c|c}
    \toprule
  $r_2$ & $q$ & $p$ & Approximate Solution & Exact Solution & Gap \\
    \midrule
        \multirow{2}{*}{100} & 2 & 208 & 1.3940($\pm0.0021$) & 1.3932($\pm0.0021$) & $0.06\%$ \\
    & 3 & 312 & 1.4327($\pm0.0021$) & 1.4317($\pm0.0021$) & $0.07\%$ \\
        \midrule
        \multirow{2}{*}{150} & 2 & 308 & 1.2814($\pm0.0016)$  & 1.2799($\pm0.0015$) & $0.11\%$ \\
    & 3 & 462 & 1.2990($\pm0.0010$)  & 1.2981($\pm0.0010$) & $0.07\%$ \\
        \midrule
        \multirow{2}{*}{200} & 2 & 408 & 1.2171($\pm0.0012$)  & 1.2156($\pm0.0012$) & $0.13\%$\\
    & 3 & 612 & 1.2732($\pm0.0018$) & 1.2721($\pm0.0018$) & $0.09\%$ \\
        \midrule
        \multirow{2}{*}{250} & 2 & 508 & 1.1798($\pm0.0008$)  & 1.1789($\pm0.0008$) & $0.08\%$ \\
    & 3 &  762 & 1.2461($\pm0.0012$) &1.2451($\pm0.0012$) & $0.09\%$ \\
            \midrule
        \multirow{2}{*}{300} & 2 & 608 & 1.1556  ($\pm0.0008 $) & 1.1549($\pm0.0008 $) & $0.06\%$ \\
    & 3 & 912 & 1.2201 ($\pm0.0011 $) & 1.2190($\pm 0.0010$) & $0.08\%$\\
\bottomrule
\end{tabular}}}
\end{table}

\hfill

\begin{table}[]
    \centering
\caption{{\bf Exact Solver Runtimes}: Running time (in seconds) of our row generation-based exact solver for \ourmethodlz~from Section~\ref{sec:time-bench}. If the MIP solver is not able to achieve the optimality gap of $0.5\%$ in under~15 minutes, we stop the solver and report the final optimality gap (smaller optimality gap is better). We also report the average number of row generation iterations initiated. Note that $q$ is the total number of categorical predictors, with $q_s$-many of them having nonzero true regression coefficients. We note that SCOPE does not present exact solvers. Our experiments demonstrate that we are able to obtain global optimality certificates for problems with $p\approx 4500$.}
    \label{table:runtimebenchmark-exact}
\scalebox{.99}{\begin{tabular}{ccccc}
\toprule
  $q$ & $p$ & $q_s$ & Exact & Row Generation Iterations \\
      \midrule
  \multirow{2}{*}{20} & \multirow{2}{*}{360} & 2  & $91.5(\pm 13)$ & 1.16($\pm0.05$)\\
  & & 3  &  $691(\pm83)$ &  1.14($\pm0.05$) \\
      \midrule
  \multirow{2}{*}{50} & \multirow{2}{*}{900} & 2 &  $(1.3\%)$ &  1.08($\pm0.04$)  \\
  &  & 3   &     $(2.7\%)$   &  1.04($\pm0.03$)\\
      \midrule
  \multirow{2}{*}{100} & \multirow{2}{*}{1800} & 2 & $(2.8\%)$ & 1.0($\pm0.0$)\\
  & & 3   & $(4.5\%)$  & 1.0($\pm0.0$) \\
    \midrule
  \multirow{2}{*}{150} & \multirow{2}{*}{2700} & 2 & $(3.2\%)$ &  1.0($\pm0.0$)  \\
  &  & 3   &  $(5.8\%)$   &  1.0($\pm0.0$) \\
  \midrule
  \multirow{2}{*}{200} & \multirow{2}{*}{3600} & 2  & $(3.4\%)$  & 1.0($\pm0.0$) \\
  & & 3   &   $(5.7\%)$ & 1.0($\pm0.0$) \\
\midrule
\multirow{2}{*}{250} & \multirow{2}{*}{4500} & 2  & $(3.9\%)$  & 1.0($\pm0.0$)  \\
    & & 3  & $(6.1\%)$ &  1.0($\pm0.0$)   \\
\bottomrule
\end{tabular}}

    \end{table}

\subsection{Benchmarking DpSegPen-$L_0$}\label{supp:dpseg-bench}

Next, we investigate the efficiency of DpSegPen-$L_0$, our exact solver for the case with one categorical predictor (that is, Problem~\eqref{one-dim-jumps}). We follow a setup similar to Section~\ref{sec:time-bench} and Table~\ref{table:runtimebenchmark}, but we set $q=q_s=1$. We also consider both $\lambda_0=0$ (corresponding to the work of~\citet{johnson2013dynamic}) and $\lambda_0>0$ (studied in Appendix~\ref{app:dp}). The results for this case are shown in Table~\ref{table:runtimebenchmark-oneidm}. In the table, DpSegPen-$L_0$ corresponds to the case $\lambda_0>0$, while for DpSegPen we set $\lambda_0=0$. Similar to Table~\ref{table:runtimebenchmark}, we report the overall runtime of fitting a path of solutions for 100 hyperparameter values (over the same grid).

\begin{table}[h!]
    \centering
    \caption{{\bf DpSegPen-$L_0$ Runtimes}: Running time (in seconds) of DpSegPen-$L_0$, our exact solver for the case with a single categorical predictor (Problem~\eqref{one-dim-jumps}). In the table, DpSegPen corresponds to the case of $\lambda_0=0$. We note that both methods studied in this table are exact and result in globally optimal solutions. We report the average and the standard error (in the parenthesis) over 50 repetitions.}
    \label{table:runtimebenchmark-oneidm}
    {\scalebox{.99}{\begin{tabular}{ccc}
    \toprule
 $r_2$ & DpSegPen-$L_0$ & DpSegPen  \\
 \midrule
 200  &  $0.932(\pm0.053)$& $0.186(\pm0.013)$ \\
 300 & $1.240(\pm0.075)$ & $0.266(\pm0.022)$  \\
 400 & $1.478(\pm0.084)$ & $0.340(\pm0.029)$ \\
 500 & $1.738(\pm0.108)$ & $0.411(\pm0.034)$ \\
600 &  $1.650(\pm0.048)$ & $0.408(\pm0.019)$ \\
\bottomrule
\end{tabular}}}
\end{table}

\subsection{Additional Synthetic Data Experiments}\label{supp:synthetic}

In this section, we present additional synthetic data numerical experiments under various settings.

First, we consider the same general setup as in Figure~\ref{fig:set1-n} but we vary~$\sigma$. First, we set $\sigma=1.5$ (leading to $\text{SNR}=2.20$) and plot the results in Figure~\ref{fig:setsupp1-n}. In these experiments, Elastic Net resulted in 147 clusters on average. We also consider~$\sigma=2.5$ (average $\text{SNR}=0.79$) and present the results in Figure~\ref{fig:setsupp2-n}. Under this setup, Elastic Net produces 145 clusters, on average. We see that \ourmethodlz~continues to have the best performance overall. Additionally, when $\sigma=2$, \ourmethod~outperforms or performs comparably to SCOPE.

For the next experiment, under the same setup, we fix $n=300$ and vary the value of $\sigma$ to study the effect of noise on different methods. We display the results in Figure~\ref{fig:set-sigma}. We observe that \ourmethodlz~and \ourmethod~perform better than SCOPE, especially in low SNR regimes, with SNR being as low as 0.65. In particular, \ourmethodlz~seems to have the best performance in terms of $R^2$ and purity overall.

\begin{figure}[t!]
\setlength{\tabcolsep}{1pt}
     \centering
    \scalebox{1.1}{\begin{tabular}{ccc}
      Test $R^2$ & Purity & \# of Clusters \\
       \includegraphics[width=0.28\linewidth]{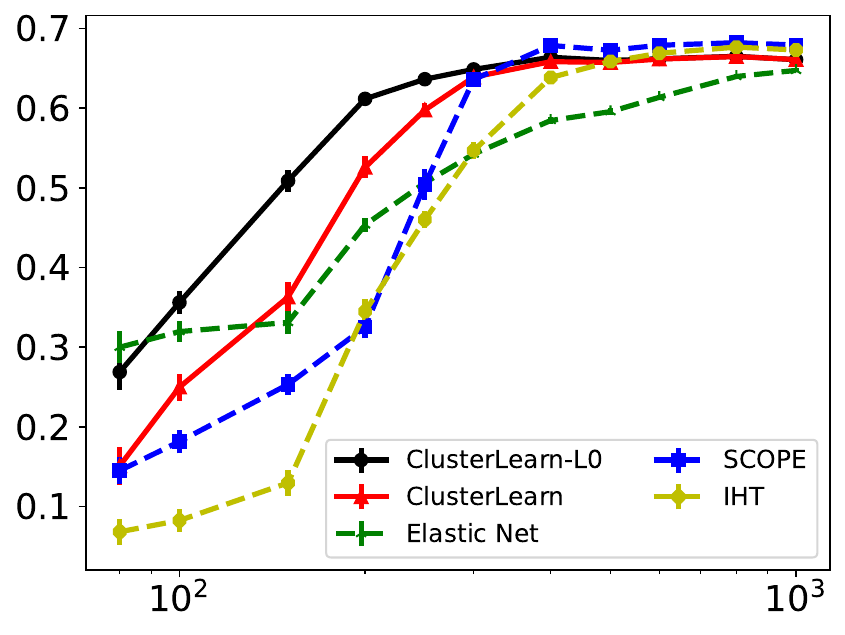}  &  \includegraphics[width=0.28\linewidth]{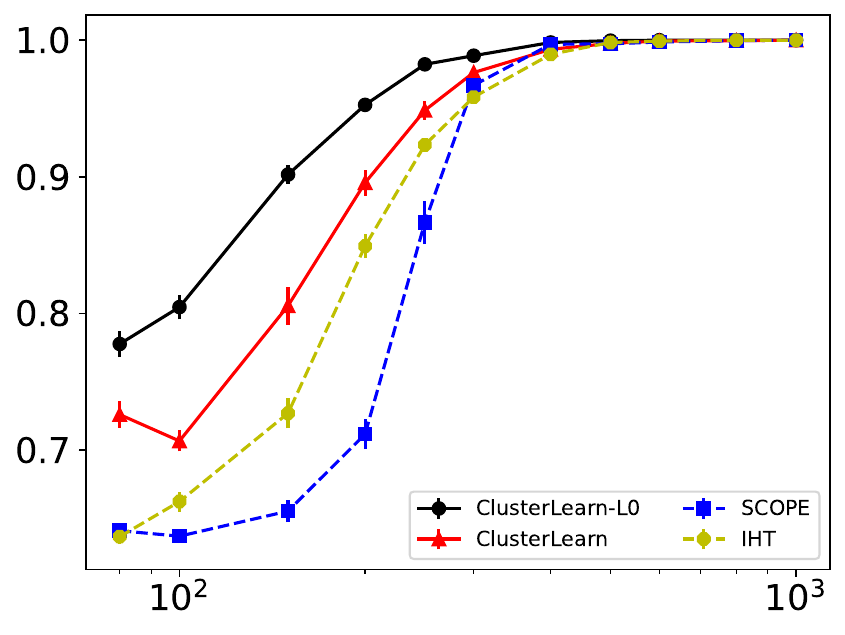} & 
       \includegraphics[width=0.28\linewidth]{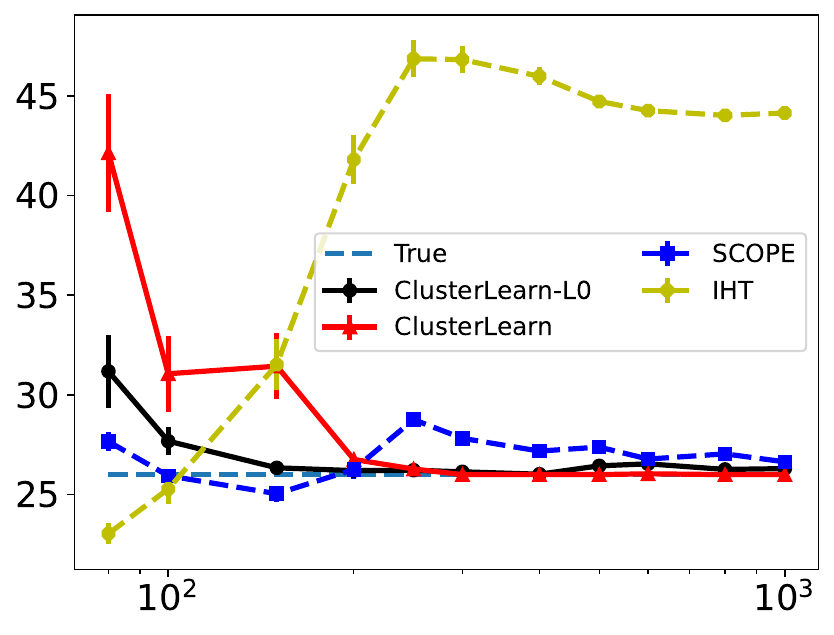} \\
       $n$ & $n$ & $n$
    \end{tabular}}
    \captionof{figure}{Experiments from Appendix~\ref{supp:synthetic} with $r_1=4,r_2=12, q=20, q_s=3$ and $\rho=0.2,\sigma=1.5$ with varying $n$. We do not plot purity and the number of clusters for Elastic Net as it produces a large number of clusters. The vertical bars at each point indicate the corresponding standard errors.}
    \label{fig:setsupp1-n}
\end{figure}

\begin{figure}[t!]
\setlength{\tabcolsep}{1pt}
     \centering
    \scalebox{1.1}{\begin{tabular}{ccc}
      Test $R^2$ & Purity & \# of Clusters \\
       \includegraphics[width=0.28\linewidth]{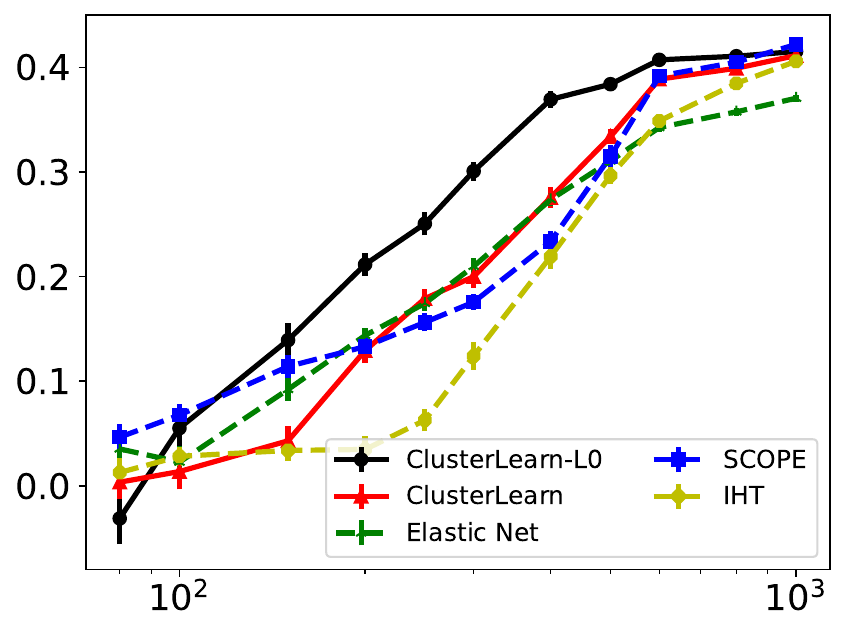}  &  \includegraphics[width=0.28\linewidth]{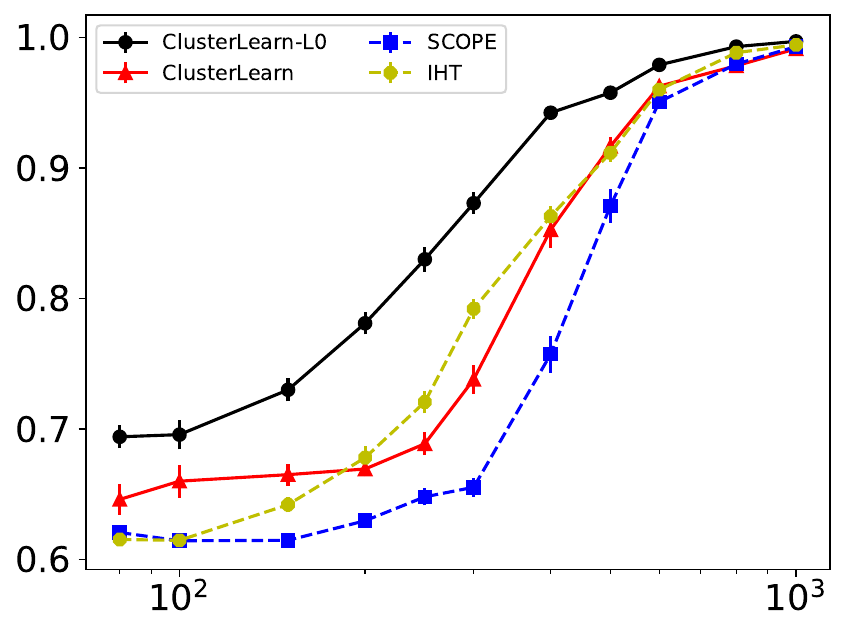} & 
       \includegraphics[width=0.28\linewidth]{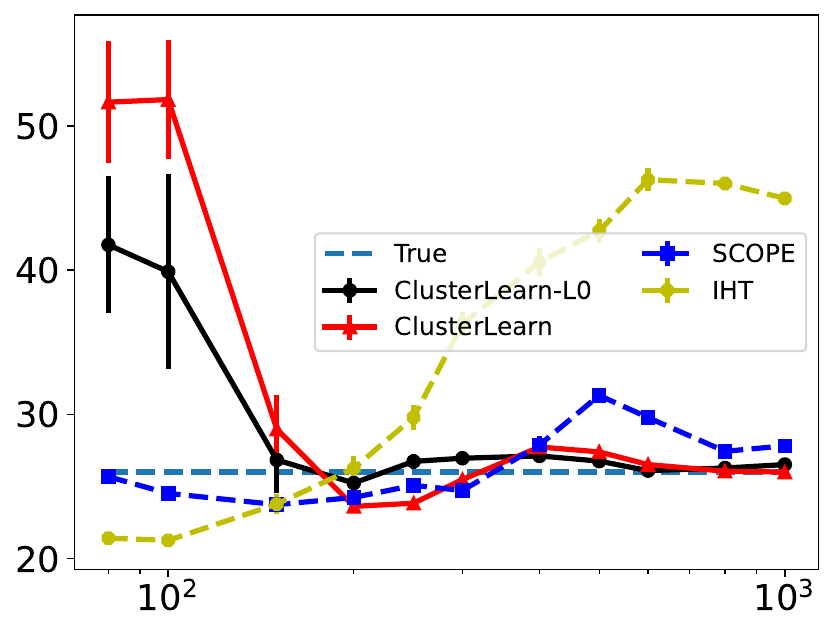} \\
       $n$ & $n$ & $n$
    \end{tabular}}
    \captionof{figure}{ Experiments from Appendix~\ref{supp:synthetic} with $r_1=4,r_2=12, q=20, q_s=3$ and $\rho=0.2,\sigma=2.5$. We do not plot the purity and the number of clusters for Elastic Net as it produces a large number of clusters. The vertical bars at each point indicate the corresponding standard errors.}
    \label{fig:setsupp2-n}
\end{figure}

\begin{figure}[t!]
\setlength{\tabcolsep}{1pt}
     \centering
    \scalebox{1.1}{\begin{tabular}{ccc}
      Test $R^2$ & Purity & \# of Clusters \\
       \includegraphics[width=0.28\linewidth]{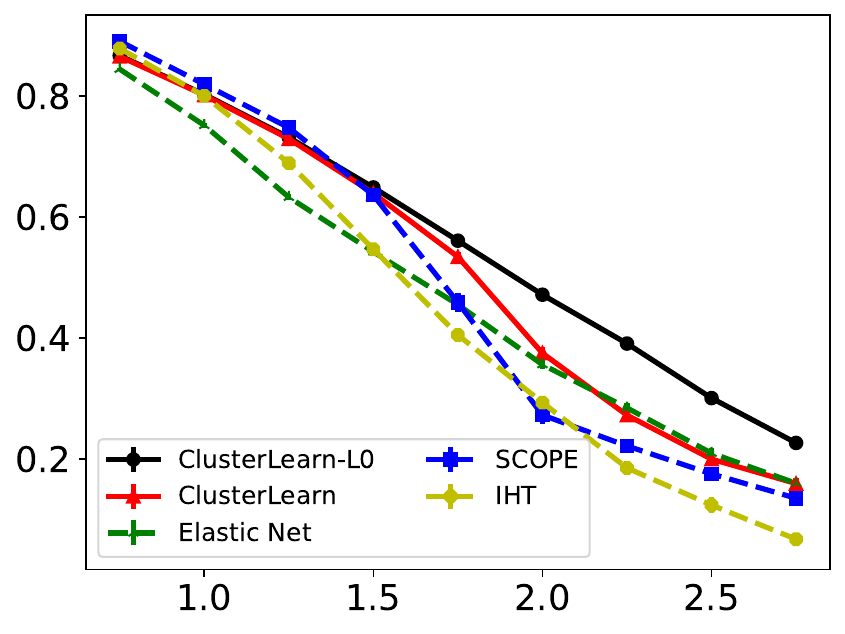}  &  \includegraphics[width=0.28\linewidth]{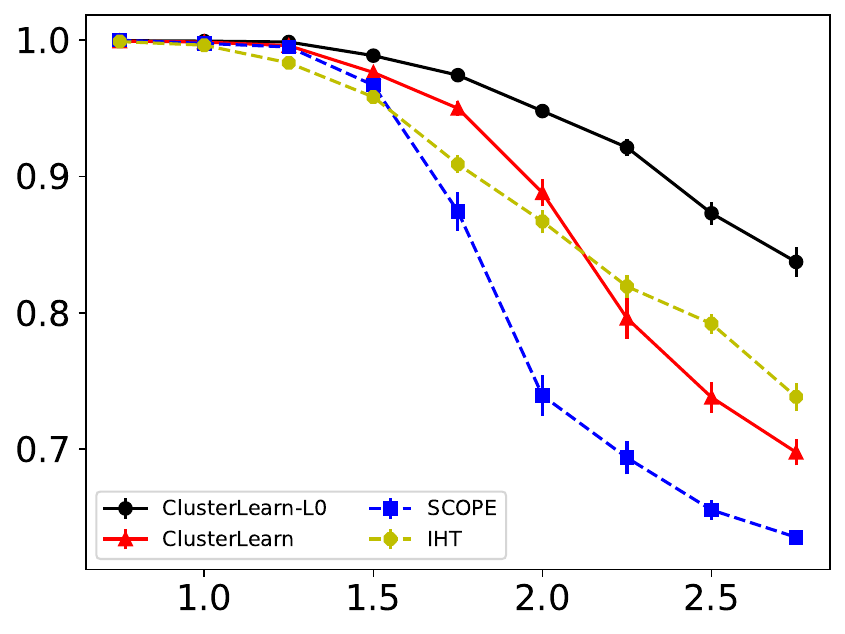} & 
       \includegraphics[width=0.28\linewidth]{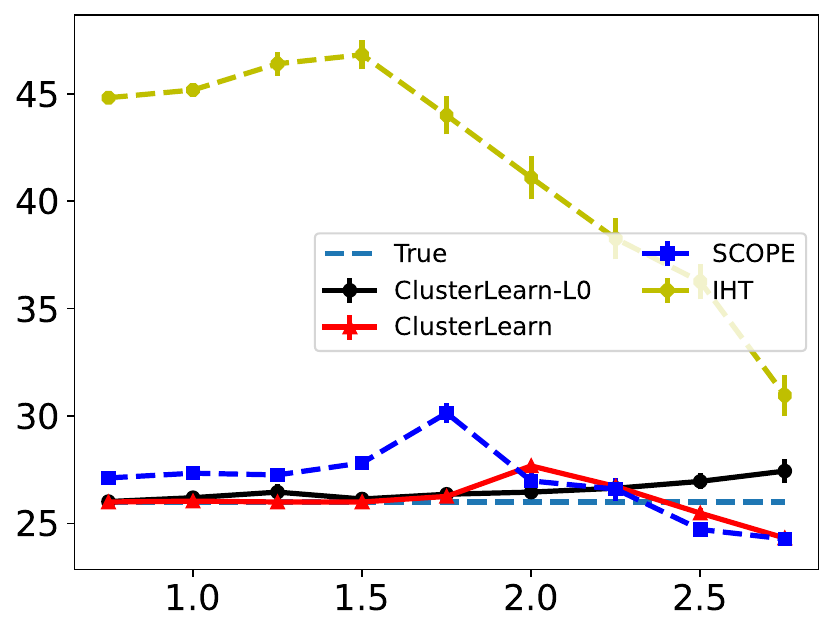} \\
       $\sigma$ & $\sigma$ & $\sigma$
    \end{tabular}}
    \captionof{figure}{ Experiments from Appendix~\ref{supp:synthetic} with $r_1=4,r_2=12, q=20, q_s=3$ and $\rho=0.2,n=300$. We do not plot the purity and the number of clusters for Elastic Net as it produces a large number of clusters. The vertical bars at each point indicate the corresponding standard errors.}
    \label{fig:set-sigma}
\end{figure}

Next, we consider another experimental setup where the true $\B\beta^*$ is given by:
\begin{equation}\label{eq:beta_star_setting_supp}
    \B\beta^{*}_{\cI_j} = \left\{
    \begin{array}{ll}
        (-2,-2,0,0,0) & \mbox{if } j \in \{1,\cdots,q_s\} \vspace{0.2cm}\\
        (0,0,0,0,0) & \mbox{otherwise.}
    \end{array}
\right.
\end{equation}
We set $q=50$ and $q_s=5$. This simulates a setting where there are several categorical predictors, but each with only a few levels. We set $\sigma=2$ leading to average $\text{SNR}=5.38$. We plot the results for this case in Figure~\ref{fig:setsupp3-n}. Elastic Net here results in 125 clusters on average. We see that overall \ourmethod~seems to have the best performance, with the highest $R^2$ and purity. Particularly, we see that \ourmethod~results in a smaller number of clusters than SCOPE when $n\geq 100$, while having better or similar purity and $R^2$. Moreover, \ourmethodlz~does not perform well here, possibly because $\B\beta^*_{\cI_j}$ for $j=1,\cdots,q_s$ does not have a large zero cluster.

\begin{figure}[t!]
\setlength{\tabcolsep}{1pt}
     \centering
    \scalebox{1.1}{\begin{tabular}{ccc}
      Test $R^2$ & Purity & \# of Clusters \\
       \includegraphics[width=0.28\linewidth]{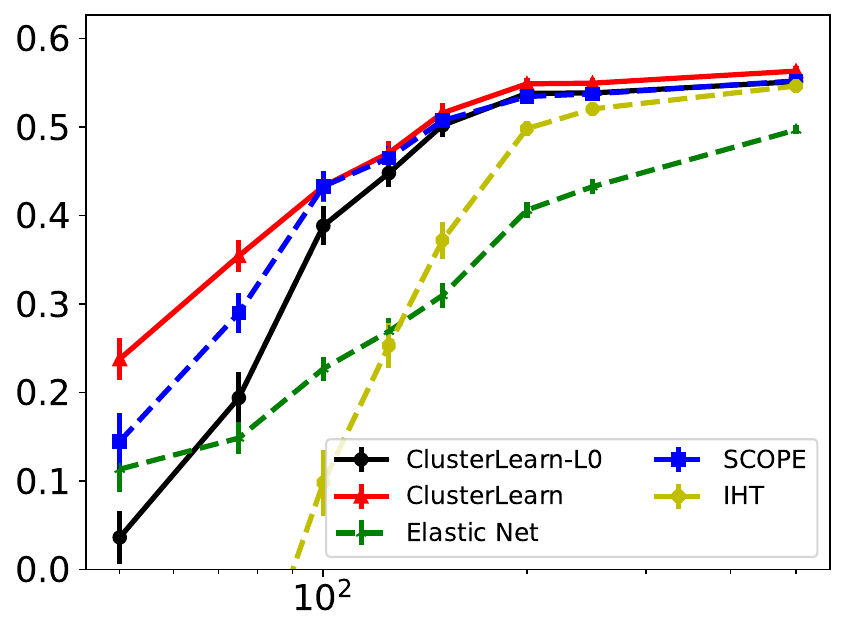}  &  \includegraphics[width=0.28\linewidth]{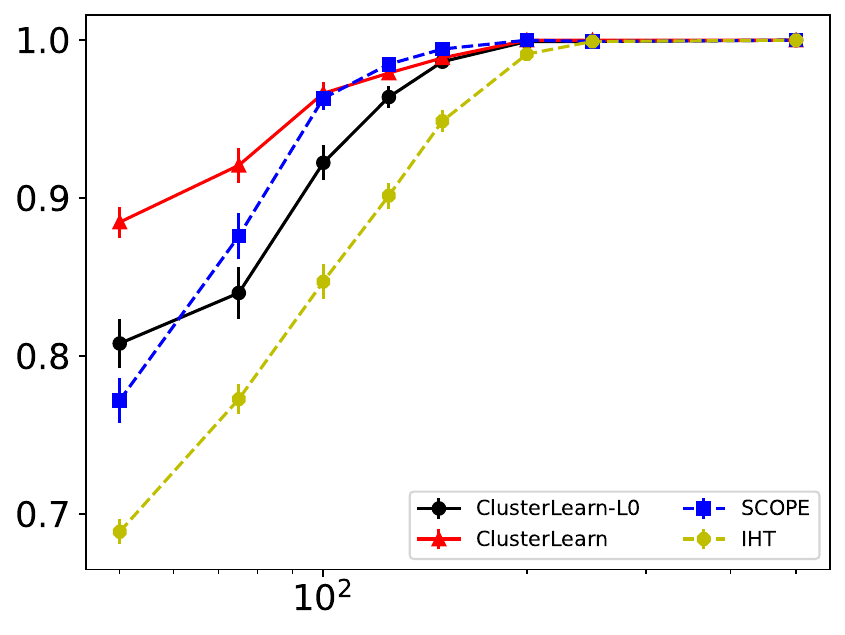} & 
       \includegraphics[width=0.28\linewidth]{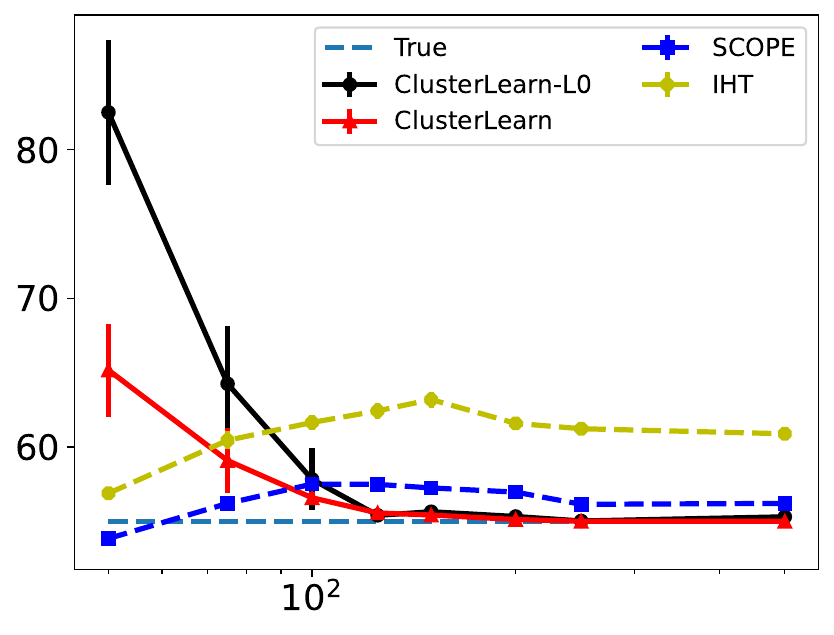} \\
       $n$ & $n$ & $n$
    \end{tabular}}
    \captionof{figure}{ Experiments from Appendix~\ref{supp:synthetic}  with the true $\B\beta^*$ given by \eqref{eq:beta_star_setting_supp} and  $q=50, q_s=5,\sigma=2,\rho=0.2$. We do not plot the purity and the number of clusters for Elastic Net as it produces a large number of clusters. The vertical bars at each point indicate the corresponding standard errors.}
    \label{fig:setsupp3-n}
\end{figure}

Finally, we consider a setup where no clustering is present among nonzero regression coefficients. In particular, under the same setup as before, we choose the $\B\beta^*$ to be:
\begin{equation}\label{eq:beta_star_setting_rebuttal}
    \B\beta^{*}_{\mathcal{I}_j} = \left\{
    \begin{array}{ll}
        (1,2,3,0,0) & \mbox{if } j \in \{1,\cdots,q_s\} \vspace{0.2cm}\\
        (0,0,0,0,0) & \mbox{otherwise.}
    \end{array}
\right.
\end{equation}
We set $q=50$ and $q_s=5$ and $\sigma=3$, corresponding to an average $\text{SNR}$ of 4.86. We present the prediction performance of different methods in Figure~\ref{fig:setrebuttal-n}. We do not plot the clustering performance here as the true $\B\beta^*$ is not clustered. We see that all methods perform more or less similar, with \ourmethod~having a slightly higher accuracy.

\begin{figure}[t!]
\setlength{\tabcolsep}{1pt}
     \centering
    \scalebox{1.1}{\begin{tabular}{c}
      Test $R^2$ \\
       \includegraphics[width=0.28\linewidth]{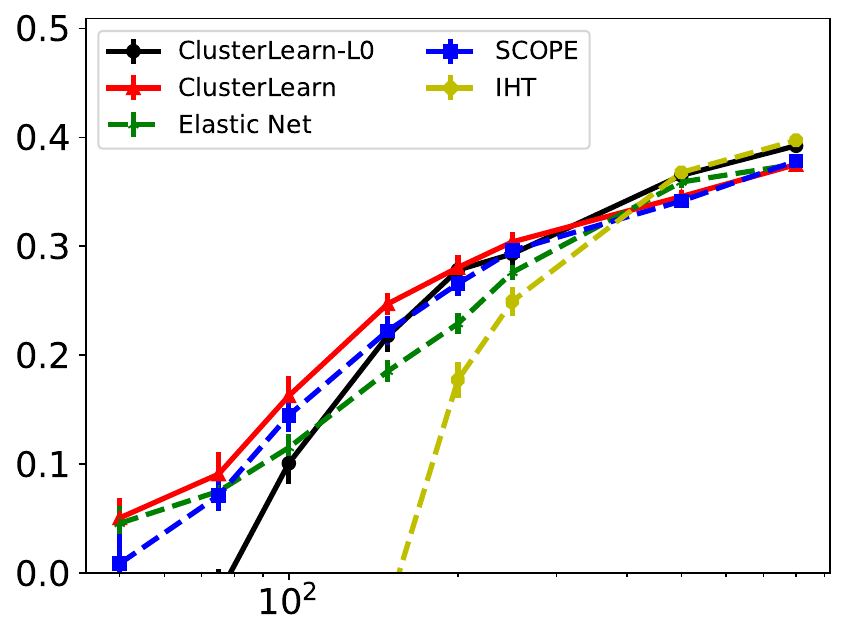} \\
       $n$ 
    \end{tabular}}
    \captionof{figure}{ Experiments from Appendix~\ref{supp:synthetic}  with the true $\B\beta^*$ given by \eqref{eq:beta_star_setting_rebuttal} and  $q=50, q_s=5,\sigma=3,\rho=0.2$. We do not plot the purity and the number of clusters as $\B\beta^*$ is not clustered. The vertical bars at each point indicate the corresponding standard errors.}
    \label{fig:setrebuttal-n}
\end{figure}

\subsection{Employee Access Dataset}\label{sec:amazon}
Nest, we explore the Amazon Employee Access Challenge \citep{amazon-employee-access-challenge} dataset, where the task is to predict whether an employee should be granted access to a specific resource. We use 3794 observations and 9 categorical predictors. Due to the large number of categorical levels, the model dimension is very high ($p=5659$). We use $n=1000$ observations as a training set, 1000 as a validation set, and the rest as a test set. We also drop observations with levels that do not appear in the training set. We compare our methods with SCOPE, Elastic Net and Lasso. For SCOPE, we use the concavity parameter of 100 that is recommended by~\cite{scope}. Table~\ref{tab:ama} reports the test classification accuracy, the number of regression coefficient levels, and the runtime. We observe that Elastic Net and \ourmethodlz~have the highest accuracy, though Elastic Net leads to a significantly higher number of levels. Moreover, \ourmethodlz~has fewer levels than \ourmethod~but more than SCOPE. In terms of runtime, both \ourmethodlz~and \ourmethod~appear to be faster than SCOPE.

\begin{table}[!h]
    \centering
        \caption{Mean performance of different algorithms on the employee access dataset in Appendix~\ref{sec:amazon}, averaged across 100 runs.  We report the average and the standard error (in parenthesis).}
    \label{tab:ama}
    \begin{tabular}{cccc}
\toprule
  Method & Accuracy & Number of levels & Runtime (seconds)\\
\midrule
Lasso & $0.743 (\pm0.004)$  & $545.0 (\pm16.9)$ & $\mathbf{0.2 (\pm0.0)}$ \\
 Elastic Net & $\mathbf{0.786(\pm0.002)}$ & $2460.7(\pm3.1)$ & $\mathbf{0.2(\pm0.0)}$ \\
SCOPE-100 & $0.768 (\pm0.003)$ & $\mathbf{19.1 (\pm1.0)}$ & $359.5 (\pm40.9)$ \\
\ourmethod & $0.779 (\pm0.003)$ & $153.2 (\pm12.2)$ & $53.1 (\pm0.3)$ \\
\ourmethodlz & $\mathbf{0.784 (\pm0.003)}$  & $115.5 (\pm11.2)$ & $78.2 (\pm0.4)$ \\
\bottomrule
\end{tabular}

\end{table}

\subsection{Solar Flare Real Dataset}\label{supp:flare}
We conduct further experiments on a real dataset. Specifically, we consider the solar flare dataset~\citep{solar_flare_89} where the goal is to predict the number of solar flares in a 24 hour interval, given several features (a regression task). This dataset includes 10 categorical features with total number of levels $p=32$. We use the count of common solar flares as the response variable, and use $n=500$ samples for training and validation, while using the rest of samples as the test set. The results for this experiment are presented in Table~\ref{tab:flare}. We see that \ourmethodlz~and \ourmethod~ provide better prediction accuracy compared to SCOPE, while having a lower number of clusters compared to Elastic Net and IHT.

\begin{table}[]
    \centering
        \caption{ Mean performance of different algorithms on the solar flare dataset in Appendix~\ref{supp:flare}, averaged across 100 runs. We report the average and the standard error (in parenthesis).}
    \label{tab:flare}
    \begin{tabular}{cccc}
\toprule
  Method & Test $R^2$ & Number of levels & Runtime (seconds) \\
\midrule
		
IHT &  $0.106(\pm0.005)$ & $19.13(\pm0.43)$   & $0.308(\pm0.001)$ \\
Elastic Net & ${0.114(\pm 0.003)}$ &    $28.76(\pm0.34)$  & $0.445(\pm0.001)$   \\
SCOPE & $0.094(\pm0.005)$ &	 $15.02(\pm0.38)$   & $0.530(\pm0.005)$ \\
\ourmethod & $0.104(\pm0.005)$ & $17.09(\pm0.41)$	   & $1.692(\pm0.004)$ \\
\ourmethodlz & $0.108(\pm0.004)$  & $18.20(\pm0.47)$	   & $6.631(\pm6.034)$   \\
\bottomrule
\end{tabular}

\end{table}

\section{Approximate Solver for the Logistic Loss}
Here, we display the algorithm that extends our BCD approximate solver to the problems with the logistic loss. See Section~\ref{sec:logistic} for more details.

\begin{algorithm}
\caption{\ourmethodlz~block CD for Problem~\eqref{logstic-reg}}\label{alg:bcd_classification}
\begin{algorithmic}
    \Require Data matrix $\B X$, response vector $\B{y}$, $\lambda,\lambda_0\geq 0$, initial $\B\beta^0 \in \R^p$
\State $\B\beta \gets \B\beta^0$\;
\While{not converged}
  \For{$j = 1$ to $q$}
  \State $\tilde{\B\beta} \gets \B\beta$ 
    \State $\B\beta_{\cI_{j}} \gets \argmin_{\tilde \beta_i:i\in\cI_{j}} {P}(\tilde{\B\beta},\B\beta,\alpha,\alpha)$\ \Comment{A univariate problem like~\eqref{blockcd-2}}
  \EndFor
    \State $\tilde{\B\beta} \gets \B\beta$ 
  \State $\B\beta_{p-N+1:p} \gets \argmin_{\tilde \beta_{p-N+1:p}} {P}(\tilde{\B\beta},\B\beta,\alpha,\alpha)$\ \Comment{Update continuous features}
  \State $\alpha \gets \argmin_{\tilde \alpha} {P}({\B\beta},\B\beta,\tilde\alpha,\alpha)$\ \Comment{Update intercept}
  \EndWhile
\end{algorithmic}

\end{algorithm}

\section{Additional Details of the Illustrative Example}\label{supp:illustrative}
In this section, we discuss the interpretation of the regression coefficients presented in Figure~\ref{fig:intro}. In particular, we investigate which categorical levels have the largest and smallest (most negative) regression coefficients. We focus on the case of $\lambda_0=50$ and $\lambda=0.5$ and report the results in Table~\ref{table:illustrative-variables}. 

We see that in terms of the weekday, Wednesday and Friday have the largest regression coefficients, while Saturday and Sunday have the smallest ones. This suggests there are more bike rentals on working days in coomparison to the weekend (for example, bikes might be mostly rented by commuters). In terms of the hours, we see that mornings and evenings have the highest bike rental counts, while late night hours have the lowest number of rentals. In fact, the cluster of zero regression coefficients for the `hour' predictor corresponds to the hours of 0 to 4. On the other hand, the singleton cluster of hour 8 has the largest coefficient, followed by the cluster of hours 17 and 18. 
As mornings and evenings coincide with the commuting hours, this further suggests that bike rentals might be dominated by the commuters.

\begin{table}[]
    \centering
\caption{The categories with largest and smallest (most negative) regression coefficients for the example in Figure~\ref{fig:intro}.}
    \label{table:illustrative-variables}
\scalebox{.99}{\begin{tabular}{ccc}
\toprule
Predictor & Largest & Smallest \\
      \midrule
      Weekday & Wednesday and Friday & Saturday and Sunday \\
      Hour & 18, 8 & 0,1\\
\bottomrule
\end{tabular}}
    \end{table}

\end{document}